\documentclass{article}

\usepackage{amsmath,amssymb,amsthm}
\usepackage{fullpage}
\usepackage{xcolor,xspace}
\usepackage{graphicx}
\usepackage{caption}
\usepackage{booktabs}
\usepackage{microtype}
\usepackage{enumerate}
\usepackage{algorithm}
\usepackage[
    indLines=true,
    noEnd=true,
    rightComments=true,
    italicComments=true,
]{algpseudocodex}
\algrenewcommand\algorithmicrequire{\textbf{Input:}}
\algrenewcommand\algorithmicensure{\textbf{Output:}}

\usepackage{thmtools}
\usepackage{thm-restate}

\declaretheorem[name=Theorem]{theorem}
\declaretheorem[name=Lemma,sibling=theorem]{lemma}
\declaretheorem[name=Corollary,sibling=theorem]{corollary}

\declaretheorem[name=Definition,sibling=theorem,style=definition]{definition}

\usepackage[
    colorlinks=true,
    linkcolor=blue!62!black,
    citecolor=green!48!black,
    urlcolor=blue!70!black,
    linktoc=page
]{hyperref}
\usepackage[nameinlink,noabbrev]{cleveref}

\crefname{theorem}{Theorem}{Theorems}
\Crefname{theorem}{Theorem}{Theorems}
\crefname{lemma}{Lemma}{Lemmas}
\Crefname{lemma}{Lemma}{Lemmas}
\crefname{corollary}{Corollary}{Corollaries}
\Crefname{corollary}{Corollary}{Corollaries}
\crefname{observation}{Observation}{Observations}
\Crefname{observation}{Observation}{Observations}
\crefname{proposition}{Proposition}{Propositions}
\Crefname{proposition}{Proposition}{Propositions}
\crefname{claim}{Claim}{Claims}
\Crefname{claim}{Claim}{Claims}
\crefname{conjecture}{Conjecture}{Conjectures}
\Crefname{conjecture}{Conjecture}{Conjectures}
\crefname{assumption}{Assumption}{Assumptions}
\Crefname{assumption}{Assumption}{Assumptions}
\crefname{definition}{Definition}{Definitions}
\Crefname{definition}{Definition}{Definitions}
\crefname{remark}{Remark}{Remarks}
\Crefname{remark}{Remark}{Remarks}
\crefname{algorithm}{Algorithm}{Algorithms}
\Crefname{algorithm}{Algorithm}{Algorithms}
\crefname{section}{Section}{Sections}
\Crefname{section}{Section}{Sections}
\crefname{appendix}{Appendix}{Appendices}
\Crefname{appendix}{Appendix}{Appendices}

\tikzset{
  algpxIndentLine/.style={draw=black!100,very thin}
}

\DeclareMathOperator{\End}{End}
\DeclareMathOperator{\mincut}{mincut}

\title{Faster Minimum \(k\)-Cut I: Simple and Sparse Weighted Graphs}

\author{
   Jason Li\thanks{Carnegie Mellon University. email: jmli@cs.cmu.edu}
   \and
   Trevor Vaughn\thanks{Carnegie Mellon University. email: tnvaughn@cmu.edu}
}

\begin{document}

\maketitle

\begin{abstract}
The minimum $k$-cut problem asks for the fewest edges whose removal leaves an input graph with at least $k$ connected components. Previously, the best algorithm for simple graphs ran in $O_k(n^{(1-\varepsilon)k+O(1)})$ time~\cite{HL22}, showing that the \(n^k\) barrier can be broken up to a polynomial overhead.

We give the first $\widetilde O_k(n^{ck})$-time algorithm for Minimum $k$-Cut on simple graphs for an absolute constant $c<1$. More precisely, the running times are $\widetilde O(n^2)$ for $k=3$, $\widetilde O(n^{55/19})$ for $k=4$, and $\widetilde O(n^{4.112007})$ for $k=5$; for every $k\ge6$, the running time is
\[
k^{O(k^2)}n^{1+(6k-6)\frac{k-1.749614}{7k-10}}(\log n)^{O(k^2)},
\]
whose exponent is $\frac67k-0.132\ldots+O(1/k)$.

The algorithm combines three ingredients. First, for weighted Minimum $k$-Cut we give a randomized
\[
k^{O(k^2)}n^{k-2}(m+n)\log^3(n)
\]
-time algorithm: it perturbs the edge weights so that any minimum $k$-cut has a side with boundary strictly smaller than average. These then cut few edges of some tree in a logarithmic-size sample from a tree packing with high probability. After we enumerate them, we recursively compute $(k-1)$-cuts to complete them to the $k$-cuts of which they were a part.
A variant of the perturbation and processing the entire packing support give a deterministic $k^{O(k^2)}n^{k+O(1)}$-time variant. Second, for cut size $s$, we give an improved FPT algorithm using a near-linear-time construction of an $(O(s\log^2 n\log\log n),s)$ edge-unbreakable tree decomposition with $O(s\log^2 n\log\log n)$ adhesion; this also gives a near-linear-time approximation algorithm for Minimum $k$-Cut. Third, we refine the border/island framework of~\cite{HL22}, using rectangular matrix multiplication to recover singleton islands and balancing it against the improved FPT algorithm.
\end{abstract}

\newpage

\tableofcontents

\newpage

\section{Introduction}

The Minimum $k$-Cut problem asks for the minimum number of edges whose removal leaves at least $k$ connected components. We write $\lambda_k(G)$ for this optimum value. The case $k=2$ is the global minimum cut problem, which admits near-linear-time algorithms \cite{Kar00}. For fixed $k$, the first polynomial-time algorithm was due to Goldschmidt and Hochbaum~\cite{GH94}, with running time $n^{O(k^2)}$. Karger and Stein~\cite{KS96} later gave a randomized contraction algorithm running in $\widetilde O(n^{2k-2})$ time, and Thorup~\cite{Tho08} gave a deterministic tree-packing algorithm which runs in $\widetilde O(n^{2k})$.

A sequence of recent works on the structure and enumeration of near-minimum $k$-cuts \cite{GLL18,GLL19,Li19} led to the algorithm of \cite{GHLL21}. They showed that all $k$-cuts of value at most $\beta\lambda_k$ can be enumerated in time $n^{\beta k}(\log n)^{O(\beta k^2)}$ with high probability. In particular, setting $\beta=1$ gives an algorithm for Minimum $k$-Cut in time $n^k(\log n)^{O(k^2)}$. For weighted graphs, a lower bound of \(\Omega(n^{k-1-o(1)})\) is known under standard clique hardness conjectures \cite{GHLL21,HL22}.

For simple unweighted graphs, however, the $n^k$ barrier is not the end of the story. He and Li~\cite{HL22} showed that one can break this barrier by combining a border-preserving version of the Kawarabayashi--Thorup sparsification framework~\cite{KT18,Li19} with matrix-multiplication methods for recovering singleton components. Their framework separates the simple unweighted case from the weighted case: after a border-preserving sparsification, the remaining difficulty is split between finding a low-value residual border and recovering the singleton ``islands'' of the cut. The latter task has a clique-like complexity resulting from fast matrix multiplication.

Our first contribution is a sparse-graph algorithm for weighted Min \(k\)-Cut, running in time \(\widetilde O_k(n^{k-2} (m + n))\). For sparse graphs, this matches the \(\Omega(n^{k-1-o(1)})\) lower bound under standard clique hardness conjectures \cite{GHLL21,HL22}. We remark that a companion paper~\cite{Vau26} gives a deterministic \(\widetilde O_k(n^{k-1})\)-time algorithm, but the speedup is significantly more technical and is unnecessary for our main result on simple graphs.

The algorithm uses a recursive strategy. After a small random perturbation of the capacities, some side of an optimum \(k\)-cut has boundary strictly below \(2\lambda_k/k\), which we call a strict-light side. By a theorem of \cite{GHLL21}, the number of \(2\)-cuts with boundary strictly below \(2\lambda_k/k\) is at most \(k^{O(k)}n\). We remark that the statement is false if strictness is removed: an unweighted cycle has \(\binom{n}{2}\) many \(2\)-cuts of size exactly \(2\lambda_k/k=2\). Observe that any optimal \(k\)-cut of an unweighted cycle also has no strict-light side, which explains why perturbation is necessary to guarantee strict-light sides.

We compute an approximate solution to the dual \(k\)-cut LP and sample
\(O(k^3\log n)\) spanning trees from its packing with probability proportional to weight.  Every
strict-light side is \((k-1)\)-respected by a sampled tree with
probability \(\Omega(1/k)\).  For each sampled tree, we enumerate the
rooted cuts crossing at most \(k-1\) tree edges and retain only the
\(k^{O(k)}n\) lightest distinct cuts.  The strict-light counting bound
guarantees that no strict-light side is discarded.  We then evaluate
the retained candidates by recursive Minimum \((k-1)\)-Cut calls on
their two sides.

\begin{theorem}[Informal]
For every \(k\ge2\), weighted Min \(k\)-Cut on an \(n\)-vertex, \(m\)-edge weighted graph can be solved with high probability in time
\[
    k^{O(k^2)}n^{k-2}(m+n)\log^3(n)
\]
\end{theorem}

This weighted algorithm is used later on the contracted multigraph produced by border-preserving sparsification. In the no-island case, the optimum simple graph cut is represented as an ordinary weighted \(k\)-cut of the contracted graph, so the weighted algorithm replaces enumeration of all \(k\)-cut borders.

A deterministic version
of the algorithm replaces the random perturbation by a lexicographic perturbation based on the incidence vectors of the cuts and processes every tree in the support of the
same approximate packing instead of sampling from it.  The packing
averaging argument guarantees that at least one support tree respects the
required strict-light side.  This gives the following polynomial-overhead
deterministic result.

\begin{theorem}[Informal]
For every \(k\ge2\), weighted Min \(k\)-Cut on an \(n\)-vertex, \(m\)-edge weighted graph can be solved deterministically in time
\[
    k^{O(k^2)} n^{k+O(1)}.
\]
\end{theorem}

The next ingredient is a faster exact algorithm when the cut value \(\lambda_k\) is small. The approach of Lokshtanov, Saurabh, and Surianarayanan~\cite{LSS22} gives such an FPT algorithm using edge-unbreakable decompositions. Informally, a bag is \((q,s)\)-edge-unbreakable if every cut of size at most \(s\) leaves at most \(q\) bag vertices on one side or the other. This property is useful for dynamic programming: in every bag and every small cut, all but one part are forced to be small. However, the edge-unbreakable decomposition of \cite{LSS22} has unbreakability parameter \(((s+1)^5,s)\) and polynomial construction time and depth. We replace it by a near-linear-time construction with nearly linear dependence on \(s\).

\begin{theorem}[Informal]
For every \(s\), there is a randomized algorithm which, with high probability, computes a compact tree decomposition of \(G\) whose bags are
\[
    (O(s\log^2 n\log\log n),s)\text{-edge-unbreakable},
\]
whose adhesions have size
\[
    O(s\log^2 n\log\log n),
\]
and whose depth is \(O(\log n)\). The running time is $\widetilde O(m+sn)$.
\end{theorem}

\begin{table}[htpb]
\centering
\begin{tabular}{lcccccc}
\toprule
\textbf{Reference} & \textbf{Time} & \textbf{Unbreakability} & \textbf{Adhesion} & \textbf{Depth} & \textbf{Subtree?} & \textbf{Vertex or edge} \\
\midrule
\cite{CLP+19} & $2^{O(s^2)}n^2m$ & $(2^{O(s)}, s)$ & $2^{O(s)}$ & $n$ & Yes & Vertex \\
\cite{CKL+21} & $2^{O(s \log s)}n^{O(1)}$ & $(i, i) \; \forall i \le s$ & $s$ & $n$ & No & Vertex \\
\cite{LSS22} & $n^{O(1)}$ & $((s+1)^5, s)$ & $s$ & $n$ & No & Edge \\
\cite{ILSS23} & $2^{O(s \log s)}n^{O(1)}$ & $(9s, s)$ & $8s$ & $O(\log n)$ & No & Vertex \\
\cite{ALL+25} & $2^{O(\frac{s}{\epsilon} \log \frac{s}{\epsilon})}m^{1+\epsilon}$ & $\left(\left(2\lceil \frac{1}{\epsilon} \rceil + 3\right)s, s\right)$ & $\left(2\lceil \frac{1}{\epsilon} \rceil + 2\right)s$ & $O(\frac{s}{\epsilon} \log n)$ & Yes & Vertex \\
\cite{ALL+25} & $2^{O(\frac{s}{\epsilon} \log \frac{s}{\epsilon})}m^{1+\epsilon}$ & $(O(s/\epsilon), s)$ & $O(s/\epsilon)$ & $O(\log n)$ & Yes & Vertex \\
\cite{Kor25} & $s^{O(s^2)}n + O(m)$ & $(i, i) \; \forall i \le s$ & $s$ & $n$ & No & Vertex \\
\cref{thm:unbreakable_decomp} & $\widetilde O(m + s n)$ & $(s \log^{O(1)} n, s)$ & $s \log^{O(1)} n$ & $O(\log n)$ & No & Edge \\
\bottomrule
\end{tabular}
\caption{Comparison of unbreakable tree decompositions. ``Subtree?'' refers to \cref{def:subtree_unbreakability}. The \(s\log^{O(1)}n\) terms in our row are \(O(s\log^2 n\log\log n)\) specifically.}
\label{tab:unbreakable_decomp}
\end{table}

Using this decomposition, we improve the FPT algorithm of~\cite{LSS22} for Minimum \(k\)-Cut parameterized by the cut value. We first sample a spanning tree that \(2k-2\)-respects an optimum \(k\)-cut, using approximate dual tree packings. For a fixed respecting tree, the dynamic program only needs to consider partitions that cut at most \(2k-2\) tree edges. Large bags are handled by a two-level argument: first we guess the projected tree edges cut by the optimum transition, and then we guess the small components relevant to the non-giant sides of the cut. The resulting structure lets each transition be evaluated by a collection of small dynamic programs. Randomizing these guessing steps, together with the improved unbreakability parameter, gives the following bound. We remark that the exponent \(6k-6\) is important and dictates the final running time of our Minimum \(k\)-Cut algorithm.

\begin{theorem}[Informal]
Given a graph \(G\) with \(m\) edges, an integer \(s\ge \lambda_k(G)\), and a compact tree decomposition whose adhesions have size \(O(\Lambda)\) and whose bags are \((O(\Lambda),s)\)-edge-unbreakable, with no redundant bags, Minimum \(k\)-Cut can be solved with high probability in time
\[
    \widetilde O(m+k^{O(k)}\Lambda^{6k-6}n).
\]
In particular, for the decomposition above, where \(\Lambda=O(s\log^2 n\log\log n)\), this is
\[
    O(m)+ k^{O(k)}s^{6k-6}n \log^{12k+O(1)}n(\log\log n)^{6k-6}.
\]
\end{theorem}

This FPT algorithm also gives a near-linear approximation algorithm for \(\lambda_k\). Following the sampling framework of~\cite{LSS22}, we first remove very small \(2\)-cuts and then sample edges independently. With high probability the sampled graph has
\[
    \lambda_k=O(k\log n/\epsilon^3),
\]
and solving it exactly gives a \((1+\epsilon)\)-approximation in the original graph.

\begin{theorem}[Informal]
For every \(0<\epsilon\le1\), there is a randomized algorithm which, given an unweighted multigraph \(G=(V,E)\), computes a \((1+\epsilon)\)-approximation to \(\lambda_k(G)\) with high probability in time
\[
    \widetilde O(k^{O(1)}m) + (k/\epsilon)^{O(k)} n\log^{18k+O(1)}n(\log\log n)^{6k-6},
\]
provided \(\epsilon^{-1}\le n^{O(1)}\). More generally, the running time is
\[
    \widetilde O(k^{O(1)}m) + k^{O(k)} \left(\frac{k}{\epsilon^3}\right)^{6k-6} n\, \log^{18k-18} n\, (\log\log n)^{6k-6} \log^{O(1)}\!\bigl(k(n+\epsilon^{-1})\bigr).
\]
\end{theorem}

Finally, we combine the weighted sparse-graph algorithm, the approximation and FPT algorithm, and a refined version of the border/island framework of \cite{HL22}. The algorithm first computes a constant-factor estimate \(s=\Theta(\lambda_k)\). If \(s\) is small, we run the FPT algorithm. Otherwise, we apply border-preserving sparsification, obtaining an NI certificate \(H\), a partition \(\mathcal P\), and a contracted weighted multigraph \(G'\) with only \(\widetilde O(n/s)\) vertices. Every optimum cut has a low-weight border respected by \(\mathcal P\). This means that every optimal cut can combine some of its singletons with its other sides and then respect $\mathcal P$, thereby surviving in $G'$.

The algorithm then guesses the number \(i\) of singleton islands. If \(i=0\), the entire optimum cut is represented inside \(G'\), and we solve the corresponding weighted Min \(k\)-Cut instance on \(G'\). If \(i>0\), we enumerate the \((k-i)\)-cut border in \(G'\), lift it, and complete it by extracting the best \(i\) singleton islands. The cases \(i=1,2\) are handled by specialized routines; the cases \(i\ge3\) are handled using minimum vertex-weighted triangle routines and rectangular matrix multiplication.

\begin{theorem}[Informal]
Minimum \(3\)-Cut on simple unweighted \(n\)-vertex graphs can be solved with high probability in time
\[
    \widetilde O(n^2).
\]
Minimum \(4\)-Cut on simple unweighted \(n\)-vertex graphs can be solved with high probability in time
\[
    \widetilde O(n^{55/19}).
\]
Minimum \(5\)-Cut on simple unweighted \(n\)-vertex graphs can be solved with high probability in time
\[
    \widetilde O(n^{4.112007}).
\]
For every \(k\ge6\), Minimum \(k\)-Cut on simple unweighted \(n\)-vertex graphs can be solved with high probability in time
\[
    k^{O(k^2)} n^{ 1+(6k-6)\frac{k-1.749614}{7k-10} } (\log n)^{O(k^2)}.
\]
\end{theorem}

The exponent for \(k\ge6\) comes from balancing the FPT branch
\[
    ns^{6k-6}
\]
against the four-island branch after sparsification. The balance gives
\[
    E_k = 1+(6k-6)\frac{k-1.749614}{7k-10} = \frac67k-0.132\ldots+O(1/k).
\]

The following table lists representative exponents proved here.
\[
\begin{array}{c|c}
k & \text{Exponent proved here} \\
\hline
3  & 2 \\
4  & 55/19 \approx 2.89474 \\
5  & <4.112007 \\
6  & <4.984737 \\
7  & <5.846511 \\
8  & <6.706875 \\
10 & <8.425348 \\
20 & <17.004185
\end{array}
\]

\section{Preliminaries}

We use \(\widetilde O(\cdot)\) to suppress factors polylogarithmic in \(n\) and \(s\). All graphs are finite undirected multigraphs unless explicitly stated otherwise. Self-loops are ignored, since they do not cross any cut. A graph is simple if it is unweighted and has at most one edge between any pair of vertices.

\paragraph{Basic graph notation.}
For a weighted graph \(Q=(V(Q),E(Q),c)\) and
\(X\subseteq V(Q)\), let
\[
    \delta_Q(X):=
    \{uv\in E(Q):|\{u,v\}\cap X|=1\},
    \qquad
    c_Q(X):=
    \sum_{e\in\delta_Q(X)}c(e).
\]
All weighted graphs in the algorithms have nonnegative polynomial-bit integer capacities. For disjoint vertex sets \(A,B\subseteq V\), let
\[
    E_G(A,B):=\{uv\in E: u\in A,\ v\in B\}.
\]
For \(X\subseteq V\), write
\[
    E_G(X):=\{uv\in E: u,v\in X\}.
\]
For a vertex \(v\), we abbreviate \(E_G(v,X):=E_G(\{v\},X)\). For an edge set \(F\subseteq E(G)\), let $\End_G(F)$ be the set of endpoints of the edges in \(F\), and let \(V(F)\) be the set of vertices incident with edges of \(F\).

For \(X\subseteq V(G)\), let
\[
    N_G(X):=\{v\in V(G)\setminus X: v \text{ has a neighbor in } X\}.
\]

\paragraph{Cuts and partitions.}

For a partition \(\Pi\) of a vertex set \(X\), define its crossing edge set by
\[
    \delta_G(\Pi) := \{uv\in E(G): u \text{ and } v \text{ lie in different parts of } \Pi\}.
\]
The value of \(\Pi\) in an unweighted graph is \(|\delta_G(\Pi)|\); in a weighted graph it is the total weight of \(\delta_G(\Pi)\).

For a partition \(\Pi\) of a set \(X\) and a subset \(Y\subseteq X\), let
\[
    \Pi|_Y:=\{P\cap Y:P\in\Pi,\ P\cap Y\ne\emptyset\}
\]
be the partition of \(Y\) induced by \(\Pi\). A partition \(\Pi\) refines a partition \(\Pi'\) if every part of \(\Pi\) is contained in a part of \(\Pi'\); equivalently, \(\Pi'\) is a coarsening of \(\Pi\).

A family \(\mathcal L\) of vertex subsets is laminar if for every \(A,B\in\mathcal L\), either
\[
    A\cap B=\emptyset,\qquad A\subseteq B,\qquad \text{or}\qquad B\subseteq A.
\]

\begin{definition}[\(s\)-cut]
\label{def:s_cut}
An \(s\)-cut in a graph \(G\) is a bipartition \((A,V(G)\setminus A)\) whose crossing edge set has size at most \(s\), i.e.
\[
    |\delta_G(A)|\le s.
\]
\end{definition}

\section{Weighted Minimum \(k\)-Cut for Sparse Graphs}
\label{sec:weighted_sparse_k_cut}

We give randomized and deterministic versions of the same recursive
algorithm.  Both compute an approximate solution to the dual \(k\)-cut LP,
viewed as a weighted packing of spanning trees.  The randomized algorithm
samples a logarithmic number of trees from the packing with probability proportional to weight, whereas
the deterministic algorithm processes every tree in its support.  Once a
suitable tree is found, both algorithms enumerate its respecting cuts by
choosing their crossing tree edges, evaluate them using orthogonal range
sums, and recursively separate the remaining \(k-1\) parts.

The other difference is the perturbation.  The randomized algorithm uses
small independent perturbations to ensure with high probability that an
optimum \(k\)-cut has a strict-light side.  The deterministic algorithm
uses a lexicographic perturbation based on the cuts themselves with the same property.  Thus all
of the structural and recursive arguments are shared by the two
algorithms.

Throughout, capacities are nonnegative integers of polynomial bit length.
For a top-level randomized call, let \(N\) and \(K\) denote its number of
vertices and its parameter.  Every recursive call retains \(N,K\) solely
as amplification parameters; its local order and parameter are denoted by
\(n,k\).  All randomized primitives are amplified to failure probability
at most \(N^{-\Omega(K^2)}\).  At the top level, of course, \(N=n\) and
\(K=k\).

\begin{definition}[Minimum \(k\)-Cut] \label{def:min_k_cut}
A \(k\)-cut of a graph \(G=(V,E)\) is a partition of \(V\) into exactly
\(k\) nonempty parts.  Its capacity is the total capacity of the edges
whose endpoints lie in different parts.  The Minimum \(k\)-Cut problem
asks for a \(k\)-cut of minimum capacity; its value is denoted by
\(\lambda_k(G)\).  We use the convention
\[
    \lambda_j(F)=+\infty
    \qquad\text{when }|V(F)|<j.
\]
Allowing at least \(k\) parts gives the same optimum for nonnegative
capacities, since parts may be merged without increasing the capacity.
\end{definition}

\begin{theorem}[Weighted global minimum cut {\cite{Kar00}}]
\label{thm:weighted_global_mincut}
Let \(Q\) be an \(n\)-vertex, \(m\)-edge weighted graph.  A weighted global minimum
cut can be found with high probability in
\[
    O(m\log^3(n))
\]
time.
\end{theorem}

\begin{theorem}[Deterministic weighted global minimum cut
{\cite{HLRW24}}]
\label{thm:deterministic_weighted_global_mincut}
A weighted global minimum cut can be found deterministically in
\[
    \widetilde O(m)
\]
time.
\end{theorem}

\begin{lemma}[Disconnected and zero-capacity cases]
\label{lem:zero_capacity_reduction_kcut}
Let \(k\ge3\).  Given algorithms for weighted Minimum \(j\)-Cut for all
\(2\le j<k\), one can either find a minimum \(k\)-cut or reduce to an
instance whose positive-capacity subgraph is connected.  The reduction
uses \(O(k^2)\) recursive calls, where for each fixed parameter
the instances are disjoint; summing over $O(k)$ parameters is absorbed by $k^{O(k)}$, an
\(O(k^3)\)-time dynamic program, and \(O(m+n)\) additional time.
\end{lemma}

\begin{proof}
Delete the zero-capacity edges and let \(C_1,\ldots,C_h\) be the connected
components of the remaining graph.  If \(h\ge k\), any grouping into
\(k\) nonempty unions of components gives a zero-capacity \(k\)-cut.  If
\(h=1\), the positive-capacity subgraph is connected.

It remains to consider \(2\le h<k\).  A global \(k\)-cut may induce
\(q_a\) parts inside \(C_a\) with \(\sum_aq_a>k\), because one global
part may meet several positive-capacity components.  Nevertheless, some
optimum satisfies
\[
    1\le q_a\le k-h+1,
    \qquad
    \sum_{a=1}^h q_a=k.
\]
Indeed, starting from any optimum, repeatedly merge two induced parts
inside one component while the total number of induced parts exceeds
\(k\).  Such a merge cannot increase the capacity.  Once exactly \(k\)
local parts remain, make them the \(k\) distinct global parts.

For every \(a\) and every
\[
    2\le p\le\min\{k-h+1,|C_a|\},
\]
recursively compute \(\lambda_p(Q[C_a])\), and put
\(\lambda_1(Q[C_a])=0\).  Every recursive parameter is at most \(k-1\).
The optimum is therefore
\[
    \min_{\substack{q_1,\ldots,q_h\ge1\\
                    \sum_aq_a=k}}
    \sum_{a=1}^h\lambda_{q_a}(Q[C_a]).
\]
Conversely, every solution to this minimization gives a valid global
\(k\)-cut by taking its local parts as distinct global parts.  The stated
dynamic program finds the minimum and the corresponding partition.  For
each fixed \(p\), the induced instances are vertex- and edge-disjoint.
\end{proof}

\subsection{Tree Packings and Respecting Cuts}

We first record the common ingredients.  Both algorithms use the same
deterministically constructed approximate dual packing.  The randomized
algorithm samples a logarithmic number of records from its normalized
weights, whereas the deterministic algorithm processes the entire support.
The latter has
\[
    O\!\left(m\eta^{-2}\log^3(2n)\right)
\]
records.

\begin{lemma}[Approximate dual tree packing {\cite{CQX19,Qua22}}]
\label{lem:dual_tree_packing}
Let \(F=(V,E,c)\) be a connected weighted graph, put
\(n:=|V|\), \(m:=|E|\), and let \(2\le k\le n\).  For every
\(0<\eta<1/2\), one can deterministically compute in
\[
    O\!\left(m\eta^{-2}\log^3(2n)\right)
\]
time an implicit finitely supported family of nonnegative tree weights
\((y_T)\) and nonnegative edge slacks \((z_e)\).  Writing
\[
    \tau:=\sum_Ty_T,
\]
they satisfy
\begin{align}
    \sum_{T\ni e}y_T
    &\le c(e)+z_e
        \qquad (e\in E),
        \label{eq:dual_packing_load}\\
    (k-1)\tau
    &\ge(1-\eta)\frac{\lambda_k(F)}2+z(E).
        \label{eq:dual_packing_approx}
\end{align}
The support has
\[
    O\!\left(m\eta^{-2}\log^3(2n)\right)
\]
records.

The records can be sampled with probabilities proportional to their
weights.  Given \(r\) sampled record indices, their spanning-tree
extensions can be materialized in
\[
    O\!\left(m\eta^{-2}\log^3(2n)+rn\right)
\]
additional time.  The entire support can likewise be streamed with
\(O(n)\) output work per materialized tree.
\end{lemma}

\begin{proof}
Quanrud's algorithm approximately solves the positive forest-packing
dual
\[
    \max
    \sum_J(|J|+k-n)y_J
    \quad\text{subject to}\quad
    \sum_{J\ni e}y_J\le c(e)
    \quad(e\in E),
\]
where \(J\) ranges over forests of \(F\).  After changing its internal
accuracy by a constant factor, it returns a feasible packing satisfying
\[
    \sum_J(|J|+k-n)y_J
    \ge(1-\eta)\operatorname{OPT}_{\mathrm{LP}}
    \ge(1-\eta)\frac{\lambda_k(F)}2.
\]

For every support forest \(J\), choose a spanning-tree extension
\(T_J\supseteq J\), transfer the weight \(y_J\) to \(T_J\), and define
\[
    z_e:=\sum_{J:\,e\in T_J\setminus J}y_J.
\]
Then
\[
    \sum_{T\ni e}y_T
    =
    \sum_{J:\,e\in J}y_J+
    \sum_{J:\,e\in T_J\setminus J}y_J
    \le c(e)+z_e.
\]
Moreover,
\[
\begin{aligned}
    (k-1)\tau-z(E)
    &=
    \sum_J
    \bigl(k-1-|T_J\setminus J|\bigr)y_J
    =
    \sum_J
    \bigl(k-1-(n-1-|J|)\bigr)y_J\\
    &=
    \sum_J(|J|+k-n)y_J
    \ge(1-\eta)\frac{\lambda_k(F)}2.
\end{aligned}
\]

In Quanrud's implementation, every inserted forest is a prefix of the
current maintained minimum spanning tree.  We choose that tree as
\(T_J\).  Record the insertion index, prefix length, and inserted weight.
Prefix sums support weighted sampling of record indices.  Because the
packing algorithm is deterministic, sorting the sampled indices and
replaying its update sequence once reconstructs the corresponding MST
states; writing one tree explicitly costs \(O(n)\).  The same replay
streams the complete support.
\end{proof}

\begin{lemma}[A sampled tree respects a light side]
\label{lem:sampled_light_side_respecting}
Use the approximate packing in \cref{lem:dual_tree_packing} with
\(\eta=1/100\), and sample \(T\) with probability \(y_T/\tau\).  If
\(A\) is a nontrivial cut satisfying
\[
    c_F(A)<\frac{2}{k}\lambda_k(F),
\]
then
\[
    \Pr\bigl[|\delta_T(A)|\le k-1\bigr]\ge\frac1{4k}.
\]
\end{lemma}

\begin{proof}
Write \(\lambda:=\lambda_k(F)\), \(D:=|\delta_T(A)|\),
\(x:=z(E)/\lambda\), and \(b:=(1-\eta)/2\).

By \eqref{eq:dual_packing_load} and
\eqref{eq:dual_packing_approx},
\begin{equation}
\label{eq:expected_packed_crossings}
    \mathbb E[D]
    \le
    (k-1)\frac{2/k+x}{b+x}.
\end{equation}
Since \(T\) is spanning and \(A\) is nontrivial, \(D\ge1\).  Thus
\[
    \Pr[D\le k-1]
    \ge\frac{k-\mathbb E[D]}{k-1}.
\]
For \(k\ge5\), we have \(2/k<b\), so the ratio in
\eqref{eq:expected_packed_crossings} is at most one and the last display
is at least \(1/(k-1)\).  For \(k=3\) and \(k=4\), the ratio is maximized
at \(x=0\), giving respectively
\[
    \Pr[D\le2]
    \ge\frac{1-9\eta}{6(1-\eta)},
    \qquad
    \Pr[D\le3]
    \ge\frac{1-4\eta}{3(1-\eta)}.
\]
For \(\eta=1/100\), all three bounds are at least \(1/(4k)\).
\end{proof}

Fix a root \(r\).  A cut is \emph{rooted} if its selected side does not
contain \(r\).  If \(R\) is any rooted tree on the graph's vertex set,
we say that a rooted cut \(S\) \(p\)-respects \(R\) when
\(|\delta_R(S)|\le p\).  The tree need not be a subgraph of the graph in
which the cut is evaluated.

\begin{lemma}[Enumeration of respecting cuts]
\label{lem:respecting_enumeration}
Let \(H=(W,E,c)\) be an \(n\)-vertex, \(m\)-edge weighted multigraph,
let \(R\) be an arbitrary rooted tree on \(W\), and let \(p,M\ge1\).
After
\[
    O((m+n)\log(2n))
\]
preprocessing, one can enumerate every nonempty rooted cut
\(S\subseteq W\setminus\{r\}\) with \(|\delta_R(S)|\le p\), compute
its capacity, and return the \(M\) lightest such cuts in
\[
    p^{O(1)}n^p\log(2n)
\]
additional time and
\[
    O((m+n)\log(2n)+pM)
\]
space.  Each returned cut is represented by its at most \(p\) crossing
tree edges.
\end{lemma}

\begin{proof}
Order \(W\) by a DFS traversal of \(R\).  For every non-loop edge
\(uv\in E(H)\), insert the two weighted points
\[
    (\operatorname{tin}(u),\operatorname{tin}(v))
    \quad\text{and}\quad
    (\operatorname{tin}(v),\operatorname{tin}(u)).
\]
A standard static two-dimensional orthogonal range-sum structure, as in
the range-counting formulation of \cite[Lemma~1]{GMW20b}, can be built in
\(O((m+n)\log(2n))\) time and space and answers a rectangle query in
\(O(\log(2n))\) time.

For a tree edge \(f\), let \(R_f\) be the rooted subtree below \(f\).
Every rooted cut has the unique representation
\begin{equation}
\label{eq:rooted_cut_parity}
    S=\mathop{\triangle}_{f\in\delta_R(S)}R_f.
\end{equation}
Indeed, membership changes along a root-to-vertex path precisely at the
crossing tree edges.  Conversely, if \(B\subseteq E(R)\), then
\(\triangle_{f\in B}R_f\) crosses exactly the edges in \(B\).  It
therefore suffices to enumerate the sets
\[
    B\subseteq E(R),
    \qquad 1\le|B|\le p.
\]

Each \(R_f\) is one DFS interval.  Given \(q:=|B|\), sort the \(2q\)
interval endpoints and sweep their parity.  This expresses both
\(\triangle_{f\in B}R_f\) and its complement as disjoint unions of at
most \(q+1\) intervals.  Its capacity is consequently the sum of at most
\((q+1)^2\) rectangle sums and is computed in
\(p^{O(1)}\log(2n)\) time.  Since
\[
    \sum_{q=1}^p\binom{n-1}{q}=p^{O(1)}n^p,
\]
the claimed enumeration bound follows.

Keep candidates in a buffer of size \(2M\).  Whenever the buffer fills,
use linear-time selection to retain its \(M\) lightest members, breaking
ties by a fixed order on their boundary-edge sets.  A pruning costs
\(O(M)\) and occurs only after \(M\) new candidates have arrived, so the
selection work is linear in the number of enumerated cuts.
\end{proof}

\begin{definition}[Canonical cut representative]
\label{def:canonical_cut_representative}
For the fixed root \(r\in V\) and every nontrivial cut \(S\subsetneq V\),
put
\[
    \kappa(S):=
    \begin{cases}
        S,&r\notin S,\\
        V\setminus S,&r\in S.
    \end{cases}
\]
\end{definition}

\begin{lemma}[Number of strict-light cuts]
\label{lem:light_cut_counting_bound}
Let \(Q\) be a connected weighted graph.  The number of cuts
\(S\subsetneq V(Q)\) satisfying
\[
    c_Q(S)<\frac{2}{k}\lambda_k(Q)
\]
is at most \(k^{O(k)}|V(Q)|\).
\end{lemma}

\begin{proof}
Gupta, Harris, Lee, and Li prove that there are fewer than \(2^{k-2}\)
cuts of value below \(\lambda_k/(k-1)\), and at most \(k^{O(k)}n\)
cuts in the remaining range below \(2\lambda_k/k\); see
\cite[Lemma~7 and Theorem~9]{GHLL21}.  The claim follows.
\end{proof}

Fix \(b_k=k^{O(k)}\) for which the bound in
\cref{lem:light_cut_counting_bound} is at most \(b_kn\), and write
\[
    M_k:=b_kn
\]
for an \(n\)-vertex instance with parameter \(k\).

\subsection{Randomized Algorithm}
\label{subsec:randomized_weighted_kcut}

We perturb the capacities independently in every recursive invocation.
The ambient parameters \(N,K\) make all isolation and sampling failures
small enough to union-bound over the entire recursion.

\begin{lemma}[Random perturbation for \(k\)-cut]
\label{lem:perturbation_kcut}
Let \(Q=(V,E,c)\) be a connected positive-capacity graph with
\(n:=|V|\ge k\ge3\), after parallel edges have been aggregated.  Set
\[
    R:=N^{c_{\mathrm{iso}}K^2}
\]
for a sufficiently large absolute constant \(c_{\mathrm{iso}}\), choose
independent \(\rho(e)\in[R]\), and define
\[
    B:=|E|R+1,
    \qquad
    c^\star(e):=Bc(e)+\rho(e).
\]
With probability at least \(1-N^{-\Omega(K^2)}\), every minimum
\(k\)-cut of \(Q^\star=(V,E,c^\star)\) is a minimum \(k\)-cut of
\(Q\), and every minimum \(k\)-cut of \(Q^\star\) has a side \(A\)
satisfying
\[
    c_{Q^\star}(A)<\frac{2}{k}\lambda_k(Q^\star).
\]
The perturbation adds \(O(K^2\log(2N))\) bits to every capacity.
\end{lemma}

\begin{proof}
For every \(F\subseteq E\),
\[
    c^\star(F)=Bc(F)+\rho(F),
    \qquad 0\le\rho(F)<B.
\]
Thus perturbed capacities compare edge sets first by original capacity.

There are at most \(n^kk^{O(k^2)}\) original minimum \(k\)-cuts by
\cite[Corollary~19]{GHLL21}.  Fix one, say
\((A_1,\ldots,A_k)\).  Connectivity implies
\(\delta_Q(A_i)\ne\delta_Q(A_j)\) for \(i\ne j\), since equality would
leave \(A_i\cup A_j\) disconnected from all other parts.  Hence
\(c_{Q^\star}(A_i)=c_{Q^\star}(A_j)\) is a nontrivial linear equation
in the independent perturbations and holds with probability at most
\(1/R\).  The choice of \(R\), a union bound over all optimum cuts and
all pairs of their sides, and \(n\le N\), \(k\le K\), give the claimed
failure probability.

For every optimum \(k\)-cut,
\[
    \sum_{i=1}^k c_{Q^\star}(A_i)=2\lambda_k(Q^\star).
\]
Once the side capacities are pairwise distinct, at least one is strictly
below their average.
\end{proof}

\begin{lemma}[Packed light-side candidates]
\label{lem:randomized_light_candidate_selection}
Let \(Q\) be a connected positive-capacity \(n\)-vertex, \(m\)-edge
graph with parameter \(k\ge3\).  With probability at least
\(1-N^{-\Omega(K^2)}\), one can construct a family \(\mathcal F\) of
at most \(M_k\) distinct canonical cuts containing the canonical
representative of every strict-light cut of \(Q\).  The running time is
\[
    K^{O(K)}n^{k-2}(m+n)\log^3(2N).
\]
\end{lemma}

\begin{proof}
Compute the approximate packing of \cref{lem:dual_tree_packing} with
\(\eta=1/100\), and independently sample
\[
    s:=C K^3\log(2N)
\]
support trees, for a sufficiently large constant \(C\).
Sample all support-record indices first and materialize their
spanning-tree extensions at once, as guaranteed by
\cref{lem:dual_tree_packing}.
Root every
sampled tree at the fixed vertex \(r\).  By
\cref{lem:sampled_light_side_respecting}, a fixed strict-light cut is
\((k-1)\)-respected by a sampled tree with probability at least
\(1/(4k)\).  Therefore the probability that none of the \(s\) trees
respects it is at most \(N^{-\Omega(K^2)}\).  A union bound over the at
most \(M_k\) strict-light cuts proves simultaneous coverage.

For every sampled tree, apply \cref{lem:respecting_enumeration} with
\(p=k-1\), retaining its \(M_k\) lightest rooted cuts.  If the rooted
representative \(\kappa(A)\) of a covered strict-light cut were discarded
from its witnessing tree, that tree's local output would contain \(M_k\)
distinct cuts of capacity at most \(c_Q(A)\).  Together with
\(\kappa(A)\), this would contradict
\cref{lem:light_cut_counting_bound}.  Thus every strict-light cut
survives at least one local truncation.

Materialize the at most \(sM_k\) local survivors, orient them by
\(\kappa\), and deduplicate them by radix sorting their incidence
vectors.  Retain the \(M_k\) lightest distinct cuts using the same
linear-selection buffer.  The preceding counting argument also shows
that no strict-light cut is lost in this final truncation.

The packing costs \(O(m\log^3(2N))\).  The fixed-tree structures and
enumerations cost
\[
    K^{O(1)}\bigl(m+n+n^{k-1}\bigr)
    \log(2n)\log(2N).
\]
Materializing the local survivors costs
\(K^{O(K)}n^2\log(2N)\).  Since
\[
    m+n+n^{k-1}
    \le O\!\left(n^{k-2}(m+n)\right)
    \qquad(k\ge3),
\]
all terms are bounded by the claimed running time.
\end{proof}

\begin{algorithm}[H]
\caption{\(\textsc{WeightedKCut}(Q,k)\)}
\label{alg:weighted_k_cut}
\begin{algorithmic}[1]
\Require An undirected capacitated multigraph \(Q\) and an integer
         \(2\le k\le |V(Q)|\)
\Ensure A minimum \(k\)-cut of \(Q\)
\If{\(k=2\)}
    \State \Return a weighted global minimum cut using
           \cref{thm:weighted_global_mincut}, amplifying success probability.
\EndIf
\State Apply \cref{lem:zero_capacity_reduction_kcut}; if it returns a
       cut, return it
\State Aggregate parallel edges and form \(Q^\star\) using
       \cref{lem:perturbation_kcut}
\State Construct \(\mathcal F^\star\) from \(Q^\star\) using
       \cref{lem:randomized_light_candidate_selection}
\State \(\mathcal P_{\mathrm{best}}\gets\bot\)
\ForAll{\(S\in\mathcal F^\star\)}
    \If{\(|S|\ge k-1\)}
        \State Recursively compute a minimum \((k-1)\)-cut of
               \(Q^\star[S]\), add \(V(Q)\setminus S\) as one part,
               and update \(\mathcal P_{\mathrm{best}}\)
    \EndIf
    \If{\(|V(Q)\setminus S|\ge k-1\)}
        \State Recursively compute a minimum \((k-1)\)-cut of
               \(Q^\star[V(Q)\setminus S]\), add \(S\) as one part,
               and update \(\mathcal P_{\mathrm{best}}\)
    \EndIf
\EndFor
\State \Return \(\mathcal P_{\mathrm{best}}\), interpreted in \(Q\)
\end{algorithmic}
\end{algorithm}

\begin{theorem}[Weighted Minimum \(k\)-Cut on Sparse Graphs]
\label{thm:weighted_k_cut_sparse}
For every \(k\ge2\), \cref{alg:weighted_k_cut} returns a minimum
\(k\)-cut with high probability and runs in
\begin{equation}
\label{eq:weighted_kcut_randomized_time}
    O(m)+k^{O(k^2)}n^{k-2}(m+n)\log^3(2n)
\end{equation}
time.  On a connected positive-capacity instance this is
\[
    O(m)+k^{O(k^2)}n^{k-2}m\log^3(2n).
\]
\end{theorem}

\begin{proof}
We prove correctness by induction on \(k\).  The case \(k=2\) is
\cref{thm:weighted_global_mincut};
\cref{lem:zero_capacity_reduction_kcut} handles the disconnected and
zero-capacity cases.  Thus assume that the positive-capacity graph is
connected.

Condition on the success of the perturbation, packing, tree sampling,
and recursive calls.  By \cref{lem:perturbation_kcut}, a minimum
\(k\)-cut of \(Q^\star\) has a strict-light side \(A\).  By
\cref{lem:randomized_light_candidate_selection},
\(S:=\kappa(A)\) belongs to \(\mathcal F^\star\).  If \(S=A\), the
other \(k-1\) optimum parts lie in \(V\setminus S\); if
\(S=V\setminus A\), they lie in \(S\).  The corresponding recursive
call therefore completes the candidate to a \(k\)-cut of value
\(\lambda_k(Q^\star)\).  Every candidate constructed by the algorithm is
a valid \(k\)-cut, so the returned one is optimum for \(Q^\star\), and
hence for \(Q\).

It remains to prove the running time.  Let
\(L:=\log(2N)\), and let \(T_j(n,m)\) denote the worst-case time of a
local call with parameter \(j\), excluding the one-time input scan.
We claim
\[
    T_j(n,m)
    \le A_j n^{j-2}(m+n)L^3,
    \qquad
    A_j=K^{O(Kj)}.
\]
The base case follows from \cref{thm:weighted_global_mincut}.  For
\(j\ge3\), candidate construction is covered by
\cref{lem:randomized_light_candidate_selection}.  There are at most
\(M_j=j^{O(j)}n\) candidates.  For one candidate, let the two induced
instances have orders \(n_1,n_2\) and edge counts \(m_1,m_2\).  They
satisfy \(n_1+n_2=n\) and \(m_1+m_2\le m\), and hence
\[
    \sum_{a=1}^2
    n_a^{j-3}(m_a+n_a)
    \le n^{j-3}(m+n).
\]
By induction, the two recursive completions for one candidate therefore
cost at most
\[
    A_{j-1}n^{j-3}(m+n)L^3.
\]
Constructing their induced graphs costs \(O(m+n)\).  Multiplying by
\(M_j\) gives
\[
    T_j(n,m)
    \le K^{O(K)}(A_{j-1}+1)
       n^{j-2}(m+n)L^3.
\]
Thus \(A_j\le K^{O(K)}(A_{j-1}+1)\), which yields
\(A_K=K^{O(K^2)}\).

For the disconnected reduction, fix a recursive parameter \(p<j\).
Its induced instances are vertex- and edge-disjoint, so
\[
    \sum_a |C_a|^{p-2}
       \bigl(|E(Q[C_a])|+|C_a|\bigr)
    \le n^{p-2}(m+n).
\]
The \(O(j^2)\) parameter--component calls and the composition dynamic
program are absorbed into the same bound.

Finally, the recursion contains at most
\(K^{O(K^2)}N^{K-2}\) randomized invocations.  By the amplification convention, each fails with probability at most
\(N^{-\Omega(K^2)}\); a union bound proves the high-probability claim.
The initial aggregation and zero-capacity scan take \(O(m)\) time.
\end{proof}

\subsection{Deterministic Algorithm}
\label{subsec:deterministic_weighted_kcut}

The randomized algorithm has only two sources of randomness: the
perturbation and the sampling of trees from the packing.  We remove the
first using a different lexicographic perturbation and the second by processing
the entire support of the same approximate packing.

We first record the deterministic perturbation.  Parallel edges are
aggregated before it is applied.

\begin{lemma}[Deterministic perturbation]
\label{lem:deterministic_perturbation_kcut}
Let \(Q=(V,E,c)\) be a connected positive-capacity graph with
\(n:=|V|\ge k\ge3\) and \(m:=|E|\).  Order
\[
    E=\{e_1,\ldots,e_m\},
\]
set
\begin{equation}
\label{eq:deterministic_kcut_perturbation}
    B:=2^{m+1},
    \qquad
    c^\star(e_i):=Bc(e_i)+2^i,
\end{equation}
and let \(Q^\star=(V,E,c^\star)\).  Every minimum \(k\)-cut of
\(Q^\star\) is a minimum \(k\)-cut of \(Q\).  Moreover, every minimum
\(k\)-cut of \(Q^\star\) has a side \(A\) satisfying
\[
    c_{Q^\star}(A)
    <\frac{2}{k}\lambda_k(Q^\star).
\]
If the original capacities have \(L\) bits, the perturbed capacities have
\(O(L+m)\) bits.
\end{lemma}

\begin{proof}
For every \(F\subseteq E\),
\[
    c^\star(F)
    =
    Bc(F)+\sum_{e_i\in F}2^i,
    \qquad
    0\le\sum_{e_i\in F}2^i<B.
\]
Thus \(c^\star\) compares edge sets first by their original capacities and
then by their incidence vectors.  In particular, every perturbed optimum
is an original optimum.

Let \((A_1,\ldots,A_k)\) be any \(k\)-cut.  For \(p\ne q\), connectivity
implies
\[
    \delta_Q(A_p)\ne\delta_Q(A_q).
\]
Indeed, equality would imply that no edge joins
\(A_p\cup A_q\) to any of the other nonempty parts.  Hence the binary
terms in the capacities
\[
    c_{Q^\star}(A_1),\ldots,c_{Q^\star}(A_k)
\]
are pairwise distinct.  For an optimum \(k\)-cut these capacities sum to
\(2\lambda_k(Q^\star)\), so one of them is strictly smaller than their
average.
\end{proof}

We define the deterministic algorithm by modifying
\cref{alg:weighted_k_cut} as follows.

\begin{algorithm}[H]
\caption{\(\textsc{DeterministicWeightedKCut}(Q,k)\)}
\label{alg:deterministic_weighted_k_cut}
\begin{algorithmic}[1]
\Require An undirected weighted multigraph \(Q\) and an integer
         \(2\le k\le|V(Q)|\)
\Ensure A minimum \(k\)-cut of \(Q\)
\State Run the recursive framework of \cref{alg:weighted_k_cut}, with
       the following modifications
\State When \(k=2\), use
       \cref{thm:deterministic_weighted_global_mincut}       
\State Replace the random perturbation by
       \cref{lem:deterministic_perturbation_kcut}
\State Compute the packing of \cref{lem:dual_tree_packing}
       deterministically with \(\eta=1/100\)
\State Process every tree \(T\) in its support: apply
       \cref{lem:respecting_enumeration} with
       \(p=k-1\), retaining the \(M_k\) lightest rooted cuts for \(T\)
\State Materialize the locally retained cuts, canonicalize and
       deduplicate them, and retain the \(M_k\) lightest distinct
       cuts overall
\State Complete these candidates recursively exactly as in
       \cref{alg:weighted_k_cut}, using the deterministic algorithm in
       every recursive call
\end{algorithmic}
\end{algorithm}

\begin{theorem}[Deterministic Weighted Minimum \(k\)-Cut]
\label{thm:deterministic_weighted_kcut}
For every fixed \(k\ge2\),
\cref{alg:deterministic_weighted_k_cut} deterministically returns a
minimum \(k\)-cut.  If the input capacities have \(L=n^{O(1)}\) bits,
its running time is
\[
    O(m)+k^{O(k^2)}n^{k+O(1)}.
\]
\end{theorem}

\begin{proof}
We prove correctness by induction on \(k\).  The base case is
\cref{thm:deterministic_weighted_global_mincut} and
\cref{lem:zero_capacity_reduction_kcut} handles the disconnected and
zero-capacity cases.  Thus consider a connected positive-capacity
instance, after parallel edges have been aggregated.

Form \(Q^\star\) using
\cref{lem:deterministic_perturbation_kcut}.  Let \(A\) be a strict-light
side of a minimum \(k\)-cut of \(Q^\star\), whose existence is guaranteed
by that lemma.  Compute the approximate dual packing from
\cref{lem:dual_tree_packing}.  If a tree \(T\) is sampled from its support
with probability \(y_T/\tau\), then
\cref{lem:sampled_light_side_respecting} gives
\[
    \Pr\bigl[|\delta_T(A)|\le k-1\bigr]\ge\frac1{4k}>0.
\]
Consequently, at least one tree in the support
\((k-1)\)-respects \(A\).  Because the deterministic algorithm processes
the entire support, it processes such a tree.

For each support tree, the algorithm retains its \(M_k\) lightest rooted
\((k-1)\)-respecting cuts.  Suppose that the rooted representative
\(\kappa(A)\) were discarded from a tree that respects it.  That tree's
local list would then contain \(M_k\) distinct cuts of capacity at most
\[
    c_{Q^\star}(A)
    <\frac{2}{k}\lambda_k(Q^\star).
\]
Together with \(\kappa(A)\), these would contradict
\cref{lem:light_cut_counting_bound}.  Thus \(\kappa(A)\) survives some
local truncation.  The same argument shows that it survives the final
deduplication and global truncation to the \(M_k\) lightest distinct cuts.

When \(\kappa(A)\) is evaluated, one of its two sides is the union of the
other \(k-1\) parts of the optimum \(k\)-cut.  By induction, the
corresponding recursive call finds a minimum \((k-1)\)-cut of that side
and therefore completes the candidate to a \(k\)-cut of value
\(\lambda_k(Q^\star)\).  Every candidate considered by the algorithm is
a valid \(k\)-cut, so the returned one is optimum for \(Q^\star\), and
hence for \(Q\).

It remains to bound the running time.  For \(\eta=1/100\), the packing
has
\[
    s_{\mathrm{pack}}
    =
    O\!\left(m\log^3(2n)\right)
\]
support trees.  For one tree,
\cref{lem:respecting_enumeration} with \(p=k-1\) takes
\[
    O\!\left(
        (m+n)\log(2n)
        +k^{O(1)}n^{k-1}\log(2n)
    \right)
\]
time.  Processing the entire support therefore takes
\[
    k^{O(1)}
    m\bigl(m+n+n^{k-1}\bigr)\log^4(2n).
\]
The algorithm retains at most \(M_k\) cuts per support tree.  Thus at
most
\[
    s_{\mathrm{pack}}M_k
    =
    k^{O(k)}mn\log^3(2n)
\]
compact cuts are materialized.  Exact canonicalization and deduplication
of their incidence vectors costs
\[
    k^{O(k)}mn^2\log^3(2n),
\]
which is covered by the preceding bound because \(k\ge3\).

After parallel edges have been aggregated, \(m\le n^2\).  Hence all
nonrecursive work in a parameter-\(k\) invocation is
\[
    k^{O(k)}n^{k+O(1)}\operatorname{poly}(L+n).
\]
There are at most
\[
    M_k=k^{O(k)}n
\]
recursive candidate completions.  If \(D_k(n,L)\) denotes the worst-case
running time after aggregation, then for an absolute constant \(C\),
\[
    D_k(n,L)
    \le
    k^{O(k)}nD_{k-1}(n,L+n^2)
    +
    k^{O(k)}n^{k+C}(L+n)^C.
\]
The bit length increases by at most \(O(n^2)\) at each of the at most
\(k\) recursion levels, and therefore remains polynomial.  Starting from
the polynomial-time deterministic \(k=2\) algorithm and using
\[
    \prod_{j=3}^k j^{O(j)}=k^{O(k^2)},
\]
the recurrence gives
\[
    D_k(n,L)
    \le k^{O(k^2)}n^{k+O(1)}.
\]
The disconnected reduction is absorbed by the same bound because, for
each fixed recursive parameter, its induced instances are
vertex-disjoint.  The initial aggregation costs \(O(m)\).
\end{proof}

\section{Unbreakable Decomposition}

Throughout this subsection \(G=(V,E)\) is an undirected unweighted multigraph, \(n=|V|\), and \(s\ge 1\) is the edge-cut parameter. We assume self-loops, if present, have been deleted, since they do not cross any cut and do not affect tree decompositions. We use near-linear computation of a tree cut-sparsifier by deleting edges below a threshold to construct our unbreakable decomposition.

\begin{definition}[Edge unbreakability]
\label{def:unbreakability}
A vertex set \(X\subseteq V(G)\) is \((q,s)\)-edge-unbreakable in \(G\) if every \(s\)-cut \((L,V(G)\setminus L)\) satisfies
\[
    |L\cap X|\le q \qquad\text{or}\qquad |(V(G)\setminus L)\cap X|\le q.
\]
A \((q,s)\)-breakable witness for \(X\) is an \(s\)-cut for which both quantities are larger than \(q\). Thus \(X\) is \((q,s)\)-edge-unbreakable iff it has no \((q,s)\)-breakable witness. Any set \(X\) with \(|X|\le q\) is vacuously \((q,s)\)-edge-unbreakable.
\end{definition}

\begin{definition}[Tree Decomposition]
\label{def:tree_decomposition}
A tree decomposition of a graph \(G\) is a pair \((\tau,\beta)\), where \(\tau\) is a tree and $\beta:V(\tau)\to 2^{V(G)}$ assigns a bag \(\beta(t)\subseteq V(G)\) to every node \(t\in V(\tau)\), such that:
\begin{itemize}
    \item for each vertex \(v\in V(G)\), the set $\{t\in V(\tau): v\in\beta(t)\}$ induces a connected subtree of \(\tau\).
    \item for each edge \(uv\in E(G)\), there is a node \(t\in V(\tau)\) such that \(u,v\in\beta(t)\).
\end{itemize}

A rooted tree decomposition is a tree decomposition together with a root \(r\in V(\tau)\). If \(t\ne r\) and \(p(t)\) is the parent of \(t\), the adhesion of \(t\) is
\[
    \sigma(t):=\beta(t)\cap\beta(p(t)).
\]
We set \(\sigma(r):=\emptyset\). The adhesion of the decomposition is $\max_{t\in V(\tau)}|\sigma(t)|$.

For a node \(t\), let \(\tau_t\) be the subtree of \(\tau\) rooted at \(t\), and define
\[
    \gamma(t):=\bigcup_{u\in V(\tau_t)}\beta(u), \qquad \alpha(t):=\gamma(t)\setminus\sigma(t),
\]
and
\[
    G_t:=G[\gamma(t)]\setminus E_G(\sigma(t)).
\]
\end{definition}

\begin{definition}[Compact Tree Decomposition]
\label{def:compact_tree_decomposition}
A rooted tree decomposition \((\tau,\beta)\) is compact if, for every node \(t\) with \(\sigma(t)\ne\emptyset\), the graph \(G[\alpha(t)]\) is connected and $N_G(\alpha(t))=\sigma(t)$.
\end{definition}

\begin{definition}[Subtree Unbreakability]
\label{def:subtree_unbreakability}
A rooted tree decomposition \((\tau,\beta)\) satisfies \((q,s)\)-subtree unbreakability if, for every node \(t\in V(\tau)\), the bag \(\beta(t)\) is \((q,s)\)-edge-unbreakable in the subgraph \(G_t\).
\end{definition}

\begin{definition}[Unbreakable Decomposition] (Definition 3.3 in \cite{ALL+25}) \label{def:unbreakable_decomposition}
A $(q, s)$-unbreakable decomposition of $G$ is a rooted tree decomposition $(\tau, \beta)$ where each bag $\beta(t)$ is $(q, s)$-unbreakable in $G$. The decomposition admits the stronger subtree unbreakability property if each bag $\beta(t)$ is $(q, s)$-unbreakable in $G_t$.
\end{definition}

\begin{theorem}[Tree Cut-Sparsifier~\cite{ADK26}] \label{thm:adk_tree_sparsifier}
There is a randomized algorithm which, given an undirected unweighted graph \(G=(V,E)\), runs in time \(\widetilde O(|E|+|V|)\) and with high probability outputs a weighted hierarchical tree $\mathcal H=(V_{\mathcal H},E_{\mathcal H},w)$ with \(V\subseteq V_{\mathcal H}\), such that \(\mathcal H\) is a tree cut-sparsifier of \(G\) of quality $\alpha(n)=O(\log^2 n\log\log n)$. That is, for every \(X\subseteq V\), $|\delta_G(X)| \le \mincut_{\mathcal H}(X,V\setminus X) \le \alpha(n)\,|\delta_G(X)|$, where
\[
    \mincut_{\mathcal H}(X,V\setminus X)
    :=
    \min_{\substack{U\subseteq V_{\mathcal H}\\
                    X\subseteq U,\; V\setminus X\subseteq V_{\mathcal H}\setminus U}}
        w(\delta_{\mathcal H}(U)).
\]
Moreover, the tree is hierarchical: every edge \(e\in E_{\mathcal H}\) corresponds to a laminar side \(S_e\subseteq V\), and its weight is $w(e)=|\delta_G(S_e)|$. The hierarchy has depth \(O(\log n)\).
\end{theorem}

Let $\alpha:=\alpha(n)=O(\log^2 n\log\log n)$, and $\Lambda:=\lceil \alpha s\rceil$. We enlarge the hidden constant in \(\alpha\), if necessary, so that \(\alpha\ge 1\). Run \cref{thm:adk_tree_sparsifier} on \(G\), obtaining a hierarchical tree cut-sparsifier $\mathcal H=(V_{\mathcal H},E_{\mathcal H},w)$ of quality \(\alpha\). We identify each graph vertex \(v\in V\) with its corresponding singleton leaf in \(\mathcal H\).

\paragraph{Low tree edges.}
Call a tree edge \(e\in E_{\mathcal H}\) low if $w(e)\le \Lambda$. Let $F:=\{e\in E_{\mathcal H}: w(e)\le \Lambda\}$.

For \(e\in F\), let \(S_e\subseteq V\) be the side corresponding to \(e\) in the hierarchy. Either side induced by deleting \(e\) may be used, since the two sides define the same cut. Define its middle set in the original graph \(G\) by $\mu(e):=\End_G(\delta_G(S_e))$.

\begin{lemma}[Small Middle Sets]
\label{lem:low_edge_middle_small}
For every \(e\in F\), $|\mu(e)|\le 2w(e)\le 2\Lambda$.
\end{lemma}
\begin{proof}
The weight of \(e\) is $w(e)=|\delta_G(S_e)|$. Therefore
\[
    |\mu(e)| = |\End_G(\delta_G(S_e))| \le 2|\delta_G(S_e)| = 2w(e) \le 2\Lambda.
\]
\end{proof}

\paragraph{Raw Quotient.}
Delete all low edges \(F\) from \(\mathcal H\). Let \(\mathcal C\) be the set of connected components of $\mathcal H-F$. Contracting each component \(C\in\mathcal C\) gives a quotient tree $Q_{\rm raw}$. The edges of \(Q_{\rm raw}\) are in one-to-one correspondence with the low tree edges \(F\). For a component \(C\in\mathcal C\), let $V_C:=V\cap C$ be the set of original graph vertices whose singleton leaves lie in \(C\).

Define the raw bag
\[
    \beta_{\rm raw}(C) := V_C \cup \bigcup_{\substack{e\in F:\\ e\text{ incident with }C\text{ in }Q_{\rm raw}}} \mu(e).
\]

\begin{lemma}[Raw Quotient is a Tree Decomposition]
\label{lem:raw_td}
\((Q_{\rm raw},\beta_{\rm raw})\) is an ordinary tree decomposition of \(G\). Moreover, for every edge \(CD\in E(Q_{\rm raw})\) corresponding to the low tree edge \(e\in F\),
\[
    \beta_{\rm raw}(C)\cap \beta_{\rm raw}(D) \subseteq \mu(e).
\]
Consequently, the raw decomposition has adhesion at most \(2\Lambda\).
\end{lemma}
\begin{proof}
We first prove connectedness of vertex occurrences. Fix \(v\in V\), and let \(C_v\in V(Q_{\rm raw})\) be the component containing the singleton leaf \(v\). Let
\[
    T_v := \{C_v\} \cup \bigcup_{vu\in E(G)} V\bigl(P_{Q_{\rm raw}}(C_v,C_u)\bigr),
\]
where \(P_{Q_{\rm raw}}(C_v,C_u)\) denotes the unique path in \(Q_{\rm raw}\) from \(C_v\) to \(C_u\). We prove that \(T_v\) is exactly the set of raw bags containing \(v\).

First let \(C\in T_v\). If \(C=C_v\), then $v\in V_{C_v}\subseteq \beta_{\rm raw}(C_v)$. Otherwise, there is an edge \(vu\in E(G)\) such that \(C\) lies on the path \(P_{Q_{\rm raw}}(C_v,C_u)\). Then \(C\) is incident, along this path, with some low edge \(e\in F\). The corresponding tree edge lies on the unique path in \(\mathcal H\) between the singleton leaves \(v\) and \(u\), and therefore the side \(S_e\) separates \(v\) and \(u\). Thus $vu\in\delta_G(S_e)$, so $v\in\mu(e)\subseteq\beta_{\rm raw}(C)$.

Conversely, suppose \(v\in\beta_{\rm raw}(C)\). If \(v\in V_C\), then \(C=C_v\), so \(C\in T_v\). Otherwise, \(v\in\mu(e)\) for some low edge \(e\in F\) incident with \(C\) in \(Q_{\rm raw}\). By the definition of \(\mu(e)\), there is an edge \(vu\in E(G)\) crossing the cut \(S_e\). Equivalently, \(e\) lies on the path in \(\mathcal H\) between the singleton leaves \(v\) and \(u\), and therefore the quotient edge corresponding to \(e\) lies on \(P_{Q_{\rm raw}}(C_v,C_u)\). Since \(C\) is incident with this quotient edge, \(C\in T_v\).

Thus the raw bags containing \(v\) are exactly \(T_v\). Since \(T_v\) is a union of paths all containing \(C_v\), it is connected.

Next we prove edge coverage. Let \(uv\in E(G)\). If \(C_u=C_v\), then $u,v\in V_{C_u}\subseteq \beta_{\rm raw}(C_u)$, so the edge is covered. Otherwise, the path from \(C_u\) to \(C_v\) in \(Q_{\rm raw}\) contains at least one low edge \(e\in F\). The corresponding tree edge separates the singleton leaves \(u\) and \(v\), so $uv\in\delta_G(S_e)$, and therefore $u,v\in\mu(e)$. Since \(\mu(e)\) is included in the two bags incident with \(e\), the edge \(uv\) is covered.

It remains to prove the adhesion statement. Let \(CD\in E(Q_{\rm raw})\), and let \(e\in F\) be the corresponding low tree edge. Suppose $v\in \beta_{\rm raw}(C)\cap\beta_{\rm raw}(D)$. By the connectedness just proved, the raw bags containing \(v\) form a connected subtree of \(Q_{\rm raw}\). Since this subtree contains the adjacent nodes \(C\) and \(D\), it crosses the edge \(CD\). Equivalently, \(v\in\mu(e)\). Thus $\beta_{\rm raw}(C)\cap\beta_{\rm raw}(D)\subseteq\mu(e)$. By \cref{lem:low_edge_middle_small}, $|\beta_{\rm raw}(C)\cap\beta_{\rm raw}(D)| \le |\mu(e)| \le 2\Lambda$.
\end{proof}

\paragraph{Compression.}
We compress the raw quotient \(Q_{\rm raw}\) by repeatedly deleting empty leaves and suppressing empty degree-two nodes. Let the resulting tree be \(Q\). Thus every node of \(Q\) either contains at least one original graph vertex, or has degree at least \(3\).

Each edge \(\varepsilon=CD\in E(Q)\) represents a path in the raw quotient whose edges correspond to low edges of the hierarchy. For each endpoint \(C\) of \(\varepsilon\), we select the first raw low edge on this path incident with \(C\). For a retained node \(C\in V(Q)\), let \(I(C)\) be the set of selected first raw low edges over all compressed edges incident with \(C\). Define
\[
    \beta(C) := V_C \cup \bigcup_{e\in I(C)}\mu(e).
\]
The resulting compressed quotient is \((Q,\beta)\). The detailed compression construction and the proof of the following lemma are deferred to \cref{app:compression}.

\begin{lemma}[Compressed Quotient is a Tree Decomposition]
\label{lem:compressed_td}
We have $|V(Q)|+|E(Q)|=O(n)$ and $\beta(C)\subseteq \beta_{\rm raw}(C)$ for all $C$. \((Q,\beta)\) is a tree decomposition of \(G\) and its adhesion is at most \(2\Lambda\).
\end{lemma}

\begin{lemma}[Bag Unbreakability]
\label{lem:bags_unbreakable}
Every bag \(\beta(C)\), \(C\in V(Q)\), is $(2\Lambda,s)$-edge-unbreakable in $G$.
\end{lemma}
\begin{proof}
Let \(Z\subseteq V\) be any cut with $|\delta_G(Z)|\le s$. By the upper side of the cut-sparsifier guarantee of \cref{thm:adk_tree_sparsifier}, there exists a tree cut \(U\subseteq V_{\mathcal H}\) such that $U\cap V=Z$ and
\[
    w(\delta_{\mathcal H}(U)) \le \alpha |\delta_G(Z)| \le \alpha s \le \Lambda.
\]

Every tree edge inside a raw component of \(\mathcal H-F\) has weight strictly larger than \(\Lambda\). Since the total \(\mathcal H\)-capacity of \(\delta_{\mathcal H}(U)\) is at most \(\Lambda\), the cut \(U\) cannot cross any such high edge. Therefore every raw component \(C\in\mathcal C\) is monochromatic with respect to \(U\): either $C\subseteq U$ or $C\cap U=\emptyset$. In particular, every retained node \(C\in V(Q)\) is monochromatic.

Fix \(C\in V(Q)\). We prove that one side of \(Z\) contains at most \(2\Lambda\) vertices of \(\beta(C)\).

First suppose $C\subseteq U$, then $V_C\subseteq Z$. We claim that $|\beta(C)\setminus Z|\le 2\Lambda$. Take \(v\in\beta(C)\setminus Z\). Since \(V_C\subseteq Z\), we have \(v\notin V_C\). Hence \(v\in\mu(e)\) for some first raw low edge \(e\in I(C)\). In particular, by \(\beta(C)\subseteq\beta_{\rm raw}(C)\), $v\in\beta_{\rm raw}(C)$. Let \(C_v\) be the raw component containing the singleton leaf \(v\). Since \(v\notin Z=U\cap V\), and raw components are monochromatic, we have $C_v\cap U=\emptyset$. By \cref{lem:raw_td}, the raw bags containing \(v\) form a connected subtree of \(Q_{\rm raw}\). This subtree contains both \(C\) and \(C_v\). The raw path from \(C\) to \(C_v\) must cross the tree cut \(U\). Let \(h\) be a raw quotient edge on this path crossing \(U\). Since both endpoints of \(h\) lie in the raw occurrence subtree of \(v\), the adhesion statement of \cref{lem:raw_td} gives $v\in\mu(h)$. Moreover, $h\in F\cap \delta_{\mathcal H}(U)$. Consequently,
\[
    \beta(C)\setminus Z \subseteq \bigcup_{h\in F\cap\delta_{\mathcal H}(U)}\mu(h).
\]
Using \cref{lem:low_edge_middle_small}, we get
\[
    |\beta(C)\setminus Z| \le \sum_{h\in F\cap\delta_{\mathcal H}(U)}|\mu(h)| \le 2\sum_{h\in F\cap\delta_{\mathcal H}(U)}w(h) \le 2w(\delta_{\mathcal H}(U)) \le 2\Lambda.
\]

The case \(C\cap U=\emptyset\) is symmetric. Then $V_C\cap Z=\emptyset$, and the same argument gives $|\beta(C)\cap Z|\le 2\Lambda$. Thus for every \(s\)-cut \(Z\), one of the two sides contains at most \(2\Lambda\) vertices of \(\beta(C)\). Hence \(\beta(C)\) is \((2\Lambda,s)\)-edge-unbreakable in \(G\).
\end{proof}

\begin{lemma}[Decomposition Depth]
\label{lem:decomposition_depth}
The decomposition tree \(Q\) has depth \(O(\log n)\).
\end{lemma}
\begin{proof}
Let the rooted hierarchy \(\mathcal H\) have depth
\(d=O(\log n)\).  Every simple path in the raw quotient lifts to a
simple path in \(\mathcal H\): expand each contracted high-edge
component along the unique hierarchy path between the incident low
edges.  Consequently, every raw-quotient path has at most
\(\operatorname{diam}(\mathcal H)\le2d\) edges.

Deleting empty leaves and suppressing empty degree-two nodes cannot
increase the number of edges on a path between retained nodes.
Therefore every path in \(Q\) has length at most \(2d\).  Rooting \(Q\)
at any retained node gives depth \(O(\log n)\).
\end{proof}

\begin{lemma}[Output Size]
\label{lem:output_size}
The compressed decomposition \((Q,\beta)\) has
\[
    \sum_{C\in V(Q)}|\beta(C)|=O(\Lambda n)
\]
total bag incidences. Moreover,
\[
    \sum_{C\in V(Q)} |I(C)|=O(n),
\]
and therefore
\[
    \left|\bigcup_{C\in V(Q)} I(C)\right|=O(n).
\]
\end{lemma}
\begin{proof}
The sets \(V_C\), over retained nodes \(C\in V(Q)\), are pairwise disjoint subsets of \(V\). Therefore $\sum_{C\in V(Q)} |V_C| \le n$. Also \(|E(Q)|=O(n)\). Each compressed edge contributes at most two first raw low edges, one to each endpoint. Hence $\sum_{C\in V(Q)} |I(C)| \le 2|E(Q)|=O(n)$. This also gives
\[
    \left|\bigcup_{C\in V(Q)} I(C)\right|=O(n).
\]
By \cref{lem:low_edge_middle_small}, every middle set included in a compressed bag has size at most \(2\Lambda\). Therefore
\[
    \sum_{C\in V(Q)}|\beta(C)| \le n + \sum_{C\in V(Q)}\sum_{e\in I(C)}|\mu(e)| \le n+2\Lambda\sum_{C\in V(Q)}|I(C)| = O(\Lambda n).
\]
\end{proof}

\noindent \textbf{Note:} after constructing the compressed quotient decomposition, we apply a cleanup routine which compactifies the decomposition and suppresses redundant bags. Compactification is useful because it simplifies the interaction of a bag with its parent, making the dynamic programming of \cref{alg:fpt_min_k_cut} cleaner. Redundancy removal is useful because it keeps the decomposition size linear and prevents the DP from spending time on nodes whose states are identical to, or subsumed by, their parent states. The statement of this cleanup routine is given below, and its proof is deferred to \cref{app:compactification}. The theorem is stated for the final decomposition after this cleanup has been applied.

\begin{theorem}[Compactification and Redundancy Removal]
\label{thm:compactify_and_suppress}
Let \((\tau,\beta)\) be a rooted tree decomposition of \(G\), with root adhesion empty. Let
\[
    M:=|V(\tau)|+|E(G)|+\sum_{t\in V(\tau)}|\beta(t)|, \qquad \eta:=\max\{|\sigma_\tau(t)|:t\ne r\},
\]
with \(\eta:=0\) if \(\tau\) has only one node. There is an algorithm which runs in
\[
    \widetilde O(M+\eta n)
\]
time and outputs a compact rooted tree decomposition \((\widehat\tau,\widehat\beta)\) of \(G\) such that:
\begin{enumerate}
    \item every output bag is a subset of some input bag.
    \item every output adhesion is either empty or a subset of some input adhesion.
    \item no non-root output bag is contained in its parent bag.
    \item if the input tree has depth \(d\), then the output tree has depth at most \(d+1\).
    \item
    \[
        |V(\widehat\tau)|\le n+1, \qquad \sum_{x\in V(\widehat\tau)}|\widehat\beta(x)| = O((\eta+1)n).
    \]
\end{enumerate}
Consequently, if every input bag is \((q,s)\)-edge-unbreakable and every input adhesion has size at most \(\eta\), then every output bag is \((q,s)\)-edge-unbreakable and every output adhesion has size at most \(\eta\).
\end{theorem}

\begin{theorem}[Near-Linear Edge-Unbreakable Decomposition]
\label{thm:unbreakable_decomp}
Let \(G=(V,E)\) be an unweighted undirected multigraph on \(n\) vertices and \(m\) edges, and let \(s\ge 1\). There is a randomized algorithm which, with high probability, constructs a compact rooted tree decomposition \((\tau,\chi)\) of \(G\) with adhesion at most $2\Lambda = O(s\log^2 n\log\log n)$, such that every bag is $(2\Lambda,s)$-edge-unbreakable. Moreover, the output decomposition satisfies
\[
    |V(\tau)|=O(n), \qquad \sum_{C\in V(\tau)}|\chi(C)| = O(\Lambda n) = O(s n\log^2 n\log\log n).
\]
The rooted decomposition tree has depth $O(\log n)$. The construction time is
\[
    \widetilde O(m+\Lambda n) = \widetilde O(m+s n).
\]
\end{theorem}
\begin{proof}
Recall $\Lambda:=\lceil \alpha s\rceil$, where \(\alpha=O(\log^2 n\log\log n)\) is the quality parameter of \cref{thm:adk_tree_sparsifier}. The algorithm first constructs the hierarchical tree cut-sparsifier of \(G\), thresholds its tree edges at value \(\Lambda\), forms the corresponding raw quotient tree decomposition, and compresses it. Let \((Q,\beta)\) denote the compressed quotient.

The raw quotient and its compression are computed from \(\mathcal H\) in \(\widetilde O(|V(\mathcal H)|)=\widetilde O(m+n)\) time, which is absorbed by \(\widetilde O(m+\Lambda n)\).

Correctness as an ordinary tree decomposition, the adhesion bound, the depth bound, and the output-size bound for the compressed quotient follow from \cref{lem:compressed_td}, \cref{lem:decomposition_depth}, and \cref{lem:output_size}. In particular, \((Q,\beta)\) has adhesion at most \(2\Lambda\), depth \(O(\log n)\), and total bag volume \(O(\Lambda n)\).

Bag unbreakability for the compressed quotient follows from \cref{lem:bags_unbreakable}: every bag of \((Q,\beta)\) is \((2\Lambda,s)\)-edge-unbreakable.

We first materialize the middle sets needed by the compressed quotient. Let
\[
    I_0:=\bigcup_{C\in V(Q)} I(C)
\]
be the set of selected raw low edges whose middle sets appear in compressed quotient bags. By \cref{lem:output_size}, \(|I_0|=O(n)\).

We mark the edges of \(I_0\) in the hierarchy tree \(\mathcal H\). Root \(\mathcal H\) arbitrarily and identify each tree edge with its deeper endpoint. Thus a marked edge is represented by a marked non-root vertex. Preprocess \(\mathcal H\) for LCA queries using binary lifting (compute the $2^i$th ancestor of every vertex for all $i$), and compute a heavy-light decomposition (identify the largest-subtree child of every vertex). For each heavy path \(P\), order its vertices from top to bottom and store the marked vertices on \(P\), sorted by their positions on \(P\).

Given two nodes \(x,y\in V(\mathcal H)\), let \(z=\operatorname{lca}(x,y)\). The marked edges on the path \(x\leadsto y\) are exactly the marked child endpoints on the two upward paths \(x\leadsto z\) and \(y\leadsto z\), excluding \(z\) itself. Each upward path is decomposed by heavy-light decomposition into \(O(\log n)\) heavy-path intervals. On each interval, two binary searches in the sorted marked list for that heavy path report all marked vertices in the interval. Thus a path query takes $O(\log^2 n+r)$ time, where \(r\) is the number of reported marked edges. The preprocessing takes \(\widetilde O(|V(\mathcal H)|)\) time and space, plus \(O(|I_0|)\) storage for the marked vertices.

For every edge \(uv\in E(G)\), let \(\ell_u,\ell_v\) be the singleton leaves of \(\mathcal H\) corresponding to \(u\) and \(v\). We report the marked edges on the path between \(\ell_u\) and \(\ell_v\). For each reported marked edge \(e\), add \(u\) and \(v\) to \(\mu(e)\). This constructs the desired middle sets: for \(e\in I_0\), the edge \(uv\) is reported for \(e\) exactly when the two singleton leaves lie on opposite sides of \(e\) in \(\mathcal H\), equivalently exactly when \(uv\) crosses the cut \(S_e\) represented by \(e\). Therefore \(u\) and \(v\) are added to \(\mu(e)\) exactly for the selected cuts whose boundary contains \(uv\).

The total number of reports is
\[
    \sum_{e\in I_0} |\delta_G(S_e)| = \sum_{e\in I_0} w(e) \le \Lambda |I_0| = O(\Lambda n),
\]
because every selected edge \(e\in I_0\) is low. Therefore the total time spent by all path-reporting queries is $\widetilde O(m+\Lambda n)$, and the compressed quotient \((Q,\beta)\) can be built explicitly in $\widetilde O(m+\Lambda n)$ time.

Root \(Q\) as in \cref{lem:decomposition_depth}. Now run \cref{alg:compactify_and_suppress} on \((Q,\beta)\), and let the resulting decomposition be \((\tau,\chi)\). By \cref{lem:compactify_correctness}, \((\tau,\chi)\) is a compact rooted tree decomposition, every output bag is a subset of an input bag, every output adhesion is a subset of an input adhesion, no non-root output bag is contained in its parent bag, and the depth remains \(O(\log n)\). By \cref{cor:compactify_preserves_parameters}, every output bag remains \((2\Lambda,s)\)-edge-unbreakable and every output adhesion has size at most \(2\Lambda\).

Finally, by \cref{lem:compactify_time}, applied with \(M=O(m+\Lambda n)\) and \(\eta=2\Lambda\), the cleanup takes $\widetilde O(m+\Lambda n)$ time and outputs a decomposition satisfying
\[
    |V(\tau)|=O(n), \qquad
    \sum_{C\in V(\tau)}|\chi(C)|=O(\Lambda n).
\]
Combining the tree cut-sparsifier construction, middle-set construction, and cleanup gives total time
\[
    \widetilde O(m+\Lambda n) = \widetilde O(m+s n).
\]
\end{proof}

\section{FPT Algorithm for Min \texorpdfstring{$k$}{k}-Cut}
\label{sec:fpt_mincut}

We modify the FPT algorithm of Lokshtanov, Saurabh, and Surianarayanan~\cite{LSS22}.  First, we use the
near-linear edge-unbreakable decomposition of
\cref{thm:unbreakable_decomp}, whose adhesion and unbreakability
parameters are
\[
    \Lambda=O(s\log^2 n\log\log n).
\]
Second, virtual trees let us compute the projection of a spanning tree
onto a bag efficiently.  Third, the large-bag dynamic program uses
randomized subset families to guess the relevant projected tree edges and
small components.  Finally, an approximate dual
\(k\)-cut packing supplies a random tree that respects an optimum cut with
good probability.  These changes give near-linear
dependence on the graph size and reduce the dependence on \(s\) to
\(s^{6k-6}\).

Throughout this section \(G=(V,E)\) is an unweighted multigraph, \(|V|=n\), \(2\le k\le n\), and \(s\) is a parameter satisfying \(\lambda_k(G)\le s\). Storing the minimizing choices returns the corresponding cut with the same asymptotic running time.

\begin{definition}[Respecting spanning tree]
\label{def:respecting_tree}
Let \(T\) be a spanning tree of \(G\), and let \(A\subseteq E(G)\) be an edge set. We say that \(T\) \(h\)-respects \(A\) if
\[
    |E(T)\cap A|\le h.
\]
If \(\Pi\) is a partition of \(V(G)\), we say that \(T\) \(h\)-respects \(\Pi\) when \(T\) \(h\)-respects \(\delta_G(\Pi)\).
\end{definition}

\begin{definition}[Nagamochi--Ibaraki sparsification~\cite{NI92}] \label{def:ni_sparsification}
For an undirected graph $G=(V,E)$, a Nagamochi--Ibaraki forest decomposition is a sequence of edge-disjoint spanning forests
\[
    F_1,F_2,\ldots
\]
where $F_i$ is a maximal spanning forest of
\[
    G\setminus \bigcup_{j<i} F_j .
\]
For an integer $r\ge 1$, the union
\[
    G_r := \bigcup_{i=1}^r F_i
\]
has at most $r(|V|-1)$ edges. Moreover, for every cut $X\subseteq V$,
\[
    |\delta_{G_r}(X)| \ge \min\{|\delta_G(X)|,r\}.
\]
In particular, every cut of value at most $r$ is preserved exactly in $G_r$, and every cut of value greater than $r$ has value at least $r$ in $G_r$. For a prescribed threshold $r$, the certificate $G_r$ can be computed in $O(|E|)$ time.
\end{definition}

The edge-unbreakable decomposition used below has adhesion and unbreakability parameter
\[
    \Lambda=O(s\log^2 n\log\log n).
\]
We keep \(\Lambda\) explicit throughout the algorithm and substitute this value only in the final theorem.

\subsection{Preliminaries and Decomposition}

Let $m_0:=|E(G)|$ denote the number of edges in the input graph. We first apply Nagamochi--Ibaraki sparsification (\cref{def:ni_sparsification}, \cite{NI92}) with threshold \(s+1\). Thus all cuts of value at most \(s+1\) are preserved exactly, cuts of value larger than \(s+1\) remain larger than \(s+1\), and the resulting graph has \(m=O(sn)\) edges, this takes $O(m_0)$ time. We replace \(G\) by this sparse graph for the remainder of the algorithm.

If \(G\) has at least \(k\) connected components, then \(\lambda_k(G)=0\).

If \(G\) has \(1<h<k\) connected components
\(G_1,\ldots,G_h\), reduce to the connected case as follows.  For each
component \(G_j\) and each
\[
    1\le p\le\min\{k-h+1,|V(G_j)|\},
\]
run the connected algorithm as a truncated solver with threshold \(s\).
Every recursive parameter is at most \(k-1\).

As in \cref{lem:zero_capacity_reduction_kcut}, an optimum global
\(k\)-cut can be normalized so that the total number of induced
component-parts is exactly \(k\): while there are more than \(k\), merge
two induced parts inside one component.  This cannot increase the cut
value.  Therefore the optimum is obtained from
\[
    \min_{\substack{p_1+\cdots+p_h=k\\p_j\ge1}}
        \sum_{j=1}^h\lambda_{p_j}(G_j),
\]
which is evaluated by a dynamic program.

The connected-case calls are amplified so that all calls used by this
dynamic program succeed with high probability.  Each component occurs
for at most \(k\) requested parameters, so the aggregate input volume is
\(O(kn)\) vertices and \(O(km)\) edges.  This additional factor is
absorbed by \(k^{O(k)}\).  Hence, from now on, assume that \(G\) is
connected. Hence, from now on, assume that \(G\) is connected.

We then run the near-linear edge-unbreakable decomposition algorithm with cut parameter \(s\). By \cref{thm:unbreakable_decomp}, this computes a compact rooted tree decomposition \((\tau,\beta)\) in time
\[
    \widetilde O(m+\Lambda n)=\widetilde O(\Lambda n),
\]
because \(m=O(sn)\) and \(\Lambda\ge s\). Set
\[
    q:= 2 \Lambda,
\]
to be the unbreakability parameter from \cref{thm:unbreakable_decomp}. The decomposition satisfies:
\begin{enumerate}
    \item every adhesion has size \(O(\Lambda)\)
    \item every bag \(\beta(t)\) is \((q,s)\)-edge-unbreakable
    \item \(|V(\tau)|=O(n)\)
    \item \(\sum_{t\in V(\tau)}|\beta(t)|=O(\Lambda n)\)
    \item \(\sum_{t\in V(\tau)}|\sigma(t)|=O(\Lambda n)\)
\end{enumerate}
Here \(\sigma(t)\) denotes the adhesion between \(t\) and its parent, with \(\sigma(r)=\emptyset\) for the root \(r\), and \(\gamma(t)\) denotes the union of bags in the subtree of \(\tau\) rooted at \(t\).

\subsection{Tree-Feasible States}

Fix a spanning tree \(T\) of \(G\). The dynamic-programming tables are indexed only by adhesion states that can arise from cutting at most \(2k-2\) edges of \(T\).

A labeled partition of a set \(X\) is a function \(\phi:X\to[k]\). Its underlying unlabeled partition is the partition into the nonempty preimages of labels. Let \(\lambda(\phi):=\phi(X)\subseteq[k]\) be the set of labels used by \(\phi\). The restriction of \(\phi\) to \(Y\subseteq X\) is denoted by \(\phi|_Y\); if \(Y=\emptyset\), this is the unique empty labeled partition \(\phi_\emptyset\).

\begin{definition}[Tree-Feasible Labeled Partition]
Let \(T'\) be a tree. A labeled partition \(\phi\) of \(V(T')\) is \(c\)-admissible with respect to \(T'\) if its underlying unlabeled partition has at most \(c\) edges of \(T'\) with endpoints in distinct parts.

For \(X\subseteq V(T')\), a labeled partition \(\phi_X:X\to[k]\) is \(T'\)-feasible if it is the restriction to \(X\) of a \((2k-2)\)-admissible labeled partition of \(V(T')\). Let \(\mathcal L_{T'}^X\) denote the family of all such labeled partitions. If \(X=\emptyset\), we include \(\phi_\emptyset\).
\end{definition}

For \(X\subseteq V(T)\), let \(T_X:=\operatorname{proj}(T,X)\) be the tree obtained from \(T\) by repeatedly deleting leaves outside \(X\) and smoothing degree-\(2\) vertices outside \(X\).

\begin{definition}[Virtual Tree]
\label{def:virtual_tree}
Let \(\mathcal T\) be a rooted tree, and let \(S\subseteq V(\mathcal T)\). The virtual tree, or auxiliary tree, on \(S\) is the minimal tree \(\mathcal T(S)\) whose vertex set contains \(S\) and is closed under lowest common ancestors in \(\mathcal T\). Its edges correspond to maximal paths of \(\mathcal T\) between consecutive retained vertices, and therefore preserve the ancestor-descendant relations among vertices of \(S\) and their LCAs.
\end{definition}

\begin{lemma}[Virtual Tree and Labeled State Bound]
\label{lem:virtual_tree_and_states}
After one \(O(n\log n)\) preprocessing step on \(T\), the tree \(T_X\) can be computed in \(O(1+|X|\log |X|)\) time. Moreover,
\[
    \mathcal L_T^X=\mathcal L_{T_X}^X,
\]
\[
    |V(T_X)|\le \max\{2|X| - 1, 1\},
\]
and
\[
    |\mathcal L_T^X| \le k^{O(k)}\max\{|X|,1\}^{2k-2}.
\]
In particular, for every adhesion \(\sigma(t)\),
\[
    |\mathcal L_T^{\sigma(t)}| \le k^{O(k)}\Lambda^{2k-2}.
\]
\end{lemma}

\begin{proof}
The virtual tree on \(X\) is computed by sorting the vertices of \(X\) by Euler-tour order and adding the LCAs of consecutive vertices. This takes \(O(n\log n)\) preprocessing and \(O(1+|X|\log |X|)\) query time, and results in at most $\max\{2|X|-1,1\}$ vertices \cite{BF00}. Thus deleting remaining non-\(X\) leaves and smoothing remaining non-\(X\) degree-\(2\) vertices takes $O(|X|)$ time.

Deleting a leaf outside \(X\) has no effect on the induced labels on \(X\). Smoothing a degree-\(2\) vertex \(v\notin X\) with neighbors \(a,b\) preserves the induced labels on \(X\): a deleted edge among \(av,vb\) separates the two sides of the path if and only if the smoothed edge \(ab\) is deleted. This transformation never increases the number of deleted tree edges, and it is reversible by replacing a deleted smoothed edge by one of the two original edges. Applying these operations throughout the construction of \(T_X\) gives \(\mathcal L_T^X=\mathcal L_{T_X}^X\).

It remains to count feasible labelings on \(T_X\). Choose \(j\le 2k-2\) deleted edges of \(T_X\). This gives \(j+1\) components. Then assign labels from \([k]\) to the resulting components, allowing components with the same label to be merged. Thus the number of possibilities is at most
\[
    \sum_{j=0}^{2k-2}\binom{\max\{2|X|-1,1\}}{j}k^{j+1} \le k^{O(k)}\max\{|X|,1\}^{2k-2}.
\]
Since every adhesion has size \(O(\Lambda)\), the adhesion-state bound follows.
\end{proof}

\subsection{Dynamic Programming for One Tree}

Fix the decomposition \((\tau,\beta)\) and a spanning tree \(T\). For each vertex \(v\in V(G)\), let \(\operatorname{top}(v)\) be the minimum-depth node whose bag contains \(v\). For an edge \(uv\in E(G)\), the nodes \(\operatorname{top}(u)\) and \(\operatorname{top}(v)\) are comparable, and the minimum-depth bag containing both endpoints is the deeper of them. We assign \(uv\) to this node \(\mu(uv)\), and define
\[
    E_t^\circ:=\{e\in E(G):\mu(e)=t\}.
\]
The sets \(E_t^\circ\) partition \(E(G)\). Let \(E_{\le t}^\circ:=\bigcup_{u\in V(\tau_t)}E_u^\circ\).

For a labeled partition \(\phi\) of a vertex set containing the endpoints of an edge set \(F\), let \(\operatorname{cost}_F(\phi)\) be the number of edges in \(F\) whose endpoints receive distinct labels. Write \(\operatorname{cost}_t(\phi):=\operatorname{cost}_{E_t^\circ}(\phi)\).

For every node \(t\), every \(\phi_{\sigma(t)}\in\mathcal L_T^{\sigma(t)}\), and every label set \(L\subseteq[k]\) with (letting $\lambda$ denote used labels) \(\lambda(\phi_{\sigma(t)})\subseteq L\), the algorithm maintains an entry
\[
    f_t(\phi_{\sigma(t)},L)\in\{0,1,\ldots,s\}\cup\{\infty\}.
\]
A finite value is stored only together with a witness labeled partition \(\phi:\gamma(t)\to[k]\) such that
\[
    \phi|_{\sigma(t)}=\phi_{\sigma(t)}, \qquad \lambda(\phi)=L,
\]
and the stored value is exactly \(\operatorname{cost}_{E_{\le t}^\circ}(\phi)\). Entries larger than \(s\) are truncated to \(\infty\).

A transition at node \(t\) consists of a labeled bag partition \(\psi:\beta(t)\to[k]\) and, for every child \(u\) of \(t\), a label set \(L_u\subseteq[k]\) such that
\[
    \psi|_{\sigma(u)}\in\mathcal L_T^{\sigma(u)} \quad\text{and}\quad f_u(\psi|_{\sigma(u)},L_u)<\infty.
\]
The transition glues the child witnesses to the bag labeling \(\psi\) along the equal labeled adhesion states. It has label set
\[
    L=\lambda(\psi)\cup\bigcup_{u\in\operatorname{children}(t)}L_u
\]
and value
\[
    \operatorname{cost}_t(\psi) + \sum_{u\in\operatorname{children}(t)} f_u(\psi|_{\sigma(u)},L_u).
\]
The transition contributes to \(f_t(\psi|_{\sigma(t)},L)\), provided \(\psi|_{\sigma(t)}\in\mathcal L_T^{\sigma(t)}\). At the root \(r\), the answer for this tree is \(f_r(\phi_\emptyset,[k])\).

\subsection{Small and Giant Parts}

In this subsection and the next two subsections, fix a labeled global optimum \(k\)-cut \(\phi^\star:V(G)\to[k]\) of value at most \(s\), and fix a node \(t\). Let \(\Pi_t^\star\) be the underlying unlabeled partition induced by \(\phi^\star|_{\beta(t)}\).

A bag \(t\) is large if \(|\beta(t)|>kq\). If \(t\) is large, then, since \(\beta(t)\) is \((q,s)\)-edge-unbreakable and \(\phi^\star\) has total cut value at most \(s\), at most one part of \(\Pi_t^\star\) contains more than \(q\) vertices. Since there are at most \(k\) parts and \(|\beta(t)|>kq\), such a part exists. We call it the giant part. Let \(S_t\) be the union of all non-giant parts of \(\Pi_t^\star\). Then
\[
    |S_t|\le (k-1)q=O(k\Lambda).
\]

Assume now that \(T\) crosses \(\phi^\star\) in at most \(2k-2\) edges. Since \(T\) crosses \(\phi^\star\) in at most \(2k-2\) edges, delete from \(T_{\beta(t)}\) every projected edge whose corresponding path in \(T\) contains an edge crossing \(\phi^\star\). The projected-edge paths are edge-disjoint in \(T\), so at most \(2k-2\) projected edges are deleted. Every remaining projected edge corresponds to a path in \(T\) whose vertices all have the same \(\phi^\star\)-label. Therefore the component partition of the resulting projected tree refines \(\Pi_t^\star\). Let \(C_t^\star\subseteq E(T_{\beta(t)})\) be a set of at most \(2k-2\) projected-tree edges such that the component partition of \(T_{\beta(t)}-C_t^\star\) refines \(\Pi_t^\star\). Let \(\mathcal R_t^\star\) be this component partition. Let \(S_t^V\) be the union of components of \(\mathcal R_t^\star\) that contain vertices of \(S_t\), and let \(S_t^E\) be the set of projected-tree edges in \(E(T_{\beta(t)})\setminus C_t^\star\) whose endpoints both lie in \(S_t^V\).

\begin{lemma}[Small-Side Projected Size]
\label{lem:small_side_projected_size}
If \(t\) is large, then
\[
    |S_t^V|=O(k\Lambda), \qquad |S_t^E|=O(k\Lambda).
\]
\end{lemma}

\begin{proof}
In \(T_{\beta(t)}\), every vertex outside \(\beta(t)\) has degree at least \(3\). After deleting \(C_t^\star\), only vertices incident with deleted edges can become leaves or degree-\(2\) vertices outside \(\beta(t)\).

Let \(\mathcal P_{\mathrm{small}}\) be the set of components of \(T_{\beta(t)}-C_t^\star\) that meet \(S_t\). For each \(P\in\mathcal P_{\mathrm{small}}\), let \(\partial_{C_t^\star}P\) be the set of deleted projected-tree edges incident with \(P\). Counting degrees in the tree gives
\[
    |P|\le 2\bigl(|P\cap\beta(t)|+|\partial_{C_t^\star}P|\bigr).
\]
Summing over \(P\in\mathcal P_{\mathrm{small}}\),
\[
    |S_t^V|\le 2|S_t|+4|C_t^\star|=O(k\Lambda).
\]
The relevant projected edges form a forest on \(S_t^V\), so \(|S_t^E|=O(k\Lambda)\).
\end{proof}

\subsection{Random Separating Families}

Let
\[
    L_\Lambda:=\left\lceil (k+1)\log\!\bigl(k(n+\Lambda+2)\bigr)\right\rceil .
\]

\begin{lemma}[Random Separating Family]
\label{lem:random_separating_family}
Let \(\mathcal U\) be a finite ordered universe, and let \(a,b\ge 1\). Put \(p:=a/(a+b)\), and let
\[
    R:=\left\lceil c_0\left(\frac{e(a+b)}{a}\right)^a L_\Lambda \right\rceil
\]
for a sufficiently large absolute constant \(c_0\). Generate \(R\) independent random subsets \(H_1,\ldots,H_R\subseteq\mathcal U\) by including each element independently with probability \(p\).

Then, for every fixed pair of disjoint sets \(A,B\subseteq\mathcal U\) with \(|A|\le a\) and \(|B|\le b\), w.h.p. some \(H_j\) satisfies \(A\subseteq H_j\) and \(H_j\cap B=\emptyset\).
\end{lemma}
\begin{proof}
For one random set \(H\),
\[
    \Pr[A\subseteq H,\ H\cap B=\emptyset] =p^{|A|}(1-p)^{|B|} \ge p^a(1-p)^b \ge \left(\frac{a}{e(a+b)}\right)^a .
\]
Thus the probability that none of the \(R\) samples separates \(A\) from \(B\) is at most
\[
    \exp\!\left(-R\left(\frac{a}{e(a+b)}\right)^a\right),
\]
which is small enough for union bounds over \(k^{O(k)}\Lambda^{O(k)}n\) generated objects, by choosing \(c_0\) sufficiently large.
\end{proof}

\subsection{First-Level Guess}
Fix a large bag \(t\), and keep the fixed optimum cut \(\phi^\star\) and respecting tree \(T\) from above. Let \(\mathcal H_1(t)\) be the family obtained from \cref{lem:random_separating_family} with universe \(\mathcal U=E(T_{\beta(t)})\) and parameters
\[
    a=2k-2, \qquad b=c_1k\Lambda,
\]
where \(c_1\) is a sufficiently large absolute constant. Then
\[
    R_1:=|\mathcal H_1(t)| = k^{O(k)}\Lambda^{2k-2}L_\Lambda.
\]

\begin{definition}[Good First-Level Guess]
A first-level guess \(\widehat C\in\mathcal H_1(t)\) is good for \(\phi^\star\) at \(t\) if
\[
    C_t^\star\subseteq\widehat C \qquad\text{and}\qquad \widehat C\cap S_t^E=\emptyset.
\]
\end{definition}

\begin{lemma}[Existence of a Good First-Level Guess]
\label{lem:good_first_guess}
W.h.p., for every large bag \(t\), some \(\widehat C\in\mathcal H_1(t)\) is good for \(\phi^\star\) at \(t\).
\end{lemma}

\begin{proof}
Apply \cref{lem:random_separating_family} with \(A=C_t^\star\) and \(B=S_t^E\). We have \(|A|\le 2k-2\), and \cref{lem:small_side_projected_size} gives \(|B|=O(k\Lambda)\). Choosing \(c_1\) large enough makes \(|B|\le c_1k\Lambda\). A union bound over the \(O(n)\) bags proves the claim.
\end{proof}

For a good \(\widehat C\), let \(\mathcal R(\widehat C)\) be the set of connected components of \(T_{\beta(t)}-\widehat C\) whose projection to \(\beta(t)\) is nonempty. Empty projected components are ignored throughout. Let \(\mathcal S_t\subseteq\mathcal R(\widehat C)\) be the set of components that intersect \(S_t^V\).

\begin{lemma}[Number of Small-Side Components]
\label{lem:small_side_components}
If \(\widehat C\) is good, then \(|\mathcal S_t|\le 2k-2\).
\end{lemma}

\begin{proof}
Since \(\widehat C\) is good, it contains \(C_t^\star\) and is disjoint from \(S_t^E\). Thus passing from \(T_{\beta(t)}-C_t^\star\) to \(T_{\beta(t)}-\widehat C\) can only delete additional projected-tree edges outside the small side. In particular, no component of \(T_{\beta(t)}-C_t^\star\) that meets \(S_t^V\) is split by the additional deletions. Therefore the components of \(T_{\beta(t)}-\widehat C\) that meet \(S_t^V\) are exactly the components of \(T_{\beta(t)}-C_t^\star\) that meet \(S_t^V\).

Deleting \(C_t^\star\) creates at most \(2k-1\) components. Since the giant part is nonempty and disjoint from the small side, at least one of these components belongs to the giant side. Hence at most \(2k-2\) components meet the small side.
\end{proof}

\subsection{Second-Level Guess and Nice Decompositions}

For fixed \(t\) and \(\widehat C\), define the auxiliary graph \(\mathcal A_t(\widehat C)\) as follows. Its vertices are the components in \(\mathcal R(\widehat C)\). Two components are adjacent if either there is an owned edge \(xy\in E_t^\circ\) with endpoints in the two components, or some adhesion \(\sigma(u)\), for \(u\in\operatorname{children}(t)\cup\{t\}\), intersects both components. Let \(N_{\mathcal A_t(\widehat C)}(\mathcal S_t)\) be the set of auxiliary neighbors of \(\mathcal S_t\) outside \(\mathcal S_t\).

This auxiliary graph captures all relevant interactions inside \(\beta(t)\). Indeed, if an edge \(xy\in E(G[\beta(t)])\) is not owned at \(t\), then its owner is a proper ancestor of \(t\), and the connectedness axiom implies \(x,y\in\sigma(t)\).

\begin{lemma}[Auxiliary-Neighborhood Bound]
\label{lem:aux_neighborhood_bound}
For the fixed optimum cut \(\phi^\star\) and every good first-level guess \(\widehat C\),
\[
    |N_{\mathcal A_t(\widehat C)}(\mathcal S_t)|=O(k\Lambda^2).
\]
\end{lemma}

\begin{proof}
Owned-edge neighbors are charged to cut edges of \(\phi^\star\). If an owned edge has one endpoint in a component of \(\mathcal S_t\) and the other endpoint outside \(\mathcal S_t\), then its endpoints lie in different parts of the optimum bag partition. Since \(\phi^\star\) has value at most \(s\), owned-edge neighbors contribute at most \(s\le\Lambda\) components.

Parent-adhesion neighbors contribute \(O(\Lambda)\), since \(|\sigma(t)|=O(\Lambda)\).

It remains to bound neighbors arising from child adhesions. Call a child \(u\) affected if \(\sigma(u)\) intersects a component of \(\mathcal S_t\) and also intersects a component outside \(\mathcal S_t\). Then \(\sigma(u)\) contains vertices on both sides of the optimum bag partition. Since the decomposition is compact, \(G[\alpha(u)]\) is connected and \(N_G(\alpha(u))=\sigma(u)\). Hence at least one edge incident with \(\alpha(u)\) is cut by \(\phi^\star\). The sets \(\alpha(u)\) over different children of \(t\) are disjoint, so these charges are distinct. Thus at most \(s\le\Lambda\) children are affected.

Each affected adhesion has size \(O(\Lambda)\), and therefore intersects at most \(O(\Lambda)\) components of \(\mathcal R(\widehat C)\). Hence child adhesions contribute \(O(\Lambda^2)\) auxiliary neighbors. Altogether,
\[
    |N_{\mathcal A_t(\widehat C)}(\mathcal S_t)|=O(\Lambda^2) \subseteq O(k\Lambda^2).
\]
\end{proof}

For a good first-level guess \(\widehat C\), choose an arbitrary component \(R_t^{\mathrm{cen}}\in\mathcal R(\widehat C)\setminus\mathcal S_t\) whose projection to \(\beta(t)\) is contained in the giant part. Such a component exists: the giant part contains a vertex of \(\beta(t)\), and \(\widehat C\supseteq C_t^\star\), so every component of \(T_{\beta(t)}-\widehat C\) projects into a single part of \(\Pi_t^\star\).

Let \(\mathcal H_2(t,\widehat C)\) be the family obtained from \cref{lem:random_separating_family} with universe \(\mathcal U=\mathcal R(\widehat C)\) and parameters
\[
    a=2k-2, \qquad b=c_2k\Lambda^2,
\]
for a sufficiently large absolute constant \(c_2\). Then
\[
    R_2:=|\mathcal H_2(t,\widehat C)| = k^{O(k)}\Lambda^{4k-4}L_\Lambda.
\]

\begin{definition}[Good Second-Level Guess]
A second-level guess \(\widehat{\mathcal S}\in\mathcal H_2(t,\widehat C)\) is good for \(\phi^\star\) at \(t\) if
\[
    \mathcal S_t\subseteq\widehat{\mathcal S}
\]
and
\[
    \widehat{\mathcal S}\cap \bigl(N_{\mathcal A_t(\widehat C)}(\mathcal S_t) \cup\{R_t^{\mathrm{cen}}\}\bigr) = \emptyset.
\]
\end{definition}

\begin{lemma}[Existence of a Good Second-Level Guess]
\label{lem:good_second_guess}
Condition on the generated first-level families. W.h.p., for every large bag \(t\) and every good first-level sample \(\widehat C\in\mathcal H_1(t)\), some \(\widehat{\mathcal S}\in\mathcal H_2(t,\widehat C)\) is good for \(\phi^\star\) at \(t\).
\end{lemma}

\begin{proof}
For fixed \(t\) and \(\widehat C\), apply \cref{lem:random_separating_family} with
\[
    A=\mathcal S_t, \qquad B=N_{\mathcal A_t(\widehat C)}(\mathcal S_t)\cup\{R_t^{\mathrm{cen}}\}.
\]
\cref{lem:small_side_components} gives \(|A|\le 2k-2\), and \cref{lem:aux_neighborhood_bound} gives \(|B|=O(k\Lambda^2)\). Choosing \(c_2\) sufficiently large gives \(|B|\le c_2k\Lambda^2\). A union bound over all large bags and first-level samples proves the claim.
\end{proof}

For this subsection, write
\[
    \mathsf{Adh}(t):=\operatorname{children}(t)\cup\{t\}.
\]

\begin{definition}[Nice Decomposition]
\label{def:nice_decomposition}
A triple
\[
    D=(\Pi^{\mathrm{nice}}_{\beta(t)},\Theta_{\beta(t)},O)
\]
is a nice decomposition of \(\beta(t)\) if \(\Theta_{\beta(t)}\) refines \(\Pi^{\mathrm{nice}}_{\beta(t)}\), \(O\) is one part of \(\Pi^{\mathrm{nice}}_{\beta(t)}\), and:
\begin{enumerate}
    \item \(\Theta_{\beta(t)}|_O=\{O\}\)
    \item for every part \(P\in\Pi^{\mathrm{nice}}_{\beta(t)}\setminus\{O\}\), the restricted partition \(\Theta_{\beta(t)}|_P\) has at most \(2k-1\) parts
    \item no edge of \(G[\beta(t)]\) has endpoints in two distinct parts of \(\Pi^{\mathrm{nice}}_{\beta(t)}\setminus\{O\}\)
    \item no adhesion \(\sigma(u)\), for \(u\in\mathsf{Adh}(t)\), intersects two distinct parts of \(\Pi^{\mathrm{nice}}_{\beta(t)}\setminus\{O\}\)
\end{enumerate}
\end{definition}

For a fixed pair \((t,\widehat C)\), we use the quotient auxiliary graph on the components of \(\mathcal R(\widehat C)\) constructed in \cref{lem:fast_compatibility}. A component of \(\mathcal R(\widehat C)\) is called \emph{heavy} if its degree in this quotient auxiliary graph is larger than $\Delta:=c k\Lambda^2$ for a sufficiently large constant \(c\). Heavy components are determined during the preprocessing for \((t,\widehat C)\), independently of the second-level guess \(\widehat{\mathcal S}\). In the algorithm below, a component is \emph{selected} if it belongs to \(\widehat{\mathcal S}\).

\begin{algorithm}[H]
\caption{\textsc{NiceDecomposition}$(t,T,\widehat C,\widehat{\mathcal S})$}
\label{alg:nice_decomposition}
\begin{algorithmic}[1]
\State Let \(\mathcal R(\widehat C)\) be the nonempty projected components of \(T_{\beta(t)}-\widehat C\).
\State Start with one atom for each component of \(\mathcal R(\widehat C)\).
\State Merge every component not in \(\widehat{\mathcal S}\) into a center atom \(\Theta_0\).
\State Send every selected heavy component to \(\Theta_0\).
\State In the induced auxiliary graph on the remaining selected components, compute the connected components of size at most \(2k-1\) using \cref{lem:fast_compatibility}.
\State Merge every remaining selected component not lying in one of these small auxiliary connected components into \(\Theta_0\).
\If{\(\Theta_0=\emptyset\)}
    \State \Return \(\bot\).
\EndIf
\State Let \(\Theta'\) be the resulting atom partition.
\State Form \(\Pi^{\mathrm{nice}}\) by merging each small auxiliary connected component into one non-center part.
\State Add the center atom \(\Theta_0\) to \(\Pi^{\mathrm{nice}}\), and set \(O\gets\Theta_0\).
\State Scan the component--adhesion incidence lists of the non-center components. Using an array indexed by \(\mathsf{Adh}(t)\) with timestamps, record for each touched adhesion the unique non-center part it touches.
\State Project \(\Theta'\), \(\Pi^{\mathrm{nice}}\), and \(O\) onto \(\beta(t)\), remove empty parts, and return
\[
    D=(\Pi^{\mathrm{nice}}_{\beta(t)},\Theta_{\beta(t)},O)
\]
together with the recorded adhesion-incidence data.
\end{algorithmic}
\end{algorithm}

A returned value \(\bot\) is ignored by the dynamic program. The timestamped array in Line~12 only avoids clearing an array over all of \(\mathsf{Adh}(t)\); only adhesions actually touched by the incidence scan are initialized.

The recorded adhesion-incidence data is well-defined. If an adhesion \(\sigma(u)\) touched two distinct non-center parts, then two selected components in those parts would be adjacent in $\mathcal A_t(\widehat C)[\widehat{\mathcal S}]$, and hence would lie in the same auxiliary connected component. Thus a fixed adhesion touches at most one non-center part. Since every non-center part contains at most \(2k-1\) components of \(\mathcal R(\widehat C)\), each adhesion appears in the scanned incidence lists of only \(O(k)\) non-center components. Therefore the incidence scan costs
\[
    k^{O(1)}(1+\deg_\tau^+(t))
\]
time.

\begin{lemma}[Fast Compatibility Structure]
\label{lem:fast_compatibility}
Fix a node \(t\), a first-level guess \(\widehat C\), and a second-level guess
\[
    \widehat{\mathcal S}\subseteq\mathcal R(\widehat C).
\]
After preprocessing \((t,\widehat C)\), the small connected components of \(\mathcal A_t(\widehat C)[\widehat{\mathcal S}]\) needed by \cref{alg:nice_decomposition}, together with the component--adhesion incidence lists needed to assign adhesions to non-center parts, can be computed in time
\[
    O\!\left( k^{O(1)} (1+\deg_\tau^+(t)+|\widehat{\mathcal S}|\Lambda^2) L_\Lambda^{O(1)} \right),
\]
after an additional preprocessing cost of
\[
    O\!\left( k^{O(1)} \left( |\beta(t)|+|E_t^\circ| + \sum_{u\in\mathsf{Adh}(t)} |\sigma(u)|^2 \right) L_\Lambda^{O(1)} \right).
\]
The preprocessing also stores owned-edge multiplicities between components of \(\mathcal R(\widehat C)\). Moreover, for a good pair \((\widehat C,\widehat{\mathcal S})\), no true small-side component is discarded by the heavy-component rule or by the size cutoff.
\end{lemma}

\begin{proof}
First compute the nonempty projected components \(\mathcal R(\widehat C)\) of \(T_{\beta(t)}-\widehat C\). For every \(u\in\mathsf{Adh}(t)\), list the distinct components of \(\mathcal R(\widehat C)\) intersecting \(\sigma(u)\), and insert \(u\) into the incidence list of each such component. These same lists are used to add the adhesion-clique edges in \(\mathcal A_t(\widehat C)\).

Build the quotient auxiliary graph on \(\mathcal R(\widehat C)\). For each owned edge \(xy\in E_t^\circ\), add its multiplicity to the pair of components containing \(x\) and \(y\), if these components are distinct. For each adhesion \(\sigma(u)\), add a clique on the distinct components it intersects. This costs \(O(|\sigma(u)|^2)\) for adhesion \(u\). Duplicate quotient edges are removed by hashing, while owned-edge multiplicities are retained.

After the quotient graph is built, mark as heavy every component whose quotient degree is larger than $\Delta=c k\Lambda^2$, where \(c\) is a sufficiently large constant. This marking is independent of \(\widehat{\mathcal S}\). For a good pair, no true small-side component is heavy. Indeed, \cref{lem:aux_neighborhood_bound} gives
\[
    |N_{\mathcal A_t(\widehat C)}(\mathcal S_t)|=O(k\Lambda^2).
\]
For every \(P\in\mathcal S_t\),
\[
    \deg_{\mathcal A_t(\widehat C)}(P) \le |\mathcal S_t|-1+ |N_{\mathcal A_t(\widehat C)}(\mathcal S_t)| = O(k\Lambda^2).
\]
Choosing \(c\) large enough therefore ensures that no component of \(\mathcal S_t\) is marked heavy.

For the query \(\widehat{\mathcal S}\), mark the selected components, send selected heavy components to the center, and run breadth-first search in the induced graph on the remaining selected components. Whenever the search explores a selected component, it scans its quotient adjacency list and follows only selected non-heavy neighbors. Since every non-heavy component has quotient degree at most \(\Delta\), this takes
\[
    O(k\Lambda^2|\widehat{\mathcal S}|L_\Lambda^{O(1)})
\]
time. Connected components of size at most \(2k-1\) are output as small auxiliary connected components; larger selected components are merged into the center.

During the same scans, aggregate the owned-edge multiplicities needed by the dynamic program. Multiplicities to unselected, heavy, or large discarded components are added to the center atom, while multiplicities between selected components in the same small auxiliary component are kept between the corresponding atoms.

Finally, for every component in an output non-center part, scan its component--adhesion incidence list. The timestamped array records, for each touched adhesion, the unique non-center part touched by that adhesion. This uniqueness follows from the paragraph before the lemma. Since every non-center part contains at most \(2k-1\) components, and each adhesion can be recorded for at most one non-center part, the total incidence-scan time is
\[
    k^{O(1)}(1+\deg_\tau^+(t)).
\]
Combining the preprocessing, graph search, multiplicity aggregation, and incidence scan gives the claimed bounds. The heavy-component argument above and \cref{lem:small_side_components} imply that, for a good pair, no true small-side component is discarded.
\end{proof}

\begin{lemma}[Generated Decomposition Refines the Optimum]
\label{lem:nice_contains_optimum}
If \(\widehat C\) and \(\widehat{\mathcal S}\) are good for \(\phi^\star\) at \(t\), then \cref{alg:nice_decomposition} does not return \(\bot\). It returns a nice decomposition
\[
    D=(\Pi^{\mathrm{nice}}_{\beta(t)},\Theta_{\beta(t)},O)
\]
such that \(\Theta_{\beta(t)}\) refines \(\Pi_t^\star\).
\end{lemma}
\begin{proof}
Since \(\widehat C\) contains \(C_t^\star\), the projected components \(\mathcal R(\widehat C)\) refine \(\Pi_t^\star\). Since \(\widehat{\mathcal S}\) contains \(\mathcal S_t\) and avoids \(N_{\mathcal A_t(\widehat C)}(\mathcal S_t)\), every connected component of $\mathcal A_t(\widehat C)[\widehat{\mathcal S}]$ that intersects \(\mathcal S_t\) is contained entirely in \(\mathcal S_t\). By \cref{lem:small_side_components}, such a component has size at most \(2k-2\). By \cref{lem:fast_compatibility}, no true small-side component is discarded by the heavy-component rule. Therefore the algorithm never merges a true small-side component into the center.

Every component of \(\mathcal R(\widehat C)\) not selected by \(\widehat{\mathcal S}\) is outside \(\mathcal S_t\), and hence projects into the giant optimum part. Similarly, every selected auxiliary connected component merged into the center is either heavy or has size greater than \(2k-1\); by the previous paragraph, it is disjoint from \(\mathcal S_t\), and so it also projects into the giant optimum part. Since \(\widehat{\mathcal S}\) avoids \(R_t^{\mathrm{cen}}\), at least one central component is merged into \(\Theta_0\). Thus the algorithm does not return \(\bot\), and the center part \(O\) is contained in the giant optimum part.

All remaining atoms of \(\Theta_{\beta(t)}\) are projected components of \(\mathcal R(\widehat C)\). Each such component is contained in a single part of \(\Pi_t^\star\), because \(\mathcal R(\widehat C)\) refines \(\Pi_t^\star\). Hence \(\Theta_{\beta(t)}\) refines \(\Pi_t^\star\).

It remains to verify that the returned triple is nice. The center condition \(\Theta_{\beta(t)}|_O=\{O\}\) holds by construction, since all central atoms are merged into one part. Every non-center part of \(\Pi^{\mathrm{nice}}_{\beta(t)}\) comes from a connected component of \(\mathcal A_t(\widehat C)[\widehat{\mathcal S}]\) of size at most \(2k-1\), so it contains at most \(2k-1\) parts of \(\Theta_{\beta(t)}\).

Distinct non-center parts of \(\Pi^{\mathrm{nice}}_{\beta(t)}\) come from distinct connected components of the selected auxiliary graph after heavy and large components have been sent to the center. Hence no auxiliary edge joins two distinct non-center parts. In particular, no owned edge joins two distinct non-center parts, and no adhesion \(\sigma(u)\), for \(u\in\mathsf{Adh}(t)\), intersects two distinct non-center parts. Finally, by the observation in the definition of \(\mathcal A_t(\widehat C)\), every edge of \(G[\beta(t)]\) that is not owned at \(t\) is represented by the parent adhesion \(\sigma(t)\). Since \(\sigma(t)\) also cannot intersect two distinct non-center parts, no edge of \(G[\beta(t)]\) joins two distinct non-center parts. Therefore \(D\) is a nice decomposition.
\end{proof}

\subsection{Evaluating Nice Decompositions}

For a nice decomposition
\(D=(\Pi^{\mathrm{nice}}_{\beta(t)},\Theta_{\beta(t)},O)\), let
\[
    \Pi^{\mathrm{nice}}_{\beta(t)}\setminus\{O\} = \{P_1,\ldots,P_{\rho(D)}\}
\]
be its non-center parts. Set
\[
    B_0:=O, \qquad B_\ell:=O\cup P_\ell \quad\text{for }1\le \ell\le \rho(D).
\]

A labeled coarsening of \(\Theta_{\beta(t)}\) is an assignment of labels in \([k]\) to the parts of \(\Theta_{\beta(t)}\); parts assigned the same label are merged. It is compatible with \(D\) if it can be obtained by independently choosing labeled coarsenings of
\[
    \Theta_{\beta(t)}|_{B_\ell}
\]
for \(\ell=0,\ldots,\rho(D)\), with all choices assigning the same label to the common atom \(O\). When \(\rho(D)=0\), there is only the center block \(B_0=O=\beta(t)\).

\begin{lemma}[Solving a Nice Decomposition]
\label{lem:solve_nice_decomposition}
Given a nice decomposition
\[
    D=(\Pi^{\mathrm{nice}}_{\beta(t)},\Theta_{\beta(t)},O),
\]
together with the incidence records produced by \cref{alg:nice_decomposition}, all compatible labeled transitions whose bag labeling coarsens \(\Theta_{\beta(t)}\) can be evaluated in time
\[
    k^{O(k)}(1+\rho(D)+\deg_\tau^+(t)).
\]
Moreover, if \(\Theta_{\beta(t)}\) refines \(\Pi_t^\star\), then this evaluation includes the transition induced by \(\phi^\star\) at \(t\).
\end{lemma}
\begin{proof}
For each block \(B_\ell\), let
\[
    \mathcal A_\ell:=\Theta_{\beta(t)}|_{B_\ell}
\]
be its local atom set. By niceness,
\[
    |\mathcal A_0|=1, \qquad |\mathcal A_\ell|\le 2k \quad\text{for }1\le \ell\le \rho(D),
\]
because \(O\) is one atom and each non-center part contains at most \(2k-1\) atoms of \(\Theta_{\beta(t)}\). Hence all local labelings
\[
    \psi_\ell:\mathcal A_\ell\to [k]
\]
can be enumerated in \(k^{O(k)}\) time per block. We enumerate the common center label \(q\in[k]\), and only consider local labelings satisfying
\[
    \psi_\ell(O)=q.
\]

Using the incidence records, assign each child \(u\in\operatorname{children}(t)\) to exactly one block. If \(\sigma(u)\) touches a non-center part \(P_\ell\), assign \(u\) to \(B_\ell\); otherwise assign \(u\) to the center block \(B_0\). This is well-defined by niceness, since no adhesion touches two distinct non-center parts. Let $\mathcal U_\ell$ be the set of children assigned to \(B_\ell\). The incidence records allow these sets to be formed in $k^{O(1)}(1+\deg_\tau^+(t))$ time.

Similarly, assign each owned edge in \(E_t^\circ\) to exactly one block. If both endpoints lie in \(O\), assign the edge to \(B_0\). Otherwise, if at least one endpoint lies in \(P_\ell\), assign the edge to \(B_\ell\). This is well-defined because no graph edge joins two distinct non-center parts. The stored owned-edge multiplicities between atoms of \(\Theta_{\beta(t)}\) give, for each local labeling \(\psi_\ell\), the local owned-edge cost
\[
    c_\ell(\psi_\ell) := \sum_{\{A,A'\}\in \mathcal E_\ell} m_t(A,A')\, \mathbf 1[\psi_\ell(A)\ne \psi_\ell(A')],
\]
where \(\mathcal E_\ell\) is the set of unordered atom pairs whose owned edges are assigned to \(B_\ell\), and \(m_t(A,A')\) is the stored multiplicity of owned edges between atoms \(A\) and \(A'\). Loops inside a single atom contribute zero and may be ignored.

Fix a block \(B_\ell\), a common center label \(q\), and a local labeling \(\psi_\ell\) with \(\psi_\ell(O)=q\). Order the children assigned to this block as
\[
    \mathcal U_\ell=\{u_1,\ldots,u_m\}.
\]
For each child \(u_j\), the local labeling fixes the labeled adhesion state
\[
    \varphi_{u_j}:=\psi_\ell|_{\sigma(u_j)}.
\]
If \(\varphi_{u_j}\notin\mathcal L_T^{\sigma(u_j)}\), then this local labeling \(\psi_\ell\) is infeasible for the block. Otherwise the child may contribute any label set \(L_{u_j}\subseteq[k]\) with
\[
    f_{u_j}(\varphi_{u_j},L_{u_j})<\infty.
\]

Define the child-knapsack table
\[
    K_{\ell,\psi_\ell}^{j}(L)
\]
for \(0\le j\le m\) and \(L\subseteq[k]\). Its value is the minimum cost of the local owned edges together with the first \(j\) child subtrees, using overall label set \(L\). Initialize
\[
    K_{\ell,\psi_\ell}^{0}(L) =
    \begin{cases}
        c_\ell(\psi_\ell), & L=\lambda(\psi_\ell),\\
        \infty, & \text{otherwise},
    \end{cases}
\]
where
\[
    \lambda(\psi_\ell):=\{\psi_\ell(A):A\in\mathcal A_\ell\}.
\]
For \(j=1,\ldots,m\), set
\[
    K_{\ell,\psi_\ell}^{j}(L) =
    \min_{\substack{ L'\subseteq[k],\ L_u\subseteq[k]\\ L=L'\cup L_u }}
    \left(
        K_{\ell,\psi_\ell}^{j-1}(L') + f_{u_j}(\varphi_{u_j},L_u)
    \right).
\]
If \(\varphi_{u_j}\notin\mathcal L_T^{\sigma(u_j)}\), equivalently all terms involving \(u_j\) are \(\infty\).

The block table records the best value for each induced parent-adhesion restriction. Let
\[
    \eta_\ell(\psi_\ell):=\psi_\ell|_{\sigma(t)\cap B_\ell}.
\]
Define
\[
    G_\ell^q(\eta,L) :=
    \min_{\substack{
        \psi_\ell:\mathcal A_\ell\to[k]\\
        \psi_\ell(O)=q\\
        \eta_\ell(\psi_\ell)=\eta
    }}
    K_{\ell,\psi_\ell}^{|\mathcal U_\ell|}(L).
\]
The number of possible pairs \((\eta,L)\) is \(k^{O(k)}\). Indeed, \(L\subseteq[k]\), and by niceness the parent adhesion \(\sigma(t)\) intersects at most one non-center part; inside any block, its restriction is therefore determined by labels on at most \(2k\) atoms.

It remains to combine the block tables. For the fixed center label \(q\), define
\[
    H_0^q(\eta,L):=G_0^q(\eta,L).
\]
For \(j=1,\ldots,\rho(D)\), define
\[
    H_j^q(\eta,L) =
    \min_{\substack{
        \eta',\eta'',\, L',L''\\
        \eta=\eta'\cup\eta''\\
        L=L'\cup L''
    }}
    \left(
        H_{j-1}^q(\eta',L') + G_j^q(\eta'',L'')
    \right).
\]
The union \(\eta'\cup\eta''\) is only meaningful when the two partial parent-adhesion labelings agree wherever both are defined. In the present setting different non-center blocks are disjoint on \(\sigma(t)\), and their only possible overlap is on the center atom \(O\), whose label is fixed to be the common value \(q\).

After all blocks are combined, the entries
\[
    H_{\rho(D)}^q(\eta,L)
\]
with
\[
    \eta\in\mathcal L_T^{\sigma(t)}
\]
contribute candidates to
\[
    f_t(\eta,L).
\]
Taking the minimum over all common center labels \(q\in[k]\) and all compatible decompositions gives the value for the table of \(t\).

Every transition produced by these recurrences is sound. It chooses a compatible bag labeling, chooses child solutions whose labeled adhesion states agree with the restriction of that bag labeling, pays exactly the owned-edge cost at \(t\), and unions the child label sets with the labels used on the bag. Since the child adhesions partition among the blocks and owned edges are assigned to exactly one block, no child contribution or owned-edge contribution is omitted or counted twice.

Now assume \(\Theta_{\beta(t)}\) refines \(\Pi_t^\star\). Label every atom of \(\Theta_{\beta(t)}\) by its label under \(\phi^\star\). This labeling is compatible with \(D\): the center atom \(O\) receives the giant optimum label, and the restrictions to the blocks \(O\cup P_\ell\) agree on \(O\). For every child \(u\), the induced adhesion state is exactly
\[
    \phi^\star|_{\sigma(u)},
\]
and the chosen child label set is
\[
    \lambda(\phi^\star|_{\gamma(u)}).
\]
By the inductive completeness invariant, the corresponding child table entry is finite and has the correct value. Therefore the local block knapsacks and the final block knapsack include exactly the transition induced by \(\phi^\star\) at \(t\).

The running time is $k^{O(k)}$ per local block, plus a linear scan over the children assigned to that block. Each child is assigned to exactly one block, so the total child-processing time is
\[
    k^{O(k)}(1+\deg_\tau^+(t)).
\]
The final block knapsack has \(k^{O(k)}\) states per block and \(\rho(D)\) blocks, hence costs
\[
    k^{O(k)}(1+\rho(D)).
\]
Thus all compatible labeled transitions for \(D\) are evaluated in time
\[
    k^{O(k)}(1+\rho(D)+\deg_\tau^+(t)).
\]
\end{proof}

\subsection{The Fixed-Tree Algorithm}

Let \(c_{\mathrm{gen}}\) be a sufficiently large constant depending only on the fixed implementations of the routines below, and set
\[
    B_{\mathrm{gen}} := k^{c_{\mathrm{gen}} k} \Lambda^{6k-6}nL_\Lambda^{c_{\mathrm{gen}}}.
\]
During \textsc{FixedTreeDP}, we count the total random-generation and selected-component scanning work: geometric-gap operations, generated first-level elements, and \(\Lambda^2\) times the number of generated second-level components. If this counter exceeds \(B_{\mathrm{gen}}\), the run returns \(\infty\).

\begin{algorithm}[H]
\caption{\textsc{FixedTreeDP}$(G,k,s,\Lambda,(\tau,\beta),T)$}
\label{alg:fixed_tree_dp}
\begin{algorithmic}[1]
\State Preprocess \(T\) for LCA queries and Euler intervals.
\State Compute the owned edge sets \(E_t^\circ\) for all \(t\in V(\tau)\).
\State Initialize the generation counter \(Z\gets 0\).
\For{nodes \(t\in V(\tau)\) in postorder}
    \State Compute \(T_{\beta(t)}=\operatorname{proj}(T,\beta(t))\).
    \State Initialize all entries \(f_t(\phi_{\sigma(t)},L)\) to \(\infty\).
    \If{\(|\beta(t)|\le kq\)}
        \State Enumerate all labeled partitions \(\psi\in\mathcal L_T^{\beta(t)}\).
        \State Evaluate all standard transitions based on these bag labelings.
    \Else
        \State Generate \(R_1\) first-level samples \(\widehat C\subseteq E(T_{\beta(t)})\), increasing \(Z\) by the number of generation operations.
        \If{\(Z>B_{\mathrm{gen}}\)}
            \State \Return \(\infty\).
        \EndIf
        \For{each first-level sample \(\widehat C\)}
            \State Compute \(\mathcal R(\widehat C)\), build the quotient auxiliary graph, store owned-edge multiplicities, build component--adhesion incidence lists, and mark heavy components.
            \State Generate \(R_2\) second-level samples \(\widehat{\mathcal S}\subseteq\mathcal R(\widehat C)\), increasing \(Z\) by the number of generation operations and by \(\Lambda^2|\widehat{\mathcal S}|\).
            \If{\(Z>B_{\mathrm{gen}}\)}
                \State \Return \(\infty\).
            \EndIf
            \For{each second-level sample \(\widehat{\mathcal S}\)}
                \State \(D\gets \textsc{NiceDecomposition}(t,T,\widehat C,\widehat{\mathcal S})\), using the quotient auxiliary graph and sending all selected heavy components to the center.
                \If{\(D\ne\bot\)}
                    \State Evaluate \(D\) using \cref{lem:solve_nice_decomposition}.
                \EndIf
            \EndFor
        \EndFor
    \EndIf
    \State Truncate every entry larger than \(s\) to \(\infty\).
\EndFor
\State \Return \(f_r(\phi_\emptyset,[k])\).
\end{algorithmic}
\end{algorithm}

The standard transition for a small bag is the labeled gluing transition defined above: for each enumerated \(\psi\in\mathcal L_T^{\beta(t)}\), choose compatible child states, union their label sets with \(\lambda(\psi)\), add the owned-edge cost, and write the resulting value into the appropriate parent-adhesion state.

\begin{lemma}[Fixed-Tree Correctness]
\label{lem:fixed_tree_algorithm_correctness}
For a fixed tree \(T\), \cref{alg:fixed_tree_dp} never returns a finite value below \(\lambda_k(G)\). If \(T\) crosses some optimum \(k\)-cut \(\phi^\star\) in at most \(2k-2\) edges, then \cref{alg:fixed_tree_dp} returns \(\lambda_k(G)\) w.h.p., unless the generation counter overflows.
\end{lemma}

\begin{proof}
Soundness follows by induction over the postorder traversal. Every finite entry is stored with an actual witness labeling of \(\gamma(t)\). At an internal node, a transition glues actual child witness labelings to an actual bag labeling, and all labels agree on the adhesions. Since the edge sets \(E_u^\circ\) partition \(E(G)\) by ownership, the transition value is exactly the number of cut edges owned in the subtree rooted at \(t\). Therefore every finite root value with label set \([k]\) is the value of an actual \(k\)-cut of \(G\), and is at least \(\lambda_k(G)\). If the generation counter overflows, the algorithm returns \(\infty\), which is also not a finite underestimate.

Now suppose that \(T\) crosses an optimum labeled \(k\)-cut \(\phi^\star\) in at most \(2k-2\) edges. For every node \(t\), the restriction \(\phi^\star|_{\sigma(t)}\) lies in \(\mathcal L_T^{\sigma(t)}\). We prove by induction over the postorder traversal that the optimum-induced entry is retained at every node, provided the required good guesses exist and the counter does not overflow.

At a small bag, \(\phi^\star|_{\beta(t)}\in\mathcal L_T^{\beta(t)}\), so the standard restricted transition enumerates the optimum bag labeling. At a large bag, if the first-level separating-family event succeeds, \cref{lem:good_first_guess} gives a first-level guess \(\widehat C\) that contains the true projected cut set and avoids the projected small-side internal edges. If the second-level event also succeeds, \cref{lem:good_second_guess} gives a second-level guess \(\widehat{\mathcal S}\) containing the true small-side components and avoiding both their auxiliary neighborhood and the chosen center anchor. The heavy-component rule does not remove any true small-side component by \cref{lem:fast_compatibility}. For this good pair, \cref{lem:nice_contains_optimum} returns a nice decomposition whose refinement \(\Theta_{\beta(t)}\) refines the optimum bag partition, and \cref{lem:solve_nice_decomposition} evaluates the optimum-induced transition.

Thus the optimum-induced entry is retained at every node. At the root this gives a value at most \(\lambda_k(G)\). Together with soundness, the root value is exactly \(\lambda_k(G)\). The separating-family events hold simultaneously w.h.p., so the claim follows.
\end{proof}

\begin{lemma}[Fixed-Tree Running Time]
\label{lem:fixed_tree_running_time}
For a fixed tree \(T\), \cref{alg:fixed_tree_dp} runs in time
\[
    O\!\left((m+k^{O(k)}\Lambda^{6k-6}n)L_\Lambda^{O(1)}\right).
\]
\end{lemma}
\begin{proof}
The bound is deterministic, because the generation counter is charged for the actual work spent on random generation and selected-component scanning. If the counter exceeds \(B_{\mathrm{gen}}\), the algorithm stops.

Preprocessing \(T\) costs \(O(nL_\Lambda^{O(1)})\). Computing \(\operatorname{top}(v)\) for all vertices and then assigning each edge \(uv\) to the deeper of \(\operatorname{top}(u)\) and \(\operatorname{top}(v)\) costs
\[
    O\!\left((m+\sum_t|\beta(t)|)L_\Lambda^{O(1)}\right) = O((m+\Lambda n)L_\Lambda^{O(1)}).
\]
Since the graph is sparsified, \(m=O(sn)\le O(\Lambda n)\).

The projected trees \(T_{\beta(t)}\) are computed using \cref{lem:virtual_tree_and_states}. Their total construction time is
\[
    O\!\left(\sum_t (1+|\beta(t)|\log|\beta(t)|)\right) = O(\Lambda n\,L_\Lambda^{O(1)}),
\]
which is dominated by the claimed bound.

For a fixed large bag \(t\) and first-level guess \(\widehat C\), \cref{lem:fast_compatibility} gives quotient-construction cost
\[
    O\!\left( k^{O(1)} \left( |\beta(t)|+|E_t^\circ| + \sum_{u\in\operatorname{children}(t)\cup\{t\}}|\sigma(u)|^2 \right) L_\Lambda^{O(1)} \right).
\]
There are \(R_1=k^{O(k)}\Lambda^{2k-2}L_\Lambda\) first-level guesses per large bag. Summing over all nodes and using
\[
    \sum_t|\beta(t)|=O(\Lambda n), \qquad \sum_t|E_t^\circ|=m=O(sn)\le O(\Lambda n),
\]
and
\[
    \sum_t\sum_{u\in\operatorname{children}(t)\cup\{t\}}|\sigma(u)|^2 \le O(\Lambda) \sum_t\sum_{u\in\operatorname{children}(t)\cup\{t\}}|\sigma(u)|
    = O(\Lambda^2 n),
\]
the total quotient-construction cost is
\[
    O(k^{O(k)}n\Lambda^{2k}L_\Lambda^{O(1)}),
\]
which is dominated by \(O(k^{O(k)}\Lambda^{6k-6}nL_\Lambda^{O(1)})\) for \(k\ge2\).

For a joint guess \((\widehat C,\widehat{\mathcal S})\), \cref{lem:fast_compatibility} gives post-quotient cost
\[
    O\!\left( k^{O(1)} (1+\deg_\tau^+(t)+\Lambda^2|\widehat{\mathcal S}|) L_\Lambda^{O(1)} \right).
\]
The term \(\Lambda^2|\widehat{\mathcal S}|\) is charged to the generation counter.

The resulting nice decomposition is evaluated, by \cref{lem:solve_nice_decomposition}, in time
\[
    k^{O(k)}(1+\rho(D)+\deg_\tau^+(t)).
\]
Every non-center part contains at least one selected component, and distinct non-center parts contain disjoint selected components, so \(\rho(D)\le|\widehat{\mathcal S}|\). Thus the \(\rho(D)\) term is also charged to the counter. The uncharged work per joint guess is \(k^{O(k)}(1+\deg_\tau^+(t))\). There are
\[
    R_1R_2 = k^{O(k)}\Lambda^{6k-6}L_\Lambda^2
\]
joint guesses per large bag. Hence the total uncharged large-bag work is
\[
    O\!\left( k^{O(k)}\Lambda^{6k-6}L_\Lambda^{O(1)} \sum_t(1+\deg_\tau^+(t)) \right) = O(k^{O(k)}\Lambda^{6k-6}nL_\Lambda^{O(1)}).
\]

For small bags, \cref{lem:virtual_tree_and_states} gives at most \(k^{O(k)}\Lambda^{2k-2}\) labeled feasible bag partitions. For a fixed bag labeling \(\psi\) and child \(u\), the adhesion state \(\psi|_{\sigma(u)}\) is fixed, and only the \(2^k\) possible label sets \(L_u\) need be considered. Thus the standard transition for a fixed \(\psi\) is evaluated by a knapsack over the children of \(t\), in time
\[
    k^{O(k)}(1+\deg_\tau^+(t)+|E_t^\circ|),
\]
where \(|E_t^\circ|\) accounts for computing \(\operatorname{cost}_t(\psi)\). Hence the total small-bag cost is
\[
    O\!\left( k^{O(k)}\Lambda^{2k-2} \sum_t (1+\deg_\tau^+(t)+|E_t^\circ|) \right) = O(k^{O(k)}\Lambda^{2k-1}n),
\]
which is dominated by \(O(k^{O(k)}\Lambda^{6k-6}n)\) for \(k\ge2\). Combining all terms proves the claimed running time.
\end{proof}

\subsection{Generation Overflow Probability}

\begin{lemma}[Generation Overflow Probability]
\label{lem:generation_overflow_probability}
For a fixed tree \(T\), \cref{alg:fixed_tree_dp} does not return \(\infty\) because of the generation counter w.h.p.
\end{lemma}

\begin{proof}
All Bernoulli subsets are generated by geometric gaps between included elements, so the time spent on random generation is linear in the number of samples plus the number of generated elements.

For first-level samples, each large bag \(t\) generates \(R_1=k^{O(k)}\Lambda^{2k-2}L_\Lambda\) subsets of \(E(T_{\beta(t)})\), with inclusion probability \(p_1=\Theta(1/\Lambda)\). Let \(Z_1\) be the total number of selected first-level projected edges over all large bags and all first-level samples. Then
\[
    \mathbb E[Z_1] = R_1p_1\sum_t|E(T_{\beta(t)})| = O(k^{O(k)}\Lambda^{2k-2}nL_\Lambda),
\]
using \(\sum_t|E(T_{\beta(t)})|=O(\Lambda n)\). By a Chernoff bound, \(Z_1\le O(k^{O(k)}\Lambda^{2k-2}nL_\Lambda^{O(1)})\) w.h.p.

For a first-level sample \(\widehat C\), we have \(|\mathcal R(\widehat C)|\le|\widehat C|+1\). Therefore, on the above high-probability event,
\[
    \sum_{t,\widehat C}|\mathcal R(\widehat C)| \le R_1|V(\tau)|+Z_1 = O(k^{O(k)}\Lambda^{2k-2}nL_\Lambda^{O(1)}).
\]

Condition on the first-level samples and on this event. The second-level samples are independent Bernoulli samples over the fixed component universes \(\mathcal R(\widehat C)\). Each such sample uses inclusion probability \(p_2=\Theta(1/\Lambda^2)\), and there are \(R_2=k^{O(k)}\Lambda^{4k-4}L_\Lambda\) samples for each first-level guess. Let \(Z_2\) be the total number of selected second-level components. Its conditional expectation is at most
\[
    \mathbb E[Z_2] \le R_2p_2 \sum_{t,\widehat C}|\mathcal R(\widehat C)| = O(k^{O(k)}\Lambda^{6k-8}nL_\Lambda^{O(1)}).
\]
Thus, by another Chernoff bound, \(Z_2\le O(k^{O(k)}\Lambda^{6k-8}nL_\Lambda^{O(1)})\) w.h.p.

The overhead of geometric-gap generation is deterministic once the number of bags and samples is fixed. It is bounded by
\[
    O(R_1|V(\tau)|+R_1R_2|V(\tau)|) = O(k^{O(k)}\Lambda^{6k-6}nL_\Lambda^{O(1)}).
\]
The remaining charged work is linear in \(Z_1+\Lambda^2Z_2\). Combining the high-probability bounds above gives
\[
    Z_1+\Lambda^2Z_2 \le O(k^{O(k)}\Lambda^{6k-6}nL_\Lambda^{O(1)})
\]
w.h.p. Hence the generation counter stays below \(B_{\mathrm{gen}}\) w.h.p.
\end{proof}

\subsection{Sampling a Respecting Tree}

\begin{lemma}[Respecting-Tree Sampler]
\label{lem:respecting_tree_sampler}
In time \(\widetilde O(m/\epsilon^2)\), one can compute an implicit nonnegative tree packing \(y\) and hence a distribution
\[
    \pi(T):=\frac{y_T}{\sum_{T'}y_{T'}}
\]
over the trees in its support. A support tree can be sampled from \(\pi\) with polylogarithmic overhead beyond accessing its implicit representation. If \(\epsilon<1/(2k-1)\), then for every optimum \(k\)-cut \(A^\star\),
\[
    \Pr_{T\sim\pi}\bigl[|E(T)\cap A^\star|\le 2k-2\bigr] \ge 1-\frac{2(k-1)}{(2k-1)(1-\epsilon)}.
\]
For \(\epsilon=1/(8k)\), this probability is at least \(1/(4k)\).
\end{lemma}
\begin{proof}
Apply \cref{lem:dual_tree_packing} with \(\eta=\epsilon\), and sample
from its normalized tree weights.

Fix an optimum \(k\)-cut \(A^\star\), let
\[
    D:=|E(T)\cap\delta_G(A^\star)|,
    \qquad
    x:=\frac{z(E)}{\lambda_k(G)}.
\]
Feasibility and \eqref{eq:dual_packing_approx} give
\[
\begin{aligned}
    \mathbb E[D]
    &\le
    \frac{\lambda_k(G)+z(E)}{\tau}\\
    &\le
    (k-1)\frac{1+x}{(1-\epsilon)/2+x}\\
    &\le\frac{2(k-1)}{1-\epsilon}.
\end{aligned}
\]
Since \(D\) is integral,
\[
    \Pr[D>h]
    \le\frac{\mathbb E[D]}{h+1}.
\]
Taking \(h=2k-2\) proves
\[
    \Pr[D\le2k-2]
    \ge
    1-\frac{2(k-1)}{(2k-1)(1-\epsilon)}.
\]
For \(\epsilon=1/(8k)\), this is at least \(1/(4k)\).
\end{proof}

\begin{corollary}[Sampled Respecting Tree]
\label{cor:sampled_respecting_tree}
Fix an optimum \(k\)-cut \(A^\star\). With \(\epsilon=1/(8k)\), sample \(c k\log n\) independent trees from the distribution of \cref{lem:respecting_tree_sampler}, for a sufficiently large constant \(c\). Then, with probability at least \(1-n^{-\Omega(c)}\), at least one sampled tree \(T\) satisfies
\[
    |E(T)\cap A^\star|\le 2k-2 .
\]
\end{corollary}
\begin{proof}
Each sample succeeds with probability at least \(1/(4k)\). Therefore the failure probability after \(c k\log n\) independent samples is at most
\[
    \left(1-\frac{1}{4k}\right)^{c k\log n} \le \exp\!\left(-\frac{c}{4}\log n\right) = n^{-c/4}.
\]
\end{proof}

\subsection{The Full FPT Algorithm}

\begin{algorithm}[H]
\caption{\textsc{FPT-MinKCut}$(G,k,s)$}
\label{alg:fpt_min_k_cut}
\begin{algorithmic}[1]
\Require Unweighted multigraph \(G\), integer \(k\ge 2\), parameter \(s\ge\lambda_k(G)\).
\Ensure A minimum \(k\)-cut of \(G\) with high probability.
\State Apply Nagamochi--Ibaraki sparsification with threshold \(s+1\).
\If{\(G\) has at least \(k\) connected components}
    \State \Return the zero-value \(k\)-cut obtained by refining connected components.
\EndIf
\If{\(G\) has \(1<h<k\) connected components}
    \State Solve the connected components independently for all feasible requested part counts and combine them by the component DP described above.
    \State \Return the best combined cut.
\EndIf
\State Run \cref{thm:unbreakable_decomp} with cut parameter \(s\), and let \(\Lambda\) be its adhesion/unbreakability parameter.
\State Let \((\tau,\beta)\) be the resulting decomposition.
\State Set \(\epsilon\gets 1/(8k)\).
\State Compute the implicit distribution \(\pi\) from \cref{lem:respecting_tree_sampler}.
\State Set \(L\gets Ck\log n\) for a sufficiently large absolute constant \(C\).
\State \(best\gets\infty\), and \(C_{\mathrm{best}}\gets\bot\).
\For{\(j=1\) \textbf{to} \(L\)}
    \State Sample a tree \(T_j\sim\pi\).
    \State \(val_j\gets \textsc{FixedTreeDP}(G,k,s,\Lambda,(\tau,\beta),T_j)\).
    \If{\(val_j<best\)}
        \State \(best\gets val_j\), and store the corresponding witnessed cut as \(C_{\mathrm{best}}\).
    \EndIf
\EndFor
\State \Return \(C_{\mathrm{best}}\).
\end{algorithmic}
\end{algorithm}

\begin{theorem}[FPT Min \(k\)-Cut Algorithm]
\label{thm:fpt_mincut}
Given an unweighted multigraph \(G\) with $|E(G)| = m_0$ and integers \(k,s\) with \(s\ge\lambda_k(G)\), \cref{alg:fpt_min_k_cut} returns an optimum \(k\)-cut w.h.p. Its stored value is never finite and smaller than \(\lambda_k(G)\).

In terms of the decomposition parameter \(\Lambda\), its running time is
\[
    O\!\left(m_0 + k^{O(k)}\Lambda^{6k-6}nL_\Lambda^{O(1)}\right).
\]
Since \(\Lambda=O(s\log^2 n\log\log n)\), this is
\[
    O(m_0) + k^{O(k)}s^{6k-6}n \log^{12k-12}n (\log\log n)^{6k-6} \log^{O(1)}(k(n+s+2)).
\]
In particular, when \(s\le n^{O(1)}\), this is
\[
    O(m_0) + k^{O(k)}s^{6k-6}n \log^{12k+O(1)}n (\log\log n)^{6k-6}.
\]
\end{theorem}
\begin{proof}
By \cref{def:ni_sparsification}, after sparsification we have \(m=O(sn)\) and all cut values up to \(s+1\) are preserved, this takes $O(m_0)$ time. Since every finite value maintained by the algorithm is truncated at \(s\), every finite witnessed cut has the same value in the sparsified graph and in the original graph.

The decomposition step takes
\[
    \widetilde O(m+\Lambda n)=\widetilde O(\Lambda n)
\]
time. By \cref{lem:respecting_tree_sampler}, the respecting-tree distribution is computed in time
\[
    \widetilde O(m/\epsilon^2) = \widetilde O(k^{O(1)}sn) \le \widetilde O(k^{O(1)}\Lambda n)
\]
for \(\epsilon=1/(8k)\).

Fix an optimum \(k\)-cut \(A^\star\). By \cref{cor:sampled_respecting_tree}, w.h.p. some sampled tree \(T_j\) crosses \(A^\star\) in at most \(2k-2\) edges. For this tree, \cref{lem:fixed_tree_algorithm_correctness} and \cref{lem:generation_overflow_probability} imply that \textsc{FixedTreeDP} returns \(\lambda_k(G)\) w.h.p. Every finite value returned by any fixed-tree run is the value of an actual \(k\)-cut, and is therefore at least \(\lambda_k(G)\). Thus the best witnessed cut is optimum w.h.p. A union bound over all randomized subroutines gives success w.h.p.

By \cref{lem:fixed_tree_running_time}, each fixed-tree run satisfies
\[
    O\!\left((m+k^{O(k)}\Lambda^{6k-6}n)L_\Lambda^{O(1)}\right) = O(k^{O(k)}\Lambda^{6k-6}nL_\Lambda^{O(1)}).
\]
There are \(L=O(k\log n)\) sampled trees, and this factor is absorbed into the \(k^{O(k)}\) factor and the \(L_\Lambda^{O(1)}\) factor. The disconnected-case reduction only solves vertex-disjoint and edge-disjoint components and combines their values by a \(k^{O(1)}\) dynamic program, so it does not change the asymptotic bound.

Finally, substituting \(\Lambda=O(s\log^2 n\log\log n)\) and \(L_\Lambda=O(k\log(k(n+s+2)))\) gives
\[
    k^{O(k)}s^{6k-6}n \log^{12k-12}n (\log\log n)^{6k-6} \log^{O(1)}(k(n+s+2)).
\]
If \(s\le n^{O(1)}\) (e.g. for a simple graph), this simplifies to
\[
    k^{O(k)}s^{6k-6}n \log^{12k+O(1)}n (\log\log n)^{6k-6}.
\]
\end{proof}

\section{Near-Linear Approximation for Min \texorpdfstring{$k$}{k}-Cut}

Our near-linear edge-unbreakable decomposition and the FPT algorithm of \cref{thm:fpt_mincut} imply a near-linear approximation algorithm for \(\lambda_k\). Following the framework of~\cite{LSS22}, we first remove very small internal \(2\)-cuts, then sample edges so that the sampled graph has small \(k\)-cut value, and finally solve Min \(k\)-Cut exactly on the sampled graph.

\begin{theorem}
\label{thm:k_cut_approx}
Let \(0<\epsilon\le 1\). Given an unweighted multigraph \(G=(V,E)\), there is a randomized algorithm that computes a \((1+\epsilon)\)-approximation of \(\lambda_k(G)\) with high probability in time
\[
    \widetilde O(k^{O(1)}m) + k^{O(k)} \left(\frac{k}{\epsilon^3}\right)^{6k-6} n\, \log^{18k-18} n\, (\log\log n)^{6k-6} \log^{O(1)}\!\bigl(k(n+\epsilon^{-1})\bigr).
\]
In particular, if \(\epsilon^{-1}\le n^{O(1)}\), this is
\[
    \widetilde O(k^{O(1)}m) + (k/\epsilon)^{O(k)} n\log^{18k+O(1)}n(\log\log n)^{6k-6}.
\]
\end{theorem}

\begin{proof}
Let
\[
    \epsilon':=\epsilon/10.
\]
If \(G\) already has at least \(k\) connected components, then \(\lambda_k(G)=0\), and we return a zero-value \(k\)-cut. Otherwise compute a \(2\)-approximation \(\widetilde\lambda_k\) to \(\lambda_k(G)\) \cite{SV95,Kar00}, so
\[
    \lambda_k(G)\le \widetilde\lambda_k\le 2\lambda_k(G).
\]
This preliminary step, and the minimum-cut computations below, take \(\widetilde O(k^{O(1)}m)\) time.

Set
\[
    \tau_0:=\frac{\epsilon'}{2k}\widetilde\lambda_k.
\]
For a graph \(H\), define
\[
    \lambda_2^+(H) := \min\{\lambda_2(K): K\text{ is a connected component of }H \text{ with }|V(K)|\ge 2\},
\]
with \(\lambda_2^+(H)=\infty\) if no such component exists.

Starting from \(G\), repeatedly find a connected component \(K\) of the current graph with
\[
    \lambda_2(K)<\tau_0
\]
and delete the edges of a minimum \(2\)-cut in \(K\). Stop when the graph has at least \(k\) connected components, or when every connected component \(K\) with \(|V(K)|\ge 2\) satisfies
\[
    \lambda_2(K)\ge \tau_0.
\]
Let the resulting graph be \(G_1\). Each deletion splits one connected component into two, and therefore increases the number of connected components by one. Hence, before the graph has \(k\) connected components, there are at most \(k-1\) deletions. The total number of deleted edges is less than
\[
    k\tau_0 = \frac{\epsilon'}2\widetilde\lambda_k \le \epsilon'\lambda_k(G).
\]

If \(G_1\) has at least \(k\) connected components, take a zero-value \(k\)-cut of \(G_1\). When evaluated in \(G\), this cut uses only deleted edges, and hence has value at most
\[
    \epsilon'\lambda_k(G)\le (1+\epsilon)\lambda_k(G).
\]
Thus assume from now on that \(G_1\) has fewer than \(k\) connected components. Then
\[
    \lambda_2^+(G_1)\ge \tau_0.
\]

Sample every edge of \(G_1\) independently with probability
\[
    p:= \min\left\{ 1,\, \frac{100\log n}{\epsilon'^2\lambda_2^+(G_1)} \right\},
\]
and let \(G_2\) be the sampled graph. By the cut-sampling lemma of~\cite{LSS22}, applied independently inside each connected component of \(G_1\) and union-bounded over the components, with high probability every cut inside every component is preserved after scaling by \(p\). Therefore, with high probability, every \(k\)-cut \(C\) of \(G_1\) satisfies
\[
    (1-\epsilon')c_{G_1}(C) \le \frac{1}{p}c_{G_2}(C) \le (1+\epsilon')c_{G_1}(C).
\]
Indeed, a \(k\)-cut value in \(G_1\) is the sum of its contributions inside the connected components of \(G_1\).

We now bound \(\lambda_k(G_2)\). Since
\[
    \lambda_2^+(G_1) \ge \tau_0 = \frac{\epsilon'}{2k}\widetilde\lambda_k \ge \frac{\epsilon'}{2k}\lambda_k(G) \ge \frac{\epsilon'}{2k}\lambda_k(G_1),
\]
if \(p<1\), then
\[
    p = \frac{100\log n}{\epsilon'^2\lambda_2^+(G_1)} \le \frac{200k\log n}{\epsilon'^3\lambda_k(G_1)}.
\]
Therefore, on the sampling event,
\[
    \lambda_k(G_2) \le (1+\epsilon')p\lambda_k(G_1) = O\!\left(\frac{k\log n}{\epsilon^3}\right).
\]
If \(p=1\), then
\[
    \lambda_2^+(G_1) \le \frac{100\log n}{\epsilon'^2},
\]
and the same lower bound
\[
    \lambda_2^+(G_1)\ge \frac{\epsilon'}{2k}\lambda_k(G_1)
\]
implies
\[
    \lambda_k(G_2)=\lambda_k(G_1) = O\!\left(\frac{k\log n}{\epsilon^3}\right).
\]
Thus, with high probability,
\[
    \lambda_k(G_2)\le s_0, \qquad s_0:=C\frac{k\log n}{\epsilon^3},
\]
for a sufficiently large absolute constant \(C\).

We solve Min \(k\)-Cut exactly on \(G_2\) with cut-size bound \(s_0\). First apply Nagamochi--Ibaraki sparsification with threshold \(s_0+1\). This preserves the minimum \(k\)-cut of \(G_2\), and the resulting graph has \(O(s_0n)\) edges. This takes $O(|E(G_2)|) \le O(m)$ time.

On this sparsified graph, compute the edge-unbreakable decomposition with cut parameter \(s_0\). Its adhesion and unbreakability parameter is
\[
    \Lambda_0 = O(s_0\log^2 n\log\log n) = O\!\left( \frac{k}{\epsilon^3} \log^3 n\log\log n \right).
\]
The decomposition construction time is
\[
    \widetilde O(s_0n+\Lambda_0n),
\]
which is dominated by the exact FPT step.

By \cref{thm:fpt_mincut}, the exact algorithm on \(G_2\) runs in time
\[
    O\!\left( k^{O(k)} \Lambda_0^{6k-6} n L_{\Lambda_0}^{O(1)} \right),
\]
where
\[
    L_{\Lambda_0} = O\!\left(k\log(k(n+\Lambda_0+2))\right).
\]
Substituting the value of \(\Lambda_0\), this is
\[
    k^{O(k)} \left(\frac{k}{\epsilon^3}\right)^{6k-6} n\, \log^{18k-18} n\, (\log\log n)^{6k-6} \log^{O(1)}\!\bigl(k(n+\epsilon^{-1})\bigr).
\]

It remains to transfer the exact solution on \(G_2\) back to \(G\). Let \(C_2\) be a minimum \(k\)-cut of \(G_2\), and let \(C_1^\star\) be a minimum \(k\)-cut of \(G_1\). On the sampling event,
\[
    c_{G_1}(C_2) \le \frac{1}{1-\epsilon'}\cdot \frac{c_{G_2}(C_2)}{p} \le \frac{1}{1-\epsilon'}\cdot \frac{c_{G_2}(C_1^\star)}{p} \le \frac{1+\epsilon'}{1-\epsilon'}\lambda_k(G_1).
\]
Since \(G_1\) is obtained from \(G\) by deleting edges,
\[
    \lambda_k(G_1)\le \lambda_k(G).
\]
Adding back the deleted edges can increase the value of any fixed cut by at most the total number of deleted edges, which is at most \(\epsilon'\lambda_k(G)\). Hence the value of \(C_2\) in \(G\) is at most
\[
    \frac{1+\epsilon'}{1-\epsilon'}\lambda_k(G) + \epsilon'\lambda_k(G) \le (1+\epsilon)\lambda_k(G),
\]
because \(\epsilon'=\epsilon/10\) and \(0<\epsilon\le 1\).

Combining the preliminary approximation, the componentwise small-cut deletions, sampling, NI sparsification, decomposition construction, and the exact FPT call gives the claimed running time.
\end{proof}

\section{Algorithm for Min \texorpdfstring{$k$}{k}-Cut}
\label{sec:min_k_cut}

We now combine the FPT algorithm from the previous section, the weighted Min \(k\)-Cut algorithm of \cref{sec:weighted_sparse_k_cut}, and the border/island framework of He and Li~\cite{HL22}. The algorithm first computes a constant-factor estimate \(s=\Theta(\lambda_k(G))\). If \(s\) is small, we run the FPT algorithm parameterized by the cut value.

Otherwise, we apply border-preserving sparsification: an NI certificate \(H\), a partition \(\mathcal P\), and a contracted weighted multigraph \(G'\) are constructed so that \(G'\) has only \(\widetilde O(n/s)\) vertices. Moreover, for every optimum \(k\)-cut \(C\), there is a subset of its singleton components that can be merged into non-singleton sides so that the resulting lower-order cut is respected by \(\mathcal P\) and appears as a low-weight cut of \(G'\). We call this lower-order cut a border of \(C\), and the merged singleton vertices are the islands.

In the large-\(s\) regime, the algorithm guesses the number \(i\) of singleton islands in the sparsified representation of an optimum cut. If \(i=0\), the entire optimum cut is represented inside \(G'\), so we solve the corresponding weighted Min \(k\)-Cut instance on \(G'\). If \(i>0\), we enumerate the \((k-i)\)-cut border in \(G'\), lift it using the stored edge bundles, and then complete the cut by extracting the best \(i\) singleton islands from the lifted border. The cases \(i=1,2\) are handled by specialized routines, while \(i\ge3\) is handled by a matrix-multiplication-based island routine. Balancing the FPT branch against the largest large-\(s\) branch yields the final exponent.

Throughout this section \(G=(V,E)\) is a simple unweighted graph on \(n\) vertices. We write \(\lambda_k=\lambda_k(G)\). We begin with some preliminaries exclusive to this section.

\begin{definition}[Border and Islands {\cite[Definition~1.2]{HL22}}]
\label{def:border_islands}
Given a \(k\)-cut \(C\) with exactly \(r\) singleton components, denote the singleton components by $S_1=\{v_1\},\ldots,S_r=\{v_r\}$, and denote the remaining components by \(S_{r+1},\ldots,S_k\). A border of \(C\) is obtained by merging some singleton components into non-singleton components. More precisely, choose a subset \(I\subseteq[r]\) and a function $\eta:I\to [k]\setminus[r]$. For each \(a\in[k]\setminus[r]\), set $S'_a:=S_a\cup\{v_j:j\in I,\ \eta(j)=a\}$. The corresponding border is the \((k-|I|)\)-cut consisting of the parts \(S'_a\) for \(a\in[k]\setminus[r]\), together with the unmerged singleton components \(S_j\) for \(j\in[r]\setminus I\). The vertices $\{v_j:j\in I\}$ are called the islands of this border.
\end{definition}

\begin{definition}[Minimum Vertex-Weighted Triangle]
\label{def:min_vertex_triangle}
Given a graph \(G=(V,E)\) and an integer vertex-weight function \(W:V\to\mathbb Z\), the \emph{Minimum Vertex-Weighted Triangle} problem is to find a triangle \(\{u,v,w\}\) minimizing
\[
    W(u)+W(v)+W(w).
\]
\end{definition}

\begin{theorem}[Minimum Vertex-Weighted Triangle {\cite[Corollary~2]{AF26}}]
\label{thm:min_vertex_weighted_triangle}
Let \(\mathrm{MM}(n)\) denote the time for multiplying two \(n\times n\) matrices. In an \(n\)-vertex graph whose integer vertex weights have absolute value at most \(W\), a minimum vertex-weighted triangle can be found in time
\[
    O(\mathrm{MM}(n)\log W).
\]
\end{theorem}

\begin{definition}[Matrix Multiplication Exponents]
\label{def:mm_exponents}
The square matrix multiplication exponent \(\omega\) is the infimum over all real numbers \(\rho\) such that two \(n\times n\) matrices can be multiplied using
\[
    n^{\rho+o(1)}
\]
arithmetic operations.

More generally, for \(a,b,c\ge 0\), the rectangular matrix multiplication exponent \(\omega(a,b,c)\) is the infimum over all real numbers \(\rho\) such that an \(n^a\times n^b\) matrix can be multiplied by an \(n^b\times n^c\) matrix using
\[
    n^{\rho+o(1)}
\]
arithmetic operations. In particular,
\[
    \omega=\omega(1,1,1).
\]

We use the standard symmetries and homogeneity of rectangular matrix multiplication following from permutation of the indices and considering $n^\alpha$ instead of $n$ for some scalar $\alpha$:
\[
    \omega(a,b,c)
\]
is invariant under permutations of \(a,b,c\), and for every \(\alpha>0\),
\[
    \omega(\alpha a,\alpha b,\alpha c)=\alpha\,\omega(a,b,c).
\]
\end{definition}

We speed up \cref{thm:min_vertex_weighted_triangle} on tripartite graphs of different part sizes.

\begin{lemma}[Rectangular Minimum Vertex-Weighted Triangle]
\label{lem:rectangular_min_vertex_weighted_triangle}
Let \(F\) be a tripartite graph with parts \(A,B,C\), where
\[
    |A|\le n^a,\qquad |B|\le n^b,\qquad |C|\le n^c.
\]
Fix any exponent
\[
    \widehat\omega(a,b,c)>\omega(a,b,c).
\]
If the integer vertex weights have absolute value at most \(n^{O(1)}\),
then one can deterministically find a triangle minimizing the sum of its
vertex weights in time
\[
    n^{\widehat\omega(a,b,c)}.
\]
\end{lemma}

The proof of
\cref{lem:rectangular_min_vertex_weighted_triangle}
is given in
\cref{app:rectangular_vertex_weighted_triangles}.

For \(r\ge3\), define
\[
    r_1:=\left\lfloor\frac r3\right\rfloor,\qquad
    r_2:=\left\lfloor\frac{r+1}{3}\right\rfloor,\qquad
    r_3:=\left\lceil\frac r3\right\rceil .
\]
We fix exponents
\[
    \widehat\Phi(r)>\omega(r_1,r_2,r_3).
\]
For the three base cases, choose fixed constants satisfying
\[
\begin{aligned}
    \omega(1,1,1)
    &<2.371339<\widehat\Phi(3)<2.37134,\\
    \omega(1,1,2)
    &=\omega(1,2,1)
      <3.250385<\widehat\Phi(4)<3.250386,\\
    \omega(1,2,2)
    &=\omega(2,1,2)
      =2\omega(1,1/2,1)\\
    &\le4.085988<\widehat\Phi(5)<4.085989.
\end{aligned}
\]
The first inequality is from \cite{ADWXXZ25}; the remaining rectangular
bounds are from \cite{VWXXZ24}.

For \(r\ge6\), define
\[
    \widehat\Phi(r):=
    \widehat\Phi(r-3)+\widehat\Phi(3).
\]
The recurrence is justified by combining an algorithm for the
\((r_1,r_2,r_3)\)-rectangular product with square matrix multiplication
on the \(n\times n\) block level.  Thus
\[
    \widehat\Phi(r)>\omega(r_1,r_2,r_3)
\]
for every \(r\ge3\).  The sequence is nondecreasing by induction.

\begin{algorithm}[H]
\caption{Island Discovery}
\label{alg:improved_islands}
\begin{algorithmic}[1]
\Require A simple unweighted graph \(F\) on at most \(n\) vertices, and an integer \(r\ge3\)
\Ensure A minimum-cost set \(I\subseteq V(F)\) of \(r\) singleton islands

\State Let
\[
    r_1=\left\lfloor\frac r3\right\rfloor,\qquad r_2=\left\lfloor\frac{r+1}{3}\right\rfloor,\qquad r_3=\left\lceil\frac r3\right\rceil .
\]
\State Construct fixed vertex sets \(V_1,V_2,V_3\). For each \(j\in\{1,2,3\}\) and each set \(S\subseteq V(F)\) of size \(r_j\), create a vertex \(x_S\in V_j\).
\State Give \(x_S\) weight
\[
    W(x_S):=\sum_{v\in S}\deg_F(v)-e_F(S).
\]
\For{each triple \((e_{12},e_{23},e_{31})\) with \(0\le e_{ab}\le r_ar_b\)}
    \State Let \(F'_{e_{12},e_{23},e_{31}}\) be the tripartite graph on \(V_1,V_2,V_3\) with an edge between \(x_{S_a}\in V_a\) and \(x_{S_b}\in V_b\) iff
    \[
        S_a\cap S_b=\emptyset \quad\text{and}\quad e_F(S_a,S_b)=e_{ab}.
    \]
    \State Find a minimum vertex-weighted triangle \((x_{S_1},x_{S_2},x_{S_3})\) in \(F'_{e_{12},e_{23},e_{31}}\) using \cref{lem:rectangular_min_vertex_weighted_triangle}.
    \State Evaluate the corresponding island set \(I=S_1\cup S_2\cup S_3\) by
    \[
        \sum_{j=1}^3 W(x_{S_j})-(e_{12}+e_{23}+e_{31}).
    \]
\EndFor
\State \Return the island set of minimum value over all guesses.
\end{algorithmic}
\end{algorithm}

\begin{lemma}[Island Discovery]
\label{lem:improved_islands}
\cref{alg:improved_islands} returns an optimal set of \(r\) singleton islands in time
\[
    k^{O(k)}n^{\widehat\Phi(r)}
\]
for every \(r\ge3\).
\end{lemma}

\begin{proof}
For a set \(I\subseteq V(F)\), \(|I|=r\), the additional cut cost of extracting \(I\) as singleton components is
\[
    \operatorname{cost}(I) = \sum_{v\in I}\deg_F(v)-e_F(I).
\]
Indeed, the degree sum counts every edge from \(I\) to \(V(F)\setminus I\) once and every edge inside \(I\) twice, while each edge inside \(I\) is cut only once.

Partition \(I\) into three disjoint sets
\[
    I=S_1\cup S_2\cup S_3, \qquad |S_j|=r_j.
\]
Then
\[
\begin{aligned}
    \operatorname{cost}(I) &= \sum_{j=1}^3 \left( \sum_{v\in S_j}\deg_F(v)-e_F(S_j) \right)  \\ &\qquad
    - \bigl( e_F(S_1,S_2)+e_F(S_2,S_3)+e_F(S_3,S_1) \bigr).
\end{aligned}
\]
Thus, after fixing the three cross-edge counts
\[
    e_{12},e_{23},e_{31},
\]
minimizing the island cost is exactly a minimum vertex-weighted triangle problem in the auxiliary tripartite graph. Conversely, every triangle in this auxiliary graph consists of three pairwise disjoint sets with the guessed cross-edge counts and therefore gives a valid \(r\)-island set with the displayed cost. Taking the minimum over all guesses is correct.

The auxiliary parts satisfy
\[
    |V_j|=\binom{|V(F)|}{r_j}\le n^{r_j} \qquad (j=1,2,3),
\]
and the auxiliary weights have absolute value at most \(kn\). Therefore, by \cref{lem:rectangular_min_vertex_weighted_triangle}, each triangle step runs in time $O(n^{\widehat\Phi(r)})$. The number of cross-edge guesses is
\[
    \prod_{1\le a<b\le3}(r_ar_b+1)=k^{O(1)}.
\]

It remains only to account for the construction of the auxiliary graphs. After building an adjacency table for \(F\), the disjointness and cross-edge count for a pair of sets can be computed in \(k^{O(1)}\) time. The largest pair table has size $n^{r_2+r_3}$. We claim
\[
    r_2+r_3\le \widehat\Phi(r). \tag{1} \label{eq:pair_table_bounded_by_phi}
\]
For \(r=3,4,5\), this follows directly from the chosen values:
\[
    2<\widehat\Phi(3),\qquad 3<\widehat\Phi(4),\qquad 4<\widehat\Phi(5).
\]
For \(r\ge6\), use the recurrence
\[
    \widehat\Phi(r)=\widehat\Phi(r-3)+\widehat\Phi(3).
\]
When \(r\) is increased by \(3\), the quantity \(r_2+r_3\) increases by \(2\), while \(\widehat\Phi(r)\) increases by \(\widehat\Phi(3)>2\). Hence \eqref{eq:pair_table_bounded_by_phi} follows by induction from the cases \(r=3,4,5\).

Thus all auxiliary vertices and auxiliary edges can be generated in
\[
    k^{O(k)}n^{\widehat\Phi(r)}
\]
time, and the same bound holds for the full algorithm.
\end{proof}

\begin{lemma}[Border-Coupled Island Extension]
\label{lem:border_coupled_islands}
Let \(\mathcal B=(B_1,\ldots,B_\ell)\) be a candidate \(\ell\)-cut border of \(G\), and let \(i=k-\ell\). Suppose \(i\ge3\). Among all \(k\)-cuts obtained from \(\mathcal B\) by extracting exactly \(i\) singleton islands from the border components, the optimum extension can be found deterministically in time
\[
    k^{O(k)}n^{\widehat\Phi(i)}.
\]
\end{lemma}

\begin{proof}
Fix a vector
\[
    (i_1,\ldots,i_\ell)
\]
of nonnegative integers with
\[
    i_1+\cdots+i_\ell=i.
\]
Here \(i_j\) is the number of singleton islands extracted from \(B_j\). We discard vectors for which \(i_j\ge |B_j|\), since the residual component \(B_j\setminus I_j\) must remain nonempty.

For a fixed border component \(B_j\), extracting islands \(I_j\subseteq B_j\) changes the border value only through edges internal to \(G[B_j]\). The additional cost is
\[
    \sum_{v\in I_j}\deg_{G[B_j]}(v)-e_{G[B_j]}(I_j).
\]
Thus, for fixed \(i_j\), the optimal choice of \(I_j\) is exactly the \(i_j\)-island problem on \(G[B_j]\). Moreover, for the fixed vector \((i_1,\ldots,i_\ell)\), the choices in different border components are independent and their additional costs add.

If \(i_j\ge3\), \cref{lem:improved_islands} solves the subproblem in time
\[
    k^{O(k)}|B_j|^{\widehat\Phi(i_j)} \le k^{O(k)}n^{\widehat\Phi(i)}.
\]
Here we use that \(i_j\le i\) and that the chosen sequence \(\widehat\Phi(\cdot)\) is nondecreasing. If \(i_j\in\{0,1,2\}\), brute force costs at most \(O(|B_j|^2)\), which is at most
\[
    O(n^2)\le O(n^{\widehat\Phi(i)}),
\]
because \(i\ge3\) and \(\widehat\Phi(i)\ge\widehat\Phi(3)>2\).

Therefore, for one fixed vector \((i_1,\ldots,i_\ell)\), all component subproblems can be solved in time
\[
    k^{O(k)}n^{\widehat\Phi(i)}.
\]
The number of feasible vectors is at most
\[
    \binom{i+\ell-1}{\ell-1}\le k^{O(k)}.
\]
Taking the best extension over all vectors gives the claimed running time and proves correctness.
\end{proof}

\begin{lemma}[Border-Preserving Sparsification]
\label{lem:border_preserving_sparsification}
There is a randomized algorithm which, given a simple graph \(G\) and an upper estimate \(s=\Theta(\lambda_k(G))\) with \(\lambda_k(G)\le s\), computes a subgraph \(H\subseteq G\), a partition
\[
    \mathcal P=\{P_1,\ldots,P_N\}
\]
of \(V(G)\), and the contracted multigraph \(G'\) obtained by contracting the parts of \(\mathcal P\) in \(H\), with the following properties w.h.p.:
\begin{enumerate}
    \item \(H\) is an NI certificate at threshold \(s+1\). In particular, every vertex partition of \(H\)-value at most \(s\) has the same value in \(G\).
    \item \(N=\widetilde O(n/s)\).
    \item \(|E(H)|=O(sn)\).
    \item \(G'\) has \(O(sn)\) edges, counted with multiplicity.
    \item The edges of \(G'\) store edge bundles: each edge of \(G'\) stores the list of \(H\)-edges it represents.
    \item For every minimum \(k\)-cut \(C\) of \(G\), there is a border \(B_{I,\eta}\) of \(C\), with \(i:=|I|\), that is respected by \(\mathcal P\), and whose contraction is a \((k-i)\)-cut of \(G'\) of weight at most
    \[
        \left( 1-\left(1-\frac{2}{\log n}\right)\frac{i}{k} \right)\lambda_k(G).
    \]
\end{enumerate}
The running time is
\[
    \widetilde O(k^{O(1)} m) + k^{O(k)} n\log^{18k+O(1)}n (\log\log n)^{6k-6}.
\]
\end{lemma}

\begin{proof}
First compute a Nagamochi--Ibaraki certificate \(H\) using the first \(s+1\) forests. Thus \(H\subseteq G\), \(|E(H)|=O(sn)\), and cuts of value at most \(s\) are preserved as desired.

Then compute the approximation required by the sparsification framework of He and Li~\cite{HL22} using \cref{thm:k_cut_approx}. This takes
\[
    \widetilde O(k^{O(1)}m) + k^{O(k)}n\log^{18k+O(1)}n(\log\log n)^{6k-6}.
\]

After the NI certificate is available, apply the Kawarabayashi--Thorup partitioning~\cite{KT18} and the trimming/shaving operations of He and Li~\cite{HL22} to obtain \(\mathcal P\). Contract the parts of \(\mathcal P\) in \(H\) to obtain \(G'\). Since \(|E(H)|=O(sn)\), the contracted multigraph has \(O(sn)\) edges counted with multiplicity. During the contraction, store with each edge of \(G'\) the list of \(H\)-edges it represents.

The guarantees
\[
    N=\widetilde O(n/s)
\]
and the stated border-preservation property are exactly the guarantees of the He--Li border-preserving sparsification, with the approximation step replaced by \cref{thm:k_cut_approx}. All steps after the approximation run in near-linear time on $H$ and $H \subseteq G$, thus in $\widetilde O(m)$ time.
\end{proof}

We give a variant of the following approximate minimum cut enumeration algorithm to enumerate borders.

\begin{theorem}[Enumeration of Near-Minimum \(k\)-Cuts
{\cite[Theorem~20]{GHLL21}}]
\label{thm:optimal_k_cut}
Let \(G\) be a weighted undirected graph.  For each \(k\ge3\) and $\alpha \ge 1$, there is
an algorithm that enumerates all \(k\)-cuts of weight at most
\(\alpha\lambda_k\) in time
\[
    n^{\alpha k}(\log n)^{O(\alpha k^2)}
\]
with probability at least \(1-1/\operatorname{poly}(n)\).
\end{theorem}

\begin{lemma}[Border Enumeration]
\label{lem:border_enumeration}
Let \(Q\) be an \(N\)-vertex weighted undirected graph, let
\(2\le\ell\le k\), and put \(p:=\alpha k\).  Suppose that
\(p\ge\ell\). When \(N\ge k\), assume additionally that \(\lambda_k(Q)>0\). If \(N\ge k\), there is a randomized procedure which,
with high probability, outputs every \(\ell\)-cut of \(Q\) of weight at
most \(\alpha\lambda_k(Q)\).
Its running time is
\[
    O(|E(Q)|)+
    \begin{cases}
        k^{O(pk)}
        N^p(\log(N+2))^{O(pk)},
            & p\ge3,\\[1ex]
        N^p
        \exp\!\left(
            O\!\left(pk\log(k+2)\sqrt{\log(N+2)}\right)
        \right),
            & 2\le p<3.
    \end{cases}
\]
In particular, for fixed \(k\) and \(p\), the second bound is
\(N^{p+o(1)}\).

If \(N<k\), the procedure outputs all \(\ell\)-cuts deterministically
in \(O(|E(Q)|)+k^{O(k)}\) time.  Cuts are represented by their terminal
partitions and contraction histories.

More generally, for any \(0<\delta<1/2\), the success probability can be
made at least \(1-\delta\) by multiplying the displayed running time by
\(O(\log(1/\delta))\).

Moreover, the procedure can maintain input-edge bundles throughout the
contractions.  When requested, it can output with a listed cut the input
edges crossing that cut, in additional time linear in the size of the
requested list.
\end{lemma}

\begin{proof}
If \(N<k\), then either \(N<\ell\), in which case there is no
\(\ell\)-cut, or we enumerate all \(\ell\)-partitions of \(V(Q)\).
There are at most
\[
    \ell^N\le k^k
\]
such partitions, proving the last assertion about this case.  Henceforth
assume \(N\ge k\).

Fix a target \(\ell\)-cut \(C\), let \(J:=\partial_Q C\), and suppose
\[
    c_Q(C)\le\alpha\lambda_k(Q)=\frac{p}{k}\lambda_k(Q).
\]
We repeatedly use the following consequence of
\cite[Theorem~17]{GHLL21}.  Suppose that \(H\) is an \(r\)-vertex
contraction minor of \(Q\) obtained without contracting an edge of \(J\),
and let \(J_H\) be the surviving image of \(J\).  Since contraction can
only restrict the collection of \(k\)-partitions,
\[
    \lambda_k(H)\ge\lambda_k(Q),
    \qquad
    c(J_H)\le c(J).
\]
Consequently, if
\[
    \alpha_H:=\frac{c(J_H)}{\lambda_k(H)},
\]
then \(\alpha_Hk\le p\).
Therefore, whenever
\[
    r\ge s\ge 8pk+2k,
\]
random contraction from \(r\) to \(s\) vertices preserves \(J_H\) with
probability at least
\begin{equation}
\label{eq:border_interval_survival}
    \left(\frac rs\right)^{-p}k^{-O(pk)}.
\end{equation}
We use here that the target \(J\) need not be the boundary of a \(k\)-cut.

We first handle \(p\ge3\).  Run the recursive-contraction algorithm in
the proof of \cite[Theorem~20]{GHLL21}, using the sizes
\[
    N_j
    =
    \left\lceil
        \max\left\{
            N^{(2/p)^j},\,20pk
        \right\}
    \right\rceil
\]
and
\[
    \left\lceil
        \left(\frac{N_j}{N_{j+1}}\right)^p
    \right\rceil
\]
independent contraction children at level \(j\).  The condition
\(p\ge3\) gives
\[
    \frac p2\ge\frac32
\]
in the depth and running-time calculation of that proof.

The proof applies without change after replacing its fixed \(k\)-cut
boundary by \(J\), because its survival estimate is
\eqref{eq:border_interval_survival}.  At a terminal graph, enumerate all
\(\ell\)-cuts instead of selecting a \(k\)-cut.  Each terminal graph has
\(O(pk)\) vertices, so this costs at most
\[
    \ell^{O(pk)}\le k^{O(pk)}
\]
per leaf.  The small-\(N\) branch and the terminal-enumeration factor are
absorbed exactly as in the proof of \cite[Theorem~20]{GHLL21}.
Accounting for the terminal enumeration gives
\[
    k^{O(pk)}
    N^p(\log(N+2))^{O(pk)}
\]
total time and high-probability coverage of every target cut.

It remains to handle \(2\le p<3\).  This is the endpoint at which the
preceding schedule ceases to make sufficiently rapid progress.

Set
\[
    \tau:=\left\lceil8pk+2k\right\rceil
\]
and choose an absolute constant \(C_0\) large enough that
\eqref{eq:border_interval_survival} is at least
\[
    \psi\left(\frac sr\right)^p,
    \qquad
    \psi:=k^{-C_0pk}.
\]
If \(N\le\tau\), enumerate all \(\ell\)-cuts directly.  Since
\(\tau=O(pk)\) and \(p<3\), this costs \(k^{O(k)}\), which is covered by
the claimed bound.

Suppose now that \(N>\tau\), and define
\[
    c:=\max\left\{
        2,\,
        \exp\!\left(\sqrt{\log(N+2)}\right)
    \right\},
    \qquad
    b:=\left\lceil
        2\psi^{-1}(2c)^p
    \right\rceil.
\]
At an \(r\)-vertex recursion node with \(r>\tau\), independently run
\(b\) contractions down to
\[
    s:=\max\left\{\tau,\left\lceil\frac rc\right\rceil\right\}
\]
vertices and recurse on all resulting graphs.  At a leaf, enumerate all
\(\ell\)-cuts of the terminal graph.

Consider a recursion node reached without contracting \(J\).  We have
\(r/s\le2c\), so every child preserves \(J\) with probability at least
\[
    \psi(2c)^{-p}.
\]
It follows that the probability that at least one of its \(b\) children
preserves \(J\) is at least
\[
\begin{aligned}
    1-
    \left(1-\psi(2c)^{-p}\right)^b
    &\ge
    1-\exp\left(-b\psi(2c)^{-p}\right)\\
    &\ge 1-e^{-2}.
\end{aligned}
\]
The recursion has depth
\[
    L
    \le
    \left\lceil\log_c\frac N\tau\right\rceil+1
    =
    O\!\left(\sqrt{\log(N+2)}\right).
\]
Thus one execution lists a fixed target \(C\) with probability at least
\begin{equation}
\label{eq:border_block_success}
    (1-e^{-2})^L
    =
    \exp\!\left(-O\!\left(\sqrt{\log(N+2)}\right)\right).
\end{equation}

We next bound the number of target cuts, in order to amplify
\eqref{eq:border_block_success} simultaneously for all of them.
Perform one ordinary contraction from \(N\) to \(\tau\) vertices.
Every target survives with probability at least
\[
    \psi\left(\frac{\tau}{N}\right)^p,
\]
whereas a terminal graph contains at most \(\ell^\tau\) distinct
\(\ell\)-cuts.  Moreover, distinct target cuts that both survive induce
distinct terminal cuts.  If \(\mathcal C\) denotes the family of target
cuts, linearity of expectation therefore gives
\[
    |\mathcal C|\,
    \psi\left(\frac{\tau}{N}\right)^p
    \le \ell^\tau.
\]
Consequently,
\begin{equation}
\label{eq:border_target_count}
    |\mathcal C|
    \le
    N^p k^{O(pk)}.
\end{equation}
Repeating the blocked recursion
\[
    \exp\!\left(
        O\!\left(\sqrt{\log(N+2)}\right)
    \right)
    \operatorname{poly}(p,k,\log(N+2))
\]
times and applying \eqref{eq:border_target_count} and a union bound
therefore lists every target with high probability.

It remains to bound the work.  Write
\[
    B:=O\!\left(\psi^{-1}2^p\right),
\]
so that \(b\le Bc^p\).  At depth \(j\), the recursion has at most
\(b^j\) nodes, each with
\[
    O\!\left(\frac{N}{c^j}+\tau\right)
\]
vertices.  A contraction trial on an \(r\)-vertex weighted graph takes
\(O(r^2)\) time after parallel capacities have been aggregated.  The
contraction work at depth \(j\) is consequently bounded by
\[
    O\!\left(
        b^{j+1}
        \left(\frac{N}{c^j}+\tau\right)^2
    \right).
\]
For the first term in the square,
\[
\begin{aligned}
    b^{j+1}\left(\frac{N}{c^j}\right)^2
    &\le
    N^2Bc^p\left(Bc^{p-2}\right)^j.
\end{aligned}
\]
Since \(p\ge2\), this sequence is maximized at the last level.  Using
\(c^L=(N/\tau)c^{O(1)}\), we obtain
\[
    \sum_{j=0}^{L-1}
    b^{j+1}\left(\frac{N}{c^j}\right)^2
    \le
    N^p
    \exp\!\left(
        O\!\left(pk\log(k+2)\sqrt{\log(N+2)}\right)
    \right).
\]
The terms involving \(\tau\), as well as the
\[
    b^L\ell^\tau
\]
terminal enumeration work, satisfy the same bound.  The repetitions
needed above only enlarge the constant hidden in the exponential.
This proves the asserted running time when \(2\le p<3\).

Finally, we maintain the edge bundles.  Initially, associate every input
edge with a singleton bundle.  Represent a merged bundle by a 
binary node whose children point to the two previous
bundles.  Thus merging parallel edges creates a new bundle in \(O(1)\)
time without copying its contents; bundles that become self-loops are
simply discarded from the current graph.  At a terminal cut, the bundles
of its crossing terminal edges form a disjoint representation of exactly
the crossing input edges.  Flattening these bundles therefore takes time
linear in the requested output size.
\end{proof}

\begin{lemma}[Enhanced border enumeration]
\label{lem:enhanced_border_enumeration_representatives}
Consider an invocation of \cref{lem:border_enumeration} on the contracted
graph \(G'\) produced by
\cref{lem:border_preserving_sparsification}.  Precompute
\[
    d_H(v):=\deg_H(v)
    \qquad(v\in V(G)).
\]
The contraction process can maintain the following information.

For every current vertex \(x\), representing a set \(X_x\subseteq V(G)\),
store:

\begin{enumerate}
    \item \(\operatorname{size}(x):=|X_x|\);
    \item the two vertices of \(X_x\) having minimum \(d_H\)-value, or all
          of \(X_x\) if \(|X_x|<2\);
    \item if \(|X_x|\ge2\), a pair of distinct vertices
          \(a_x,b_x\in X_x\) minimizing
          \[
              d_H(a_x)+d_H(b_x)-\mu_H(a_x,b_x),
          \]
          where \(\mu_H(u,v)\) is the number of \(H\)-edges between
          \(u\) and \(v\).
\end{enumerate}

For every current edge \(xy\), store, in addition to its edge
bundle, an original edge \(uv\) in that bundle minimizing
\(d_H(u)+d_H(v)-\mu_H(u,v)\).

All this information can be maintained with constant work whenever two
current vertices or two parallel current edges are merged.  At a terminal
\(\ell\)-cut, for each side one can compute its represented size, its two
minimum-degree representatives, and its optimum stored pair in
\(k^{O(1)}\) time.  The side itself need not be materialized.
\end{lemma}

\begin{proof}
The initial values are obtained by one scan of the vertices and edges of
\(H\).  Suppose that current vertices \(x,y\) are contracted to a vertex
\(z\).  Then
\[
    \operatorname{size}(z)
    =
    \operatorname{size}(x)+\operatorname{size}(y),
\]
and the two minimum-degree representatives of \(z\) are the two smallest
members of the two stored representative lists.

The optimum pair in \(X_x\cup X_y\) is the best of four possibilities:

\begin{enumerate}
    \item the stored optimum pair of \(x\);
    \item the stored optimum pair of \(y\);
    \item the best original \(H\)-edge in the current \(xy\)-bundle;
    \item the pair formed by a minimum-degree representative of \(x\) and
          one of \(y\).
\end{enumerate}

Indeed, a pair using two vertices of the same represented set is covered
by the first two cases.  For a cross pair \(u\in X_x\), \(v\in X_y\),
an adjacent pair is covered by the third case.  If \(u,v\) are
nonadjacent, the fourth case has degree sum no larger than
\(d_H(u)+d_H(v)\).  If its two representatives happen to be adjacent,
the corresponding edge candidate is only better after subtracting its
multiplicity.  Thus one of the four candidates is optimal.

When parallel current edges are merged, their best original-edge records
are combined by taking the better of the two.  Hence every update takes
constant time.

A terminal cut has \(O(\alpha k^2)=k^{O(1)}\) current vertices.  For one
of its sides, combine the vertex records and the current edges whose
endpoints lie on that side by the same rule.  This yields its size,
representatives, and optimum pair in \(k^{O(1)}\) time.  No original
vertex set is expanded.
\end{proof}

\begin{lemma}[Few Islands Extraction]
\label{lem:few_islands}
Let \(H\subseteq G\) be the certificate produced by
\cref{lem:border_preserving_sparsification}, and assume that the metadata
of \cref{lem:enhanced_border_enumeration_representatives} has been
initialized.  Let \(\mathcal B'\) be a candidate border in \(G'\), given
by its terminal contraction state, and let
\[
    \mathcal B=(B_1,\ldots,B_\ell)
\]
denote its implicit lift.  Suppose that the list of \(H\)-edges crossing
\(\mathcal B\) has size at most \(s\).

For \(i=1\), one can return the minimum-value one-island extension among
those whose total \(H\)-value is at most \(s\), or report that none
exists, in \(O(s+k^{O(1)})\) time.
For \(i=2\), one can return the minimum-value two-island
extension among those whose total \(H\)-value is at most \(s\), or
report that no such extension exists, in
\[
    O(s^2+k^{O(1)})
\]
time.  The returned extension is represented implicitly.

If the returned extension has \(H\)-value at most \(s\), then it has
the same value in \(G\).  It can be materialized as a partition of
\(V(G)\) in \(O(n)\) time when requested.
\end{lemma}

\begin{proof}
Write
\[
    W_0:=|\delta_H(\mathcal B)|.
\]
Expand the crossing bundles and mark every crossing \(H\)-edge with the
current timestamp.  Its endpoint occurrences identify the implicit
border component containing each touched vertex.  For a touched vertex
\(v\), let
\[
    b(v):=|E_H(v,V\setminus B(v))|,
\]
where \(B(v)\) is its border component, and let
\[
    T:=\{v:b(v)>0\}.
\]
Then
\[
    |T|=O(s),
    \qquad
    \sum_{v\in T}b(v)=2W_0=O(s).
\]
Define the additional singleton cost
\[
    c(v):=
    \begin{cases}
        d_H(v)-b(v),&v\in T,\\
        d_H(v),&v\notin T.
    \end{cases}
\]

\paragraph{One island.}
For every component \(X\) with \(|X|\ge2\), evaluate its stored
minimum-degree representative, using its adjusted cost if it is touched.
Also evaluate every touched vertex in \(X\).  An untouched optimum has
cost \(d_H(v)\) and is dominated by the minimum-degree representative;
a touched optimum is explicitly evaluated.  Hence the best feasible
candidate is optimal. Return the best candidate only if its total value \(W_0+c(v)\) is at
most \(s\); otherwise report that no qualifying extension exists.

\paragraph{Two islands.}
If \(W_0>s\), no qualifying extension exists.  Otherwise put
\[
    R:=s-W_0.
\]
We seek the minimum additional cost at most \(R\).

For pairs in distinct components, compute the cheapest feasible
singleton of every component \(X\) with \(|X|\ge2\), using its stored
representatives and its touched vertices.  The best distinct-component
pair consists of the two cheapest values belonging to different
components.

Now consider a component \(X\) with \(|X|\ge3\).  For distinct
\(u,v\in X\), their additional cost is
\begin{equation}
\label{eq:two_island_additional_cost}
    c(u)+c(v)-\mu_H(u,v).
\end{equation}
Evaluate the stored optimal pair of \(X\), replacing \(d_H\) by the
adjusted value \(c\) at touched endpoints.  This dominates every pair
with no touched endpoint: before adjustment it is globally optimal,
and adjustment can only decrease its value.

It remains to cover pairs with a touched endpoint.  Such a pair can
have additional cost at most \(R\) only if its touched endpoint \(u\)
satisfies
\[
    c(u)\le R+1.
\]
Let
\[
    T^\star:=\{u\in T:c(u)\le R+1\}.
\]
Then
\[
\begin{aligned}
    \sum_{u\in T^\star}d_H(u)
    =\sum_{u\in T^\star}\bigl(b(u)+c(u)\bigr)
    \le\sum_{u\in T}b(u)+(s+1)|T|
     =O(s^2).
\end{aligned}
\]

For every component \(X\), retain the two cheapest distinct vertices
among its touched vertices and its two stored minimum-degree
representatives, using adjusted costs.  For every
\(u\in T^\star\cap X\), scan the \(H\)-adjacency list of \(u\).  An
incident edge is internal to \(X\) exactly when it was not timestamped
as a crossing edge.  Evaluate every internal adjacent pair using
\eqref{eq:two_island_additional_cost}.  Also evaluate every retained
candidate different from \(u\) that was not encountered as an internal
neighbor.

These candidates dominate every omitted nonadjacent partner.  Let
\(v\ne u\) be such an omitted nonadjacent partner.  There is a retained
candidate \(w\ne u\) with \(c(w)\le c(v)\).  If \(uw\notin E(H)\), the
algorithm evaluates
\[
    c(u)+c(w)\le c(u)+c(v).
\]
If \(uw\in E(H)\), the adjacency scan evaluates
\[
    c(u)+c(w)-1<c(u)+c(v).
\]
Thus the fact that both retained representatives might be adjacent to
\(u\) causes no gap: one of those adjacent pairs is then strictly
better than the omitted nonadjacent pair.

The total adjacency-scan time is
\[
    O\left(\sum_{u\in T^\star}d_H(u)\right)=O(s^2).
\]
The algorithm therefore finds the minimum qualifying two-island
extension or correctly reports that none exists.
\end{proof}

\subsection{Main Algorithm}

We will use the following fact about \(\widehat\Phi(r)\).  For every
\(r\ge4\),
\[
    \widehat\Phi(r)
    \le
    \widehat\Phi(4)+\frac{6(r-4)}7
    <
    3.250386+\frac{6(r-4)}7 .
    \tag{1}\label{eq:phi_tail_bound_main}
\]
For \(r=4\), the first inequality is equality.  For \(r=5\), it follows
from
\[
    \widehat\Phi(5)
    <4.085989
    <3.250385+\frac67
    <\widehat\Phi(4)+\frac67.
\]
For \(r=6\), it follows from
\[
    2\widehat\Phi(3)
    <4.74268
    <3.250385+\frac{12}{7}
    <\widehat\Phi(4)+\frac{12}{7}.
\]
The remaining cases follow by induction, because
\[
    \widehat\Phi(r+3)
    =
    \widehat\Phi(r)+\widehat\Phi(3)
    <
    \widehat\Phi(r)+\frac{18}{7}.
\]

Define
\[
    \theta_k:=
    \begin{cases}
        \displaystyle \frac1{12}, & k=3,\\[1.2ex]
        \displaystyle \frac2{19}, & k=4,\\[1.2ex]
        \displaystyle \frac{1+2.37134}{26}, & k=5,\\[1.2ex]
        \displaystyle \frac{k+3.250386-5}{7k-10}, & k\ge6.
    \end{cases}
    \tag{2}\label{eq:theta_k_definition}
\]
Let
\[
    E_k:=1+(6k-6)\theta_k.
    \tag{3}\label{eq:E_k_definition}
\]

\begin{algorithm}[H]
\caption{Minimum \(k\)-Cut}
\label{alg:k_cut}
\begin{algorithmic}[1]
\Require Simple unweighted graph \(G=(V,E)\), integer \(k\ge3\)
\Ensure A minimum \(k\)-cut of \(G\)

\If{\(|V|=k\)}
    \State \Return the all-singletons \(k\)-cut.
\EndIf

\State Compute a \(2\)-approximation \(s\) of \(\lambda_k(G)\) \cite{SV95,Kar00}.
\If{\(s=0\)}
    \State \Return a zero-value \(k\)-cut obtained by grouping connected
           components
\EndIf
\State Set \(\tau:=n^{\theta_k}\).

\If{\(s\le\tau\)}
    \State \Return the output of \cref{thm:fpt_mincut} on \(G\) with cut parameter \(s\).
\EndIf

\State Apply \cref{lem:border_preserving_sparsification} to obtain \(H,\mathcal P,G'\), where \(G'\) has $N=\widetilde O(n/s)$ vertices and \(O(sn)\) edges.
\State Compute \(d_H(v)\) for all \(v\), construct the adjacency lists
       of \(H\), and initialize the vertex and edge-bundle metadata of
       \cref{lem:enhanced_border_enumeration_representatives}.
\State Initialize \(C_{\mathrm{best}}\gets\bot\) and \(w_{\mathrm{best}}\gets\infty\).

\If{\(|V(G')|\ge k\)}
    \State Run \cref{thm:weighted_k_cut_sparse} on \(G'\), amplifying the success probability, and obtain a
       minimum \(k\)-cut \(\mathcal C'\) of \(G'\)
    \State Lift \(\mathcal C'\) to a \(k\)-partition \(\mathcal C\) of \(V(G)\) and compute its \(H\)-value \(w\).
    \If{\(w\le s\)}
        \State Set \(C_{\mathrm{best}}\gets\mathcal C\) and \(w_{\mathrm{best}}\gets w\).
    \EndIf
\EndIf

\For{\(i=1\) \textbf{to} \(k-1\)}
    \State Set
    \[
        \ell:=k-i,
        \qquad
        \beta_i:=1-\left(1-\frac{2}{\log n}\right)\frac{i}{k}.
    \]
    \If{\(\ell=1\)}
        \State Let \(\mathcal L_i\) consist of the trivial one-part border of \(G'\).
    \Else
        \State Run the appropriate form of border enumeration on \(G'\), amplified
            to failure probability \(n^{-\Omega(k)}\), where \(n\) is the
            order of the original graph:
        \If{\(i\le2\)}
            \State Use
                \cref{lem:enhanced_border_enumeration_representatives}
        \Else
            \State Use \cref{lem:border_enumeration}
        \EndIf
        \State Let \(\mathcal L_i\) be the resulting family of candidate borders
    \EndIf

    \For{each candidate border \(\mathcal B'\in\mathcal L_i\)}
        \If{\(i\le2\)}
            \State Keep the lift implicit in the terminal contraction state
                and its side metadata
            \State Flatten the crossing-edge bundles, stopping after
                \(s+1\) edges
            \If{more than \(s\) crossing edges are found}
                \State \textbf{continue}
            \EndIf
            \State Apply \cref{lem:few_islands}
            \If{no qualifying extension is returned}
                \State \textbf{continue}
            \EndIf
            \State Let \(C\) be the implicit output and \(w\) its \(H\)-value
            \If{\(w>s\)}
                \State \textbf{continue}
            \EndIf
        \Else
            \State Materialize the lift
                \(\mathcal B=(B_1,\ldots,B_\ell)\)
            \State Apply \cref{lem:border_coupled_islands} to obtain \(C\),
                and set \(w:=|\delta_G(C)|\)
        \EndIf
        \If{\(w<w_{\mathrm{best}}\)}
            \State Store the representation of \(C\) and set
                \(w_{\mathrm{best}}\gets w\)
        \EndIf
    \EndFor
\EndFor

\State Materialize \(C_{\mathrm{best}}\), if necessary, and return it.
\end{algorithmic}
\end{algorithm}

\begin{theorem}[Minimum \(k\)-Cut on Simple Graphs]
\label{thm:simple_mincut}
For \(3\le k\le n\), \cref{alg:k_cut} returns a minimum \(k\)-cut of a simple unweighted \(n\)-vertex graph with high probability and runs in time
\[
    k^{O(k^2)} n^{E_k} (\log n)^{O(k^2)},
\]
where
\[
    E_3=2,\qquad E_4=\frac{55}{19},
\]
\[
    E_5
    =
    1+\frac{12}{13}(1+2.37134)
    <4.112007,
\]
and, for every \(k\ge6\),
\[
    E_k
    =
    1+(6k-6)\frac{k-1.749614}{7k-10}.
\]
\end{theorem}

\begin{proof}
The \(2\)-approximation gives
\[
    \lambda_k(G)\le s\le2\lambda_k(G),
\]
so \(s=\Theta(\lambda_k(G))\), and \(s\) is a valid cut-size upper bound for the FPT branch.

All randomized calls are amplified relative to the original order \(n\).
The border-preserving sparsification, the weighted call on \(G'\), and
each of the \(O(k)\) border-enumeration calls therefore fail with
probability at most \(n^{-\Omega(k)}\).  After increasing the amplification
constants, a union bound shows that all required events hold simultaneously
with high probability.

\paragraph{Correctness.}
If \(s\le n^{\theta_k}\), the algorithm invokes the exact FPT algorithm of \cref{thm:fpt_mincut}. Since \(\lambda_k(G)\le s\), this returns a minimum \(k\)-cut.

Assume \(s>n^{\theta_k}\), and let \(C^\star\) be a minimum \(k\)-cut of \(G\). By \cref{lem:border_preserving_sparsification}, there is a border \(B_{I,\eta}\) of \(C^\star\), with
\[
    i:=|I|,
\]
that is respected by \(\mathcal P\). Let \(\ell:=k-i\). The contraction of this border is an \(\ell\)-cut \(\mathcal B'^\star\) of \(G'\), and its weight is at most
\[
    \beta_i\lambda_k(G) = \left( 1-\left(1-\frac{2}{\log n}\right)\frac{i}{k} \right)\lambda_k(G).
\]

If \(i=0\), then \(C^\star\) itself is respected by \(\mathcal P\), so its contraction is a \(k\)-cut of \(G'\) of value \(\lambda_k(G)\). Hence \(|V(G')|\ge k\). Every \(k\)-cut of \(G'\) lifts to a \(k\)-partition of \(V(G)\) with the same value in \(H\). If \(G'\) had a \(k\)-cut of value below \(\lambda_k(G)\), the lift would have \(H\)-value below \(\lambda_k(G)\le s\), and the NI certificate would give the same value in \(G\), contradiction. Therefore
\[
    \lambda_k(G')=\lambda_k(G).
\]
The weighted algorithm on \(G'\) returns a minimum \(k\)-cut of \(G'\), whose lift has \(H\)-value \(\lambda_k(G)\le s\). The NI certificate then certifies that this lift has the same value in \(G\).

Now suppose \(i\ge1\). If \(\ell=1\), the algorithm explicitly includes the trivial one-part border. If \(\ell\ge2\), then \cref{lem:border_enumeration}, invoked with \(\alpha=\beta_i\), lists \(\mathcal B'^\star\). Indeed, if \(G'\) has fewer than \(k\) vertices, all \(\ell\)-cuts are listed. If \(G'\) has at least \(k\) vertices, the lifting argument above gives
\[
    \lambda_k(G')\ge\lambda_k(G),
\]
and hence
\[
    c_{G'}(\mathcal B'^\star) \le \beta_i\lambda_k(G) \le \beta_i\lambda_k(G').
\]

After \(\mathcal B'^\star\) is listed, its lift
\[
    \mathcal B^\star=(B_1,\ldots,B_\ell)
\]
is a border of \(C^\star\), and \(C^\star\) is obtained from it by extracting exactly \(i\) singleton islands. If \(i\in\{1,2\}\), \cref{lem:few_islands} computes the optimum extension of \(\mathcal B^\star\) with respect to \(H\). The extension corresponding to \(C^\star\) has \(H\)-value \(\lambda_k(G)\le s\), so the returned extension also has \(H\)-value at most \(s\), and the NI certificate gives the same value in \(G\). If \(i\ge3\), \cref{lem:border_coupled_islands} computes the optimum extension with respect to \(G\), and hence returns a cut of value at most \(\lambda_k(G)\).

Every stored candidate is a valid \(k\)-cut of \(G\), and the value used for comparison is its true \(G\)-value: for \(i\le2\), this follows from the check \(w\le s\) and the NI certificate; for \(i\ge3\), the value is evaluated directly in \(G\). Therefore the best stored candidate has value exactly \(\lambda_k(G)\).

\paragraph{Preprocessing and small-\(s\) branch.}
The initial approximation and the sparsification step cost
\[
    \widetilde O(k^{O(1)}m)
    +
    k^{O(k)}n\log^{18k+O(1)}n(\log\log n)^{6k-6}.
\]
Since \(G\) is simple, \(m=O(n^2)\), and \(E_k\ge2\) for every
\(k\ge3\).  Hence this is at most
\[
    k^{O(k^2)}n^{E_k}(\log n)^{O(k^2)}.
\]

If \(s\le n^{\theta_k}\), \cref{thm:fpt_mincut} runs in
\[
    O(m)
    +
    k^{O(k)}ns^{6k-6}
    \log^{12k+O(1)}n(\log\log n)^{6k-6}.
\]
Using \(s\le n^{\theta_k}\), this is at most
\[
\begin{aligned}
    k^{O(k)}n^{1+(6k-6)\theta_k}(\log n)^{O(k^2)}
    =
    k^{O(k)}n^{E_k}(\log n)^{O(k^2)}
    \le
    k^{O(k^2)}n^{E_k}(\log n)^{O(k^2)}.
\end{aligned}
\]

\paragraph{Large-\(s\) branch: basic bounds.}
Assume \(s>n^{\theta_k}\), and let
\[
    N=\widetilde O(n/s)
\]
be the number of vertices of \(G'\).  Since \(G'\) is obtained by
contraction, also \(N\le n\).

We include all amplification costs, all polylogarithmic losses in the
bound on \(N\), and all parameter-dependent enumeration factors in the
common factor
\[
    k^{O(k^2)}(\log n)^{O(k^2)}.
\]
Indeed, the border-enumeration parameter \(p_i\) introduced below is
\(O(k)\), so its \(k^{O(p_i k)}\) dependence is covered.  Amplifying
each of the \(O(k)\) calls costs only an additional
\(k^{O(1)}\log n\) factor.  The amplified additive
\(O(|E(G')|)\) terms contribute at most
\[
    k^{O(1)}sn\log n
    \le k^{O(1)}n^2\log n,
\]
because
\[
    s\le2\lambda_k(G)\le2(k-1)n.
\]
This is also covered by the claimed bound.

If \(N<k\), each border-enumeration call outputs at most \(k^{O(k)}\)
terminal partitions.  For \(i\le2\), their total extension work is at
most \(k^{O(k)}n^2\).  For \(i\ge3\), it is at most
\[
    k^{O(k)}n^{\widehat\Phi(i)}
    \le
    k^{O(k)}n^{\widehat\Phi(k-1)}.
\]
These quantities are dominated by the bounds below.  We may
therefore assume for the remainder of the branch analysis that \(N\ge k\).

For \(i=1,2\), candidate borders remain implicit, so no \(O(n)\)-time
lift is performed per border.  For \(i\ge3\), materializing a lift costs
\(O(n)\), which is absorbed by the
\(n^{\widehat\Phi(i)}\) extension time.  The winning cut is materialized
only once.

The no-island branch uses \cref{thm:weighted_k_cut_sparse}.  Since
\(G'\) has \(N=\widetilde O(n/s)\) vertices and \(O(sn)\) edges, its
running time is at most
\[
\begin{aligned}
    O(sn)
    &+
    k^{O(k^2)}
    N^{k-2}(sn+N)\log^3(2N)\\
    &\le
    k^{O(k^2)}
    \widetilde O\left(\frac{n^{k-1}}{s^{k-3}}\right).
\end{aligned}
\]
Thus its exponent is at most
\[
    k-1-(k-3)\theta_k.
    \tag{4}\label{eq:no_island_exp}
\]

For every enumerated branch \(1\le i\le k-2\), put
\[
    p_i:=\beta_i k
        =k-i+\frac{2i}{\log n}.
\]
The fractional power satisfies
\[
    N^{2i/\log n}
    \le
    n^{2i/\log n}
    =
    \exp(O(i))
    =
    \exp(O(k)).
\]
It is therefore absorbed into \(k^{O(k^2)}\).  Apart from the endpoint
factor treated below which arises from the case distinction in \cref{lem:border_enumeration}, border enumeration consequently contributes
\[
    k^{O(k^2)}
    \left(\frac ns\right)^{k-i}
    (\log n)^{O(k^2)}
\]
time and candidates.

For \(i=1\), the one-island routine costs \(O(s+k^2)\) per retained
border.  Since \(s\ge1\), the resulting exponent is at most
\[
    k-1-(k-2)\theta_k.
    \tag{5}\label{eq:one_island_exp}
\]
This is smaller than \eqref{eq:no_island_exp} by \(\theta_k\).

For \(i=2\), if \(k=3\), then \(\ell=1\) and there is no border
enumeration.  The two-island routine costs
\[
    O(s^2+k^{O(1)})\le k^{O(1)}n^2.
\]
Suppose \(k\ge4\).  The two-island extension costs
\(O(s^2+k^{O(1)})\) per retained border, so its contribution is
\[
    \left(\frac ns\right)^{k-2}s^2
    =
    \frac{n^{k-2}}{s^{k-4}}.
\]
Its exponent is at most
\[
    k-2-(k-4)\theta_k.
    \tag{6}\label{eq:two_island_exp_new}
\]
This is smaller than \eqref{eq:no_island_exp} by
\(1-\theta_k>0\).

For \(3\le i\le k-2\), border enumeration followed by
\cref{lem:border_coupled_islands}, apart from a possible endpoint
factor, has exponent at most
\[
    k-i+\widehat\Phi(i)-(k-i)\theta_k.
    \tag{7}\label{eq:many_island_exp}
\]

Finally, when \(i=k-1\), there is no border enumeration, and the exponent
is
\[
    \widehat\Phi(k-1).
    \tag{8}\label{eq:last_island_exp}
\]

\paragraph{The endpoint factor.}
If \(p_i<3\), \cref{lem:border_enumeration} contributes the additional
factor
\[
    F_{k,N}
    :=
    \exp\!\left(
        O\!\left(
            k\log(k+2)\sqrt{\log(N+2)}
        \right)
    \right),
    \tag{9}\label{eq:endpoint_factor}
\]
because \(p_i<3\). Since \(N\le n\), Young's inequality gives, for every \(\delta>0\),
\[
\begin{aligned}
    F_{k,N}
    &\le
    n^{\delta/2}
    \exp\!\left(
        O\!\left(
            \frac{k^2\log^2(k+2)}{\delta}
        \right)
    \right).
\end{aligned}
\tag{10}\label{eq:endpoint_young}
\]

Because \(k-i\) is an integer at least \(2\), the inequality \(p_i<3\)
forces $k-i=2$.
Among the first two island branches, the only possible endpoint cases
are therefore \((k,i)=(3,1)\) and \((k,i)=(4,2)\).  Their polynomial
exponents and respective slacks below \(E_k\) are
\[
    E_3-(2-\theta_3)=\theta_3=\frac1{12}
\]
and
\[
    E_4-2=\frac{17}{19}.
\]
Applying \eqref{eq:endpoint_young} with either fixed slack absorbs the
endpoint factor.

For \(k=5\), the polynomial part of the endpoint branch has exponent
\[
    2+\widehat\Phi(3)-2\theta_5.
\]
Since \(26\theta_5=1+2.37134\),
\[
\begin{aligned}
    E_5-
    \bigl(2+\widehat\Phi(3)-2\theta_5\bigr)
    &=
    2.37134-\widehat\Phi(3)>0.
\end{aligned}
\]
Likewise, for \(k=6\),
\[
\begin{aligned}
    E_6-
    \bigl(2+\widehat\Phi(4)-2\theta_6\bigr)
    &=
    3.250386-\widehat\Phi(4)>0.
\end{aligned}
\]
These are fixed positive gaps.  Applying
\eqref{eq:endpoint_young} with the corresponding gap absorbs the
endpoint factor for \(k=5,6\).  The constants needed for these finitely
many fixed values of \(k\) are contained in \(k^{O(k^2)}\).

It remains to treat \(k\ge7\).  The chosen exponents satisfy
\[
    \widehat\Phi(r)
    <
    \frac45r+\frac1{10}
    \qquad(r\ge3).
    \tag{11}\label{eq:phi_endpoint_bound}
\]
For \(r=3,4,5\), this follows from
\[
    \widehat\Phi(3)<2.37134<\frac52,\qquad
    \widehat\Phi(4)<3.250386<\frac{33}{10},
\]
and
\[
    \widehat\Phi(5)<4.085989<\frac{41}{10}.
\]
The remaining cases follow by induction, since increasing \(r\) by
\(3\) increases \(\widehat\Phi(r)\) by
\[
    \widehat\Phi(3)<\frac{12}{5}.
\]

Let
\[
    B_k:=2+\widehat\Phi(k-2)-2\theta_k,
    \qquad
    \Delta_k:=E_k-B_k,
\]
and put
\[
    a:=3.250386.
\]
Using \eqref{eq:phi_endpoint_bound} and
\(\theta_k=(k+a-5)/(7k-10)\), we obtain
\[
\begin{aligned}
    \Delta_k
    >
    \frac12-\frac45k+(6k-4)\theta_k
    =
    \frac{
        4k^2+(60a-225)k+150-40a
    }{
        10(7k-10)
    }.
    \label{eq:endpoint_slack}
\end{aligned}
\]
The numerator on the right is positive at \(k=7\) and strictly
increasing thereafter.  Hence \(\Delta_k>0\) for every \(k\ge7\).
Moreover, for \(k\ge15\), direct comparison gives
\[
    \Delta_k
    >
    \frac{2k}{35}-\frac23
    \ge
    \frac{k}{100}.
\]

Apply \eqref{eq:endpoint_young} with \(\delta=\Delta_k\).  For
\(k\ge15\), its parameter-dependent multiplier is at most
\[
    \exp(O(k\log^2(k+2)))
    \le
    k^{O(k^2)}.
\]
For \(7\le k\le14\), the finitely many values of \(\Delta_k\) are fixed
and positive, so the multiplier is also absorbed into \(k^{O(k^2)}\).
Consequently, the full endpoint contribution is at most
\[
\begin{aligned}
    k^{O(k^2)}
    n^{B_k}F_{k,N}
    (\log n)^{O(k^2)}
    \le
    k^{O(k^2)}
    n^{B_k+\Delta_k/2}
    (\log n)^{O(k^2)}
    \le
    k^{O(k^2)}
    n^{E_k}
    (\log n)^{O(k^2)}.
\end{aligned}
\]
Thus every endpoint factor is absorbed uniformly in \(k\).

\paragraph{Optimization for \(k=3,4,5\).}
For \(k=3\), \(\theta_3=1/12\), and hence
\[
    E_3=1+12\theta_3=2.
\]
The no-island branch has exponent \(2\), the one-island branch has
exponent \(2-\theta_3<2\), and the two-island branch costs at most
\(k^{O(1)}n^2\).

For \(k=4\), \(\theta_4=2/19\), and hence
\[
    E_4=1+18\theta_4=\frac{55}{19}.
\]
The no-island branch has exponent
\[
    3-\theta_4=\frac{55}{19}.
\]
The one-island and two-island branches are smaller, and the \(i=3\)
branch has exponent
\[
    \widehat\Phi(3)<2.37134<\frac{55}{19}.
\]

For \(k=5\),
\[
    \theta_5=\frac{1+2.37134}{26},
\]
so
\[
    E_5
    =
    1+24\theta_5
    =
    1+\frac{12}{13}(1+2.37134).
\]
The \(i=3\) branch has polynomial exponent
\[
    2+\widehat\Phi(3)-2\theta_5,
\]
and its slack below \(E_5\) is
\[
    2.37134-\widehat\Phi(3)>0.
\]
Its possible endpoint factor was absorbed above.  The no-island branch
is dominated because
\[
    E_5-\bigl(4-2\theta_5\bigr)
    =
    2.37134-2>0.
\]
The two-island branch has exponent \(3-\theta_5\), which is smaller
than the no-island exponent, and the one-island branch is smaller still.
Finally,
\[
    \widehat\Phi(4)<3.250386<E_5,
\]
so the \(i=4\) branch is also dominated.  Therefore
\[
    E_5
    =
    1+\frac{12}{13}(1+2.37134)
    <
    4.112007.
\]

\paragraph{Optimization for \(k\ge6\).}
Put
\[
    a:=3.250386.
\]
Then
\[
    \theta_k=\frac{k+a-5}{7k-10}.
\]
The rounded upper bound for the \(i=4\) exponent is
\[
    k-4+a-(k-4)\theta_k.
\]
It balances the FPT exponent exactly:
\[
\begin{aligned}
    E_k-
    \bigl(k-4+a-(k-4)\theta_k\bigr)
    =
    (7k-10)\theta_k-(k+a-5)
    =0.
\end{aligned}
\]
The actual \(i=4\) exponent is strictly smaller because
\(\widehat\Phi(4)<a\).

The no-island branch is dominated because
\[
\begin{aligned}
    E_k-
    \bigl(k-1-(k-3)\theta_k\bigr)
    =
    (7k-9)\theta_k-(k-2)
    =
    \frac{
        (7a-20)k+25-9a
    }{
        7k-10
    }
    >0
\end{aligned}
\]
for every \(k\ge6\).  The one-island branch is smaller than the
no-island branch by \(\theta_k\), and the two-island branch is smaller
by \(1-\theta_k\).

For the \(i=3\) branch, use
\(\widehat\Phi(3)<2.37134\).  Then
\[
\begin{aligned}
    E_k-
    \bigl(k-3+\widehat\Phi(3)-(k-3)\theta_k\bigr)
    &>
    E_k-
    \bigl(k-3+2.37134-(k-3)\theta_k\bigr)\\
    &=
    \frac{
        (7a-7(2.37134)-6)k
        +10(2.37134)-9a+5
    }{
        7k-10
    }
    >0
\end{aligned}
\]
for \(k\ge6\).

For \(5\le i\le k-2\), \eqref{eq:phi_tail_bound_main} gives
\[
    \widehat\Phi(i)
    <
    a+\frac{6(i-4)}7.
\]
Consequently,
\[
\begin{aligned}
    E_k-
    \bigl(k-i+\widehat\Phi(i)-(k-i)\theta_k\bigr)
    >
    E_k-
    \left(
        k-i+a+\frac{6(i-4)}7-(k-i)\theta_k
    \right)
    =
    \frac{
        (25-7a)(i-4)
    }{
        7(7k-10)
    }
    >0.
\end{aligned}
\]
This handles every nonendpoint branch in the range; the possible
endpoint \(i=k-2\) was handled uniformly above.

Finally,
\[
    \widehat\Phi(k-1)
    <
    a+\frac{6(k-5)}7.
\]
Therefore
\[
\begin{aligned}
    E_k-\widehat\Phi(k-1)
    >
    1+(6k-6)\frac{k+a-5}{7k-10}
      -a-\frac{6(k-5)}7
    =
    \frac{
        (67-7a)k+28a-160
    }{
        7(7k-10)
    }
    >0.
\end{aligned}
\]
Thus every large-\(s\) branch is bounded by
\[
    k^{O(k^2)}n^{E_k}(\log n)^{O(k^2)}.
\]
Combining this with the preprocessing and small-\(s\) bounds proves the
theorem for every \(3\le k\le n\).
\end{proof}

\section*{Acknowledgement} The first author would like to thank Amir Abboud, Shyan Akmal, Thatchaphol Saranurak, and Yinzhan Xu for early discussions on minimum $3$-cut. The second author used OpenAI’s GPT 5.5 Sol on Max effort to assist with drafting and revising portions of the exposition. The authors assume responsibility for all content.

\bibliographystyle{alpha}
\bibliography{references}

\appendix

\section{Deferred Proofs for the Unbreakable Decomposition}
\label{app:unbreakable_decomposition}

\subsection{Compression of the Raw Quotient}
\label{app:compression}

\begin{proof}[Proof of \cref{lem:compressed_td}]
Starting from \(Q_{\rm raw}\), repeatedly delete every leaf component \(C\) with \(V_C=\emptyset\). Then remove every degree-two component \(C\) with \(V_C=\emptyset\) and connect its endpoints with an edge. Let the resulting tree be \(Q\). Thus every node \(C\in V(Q)\) is either terminal, meaning \(V_C\neq\emptyset\), or branching, meaning it has degree at least \(3\) in \(Q\).

Since the nonempty sets \(V_C\) are disjoint and each contains at least one graph vertex, there are at most \(n\) terminal nodes. A tree with \(k\) terminal or marked nodes and all unmarked nodes of degree at least \(3\) has \(O(k)\) nodes. Hence
\[
    |V(Q)|+|E(Q)|=O(n).
\]
Each edge \(\varepsilon=CD\in E(Q)\) corresponds to a path in the raw quotient
\[
    C=C_0,e_1,C_1,e_2,\ldots,e_r,C_r=D,
\]
where each \(e_i\in F\) is a low edge of the original hierarchical tree. For the endpoint \(C\), call \(e_1\) the first raw low edge of \(\varepsilon\) incident with \(C\). Similarly, \(e_r\) is the first raw low edge of \(\varepsilon\) incident with \(D\).

For a retained node \(C\in V(Q)\), let \(I(C)\) be the set of first raw low edges over all compressed edges of \(Q\) incident with \(C\). Define
\[
    \beta(C) := V_C \cup \bigcup_{e\in I(C)}\mu(e).
\]
Observe that every \(e\in I(C)\) is incident with \(C\) in \(Q_{\rm raw}\). Therefore
\[
    \beta(C)\subseteq \beta_{\rm raw}(C).
\]

We first prove vertex connectedness directly. Fix \(v\in V\), and let \(C_v\in V(Q)\) be the retained node containing the singleton leaf \(v\). Let
\[
    T'_v := \{C_v\} \cup \bigcup_{vu\in E(G)} V\bigl(P_Q(C_v,C_u)\bigr),
\]
where \(P_Q(C_v,C_u)\) is the unique path from \(C_v\) to \(C_u\) in the compressed tree \(Q\). We prove that \(T'_v\) is exactly the set of compressed bags containing \(v\). First let \(C\in T'_v\). If \(C=C_v\), then $v\in V_C\subseteq \beta(C)$. Otherwise, there is an edge \(vu\in E(G)\) such that \(C\) lies on \(P_Q(C_v,C_u)\). Let \(\varepsilon\) be the compressed edge incident with \(C\) in the direction of \(C_v\). The first raw low edge of \(\varepsilon\) incident with \(C\) lies on the raw tree path between the singleton leaves \(v\) and \(u\). Therefore the graph edge \(vu\) crosses the cut represented by this raw low edge, and so \(v\) belongs to the corresponding middle set included in \(\beta(C)\).

Conversely, suppose \(v\in\beta(C)\). If \(v\in V_C\), then \(C=C_v\), so \(C\in T'_v\). Otherwise, \(v\in\mu(e)\) for some \(e\in I(C)\). By definition of \(I(C)\), the edge \(e\) is the first raw low edge of some compressed edge incident with \(C\). Since \(v\in\mu(e)\), there exists an edge \(vu\in E(G)\) crossing the cut \(S_e\). Thus \(e\) lies on the raw path between the singleton leaves \(v\) and \(u\). Since \(e\) is incident with the retained endpoint \(C\) of its compressed path segment, the retained node \(C\) lies on the compressed path \(P_Q(C_v,C_u)\). Hence \(C\in T'_v\).

Thus the compressed bags containing \(v\) are exactly \(T'_v\). Since \(T'_v\) is a union of paths all containing \(C_v\), it is connected. Next let \(uv\in E(G)\). If \(C_u=C_v\), then $u,v\in V_{C_u}\subseteq \beta(C_u)$. Otherwise, consider the first compressed edge on the path from \(C_u\) to \(C_v\). Let \(e\) be its first raw low edge incident with \(C_u\). Since \(e\) lies on the raw tree path between the singleton leaves \(u\) and \(v\), the edge \(uv\) crosses the cut represented by \(e\). Hence $u,v\in\mu(e)\subseteq \beta(C_u)$, so \(uv\) is covered.

Finally we prove the adhesion bound. Let \(\varepsilon=CD\in E(Q)\), and let
\[
    C=C_0,e_1,C_1,e_2,\ldots,e_r,C_r=D
\]
be the corresponding raw quotient path. Suppose
\[
    v\in \beta(C)\cap\beta(D).
\]
Because \(\beta(C)\subseteq\beta_{\rm raw}(C)\) and \(\beta(D)\subseteq\beta_{\rm raw}(D)\), the raw occurrence subtree of \(v\) contains both \(C\) and \(D\). By \cref{lem:raw_td}, that raw occurrence set is connected, so it contains the entire raw path from \(C\) to \(D\). In particular it contains both endpoints of the raw edge \(e_1\). Applying the adhesion statement of \cref{lem:raw_td} to that raw edge gives $v\in \mu(e_1)$. Therefore $\beta(C)\cap\beta(D)\subseteq \mu(e_1)$. By \cref{lem:low_edge_middle_small}, $|\beta(C)\cap\beta(D)| \le |\mu(e_1)| \le 2\Lambda$.
\end{proof}

\subsection{Compactification and Redundancy Removal}
\label{app:compactification}
We now give a cleanup algorithm which makes a rooted tree decomposition compact and suppresses redundant bags. We use the convention that the root adhesion is empty, i.e.
\[
    \sigma_\tau(r):=\emptyset .
\]
Let
\[
    M:=|V(\tau)|+|E(G)|+\sum_{t\in V(\tau)}|\beta(t)|, \qquad \eta:=\max\{|\sigma_\tau(t)|:t\ne r\},
\]
with \(\eta:=0\) if \(\tau\) has only one node.

For every vertex \(v\in V(G)\), let \(\operatorname{top}(v)\) be the minimum-depth node \(t\) such that \(v\in\beta(t)\). This is well-defined because the bags containing \(v\) form a connected subtree.

\begin{lemma}[Activation Identity]
\label{lem:activation_identity}
For every node \(t\in V(\tau)\),
\[
    \alpha_\tau(t) = \{v\in V(G):\operatorname{top}(v)\in V(\tau_t)\}.
\]
Moreover, if \(uv\in E(G)\), then \(\operatorname{top}(u)\) and \(\operatorname{top}(v)\) are comparable in \(\tau\). If \(a(uv)\) is the shallower of these two nodes, then
\[
    uv\in E(G[\alpha_\tau(t)]) \qquad\Longleftrightarrow\qquad a(uv)\in V(\tau_t).
\]
\end{lemma}

\begin{proof}
The first identity follows directly from connectedness of the bags containing a fixed vertex: a vertex appears in \(\gamma_\tau(t)\) but not in \(\sigma_\tau(t)\) exactly when its topmost occurrence lies in the subtree of \(t\).

For an edge \(uv\), some bag contains both endpoints. Therefore \(\operatorname{top}(u)\) and \(\operatorname{top}(v)\) are both ancestors of that bag, and hence are comparable. By the first identity, both endpoints lie in \(\alpha_\tau(t)\) exactly when both top nodes lie in \(V(\tau_t)\). Since the two top nodes are comparable, this is equivalent to their shallower top node \(a(uv)\) lying in \(V(\tau_t)\).
\end{proof}

Thus \(G[\alpha_\tau(t)]\) is obtained by a monotone process: vertex \(v\) is inserted at \(\operatorname{top}(v)\), and edge \(uv\) is inserted at \(a(uv)\).

\begin{algorithm}[H]
\caption{\textsc{CompactifyAndSuppress}$(G,(\tau,\beta))$}
\label{alg:compactify_and_suppress}
\begin{algorithmic}[1]
\State Compute \(\operatorname{top}(v)\) for every \(v\in V(G)\).
\State Assign every edge \(uv\in E(G)\) to \(a(uv)\), the shallower of \(\operatorname{top}(u)\) and \(\operatorname{top}(v)\).
\State Process the nodes of \(\tau\) in postorder using union-find.
\State At node \(t\), activate all vertices \(v\) with \(\operatorname{top}(v)=t\), and all edges assigned to \(t\).
\State For every final union-find component \(C\) whose vertex set was not already represented by a component before processing \(t\), create an event node \(x_C\). Its children are the maximal event nodes, created before processing \(t\), whose represented sets are proper subsets of \(C\).
\State When \(x_C\) is created, materialize its candidate bag
\[
    B(x_C) := N_G(C) \cup \bigl(C\cap\beta(t)\bigr) \cup \bigcup_{x_{C_i}\text{ child of }x_C} N_G(C_i).
\]
Also store \(N_G(C)\).
\State If the event forest has several roots, add a dummy root with empty candidate bag.
\State Traverse the event forest top-down. Keep the root. For every other event node \(x\), let \(a\) be its nearest kept ancestor. If \(B(x)\subseteq B(a)\), suppress \(x\); otherwise keep \(x\) and attach it as a child of \(a\) in the output tree.
\State Return the kept nodes with bags \(B(x)\).
\end{algorithmic}
\end{algorithm}

Let the returned decomposition be \((\widehat\tau,\widehat\beta)\).

\begin{lemma}[Correctness of the Cleanup]
\label{lem:compactify_correctness}
\cref{alg:compactify_and_suppress} returns a compact rooted tree decomposition \((\widehat\tau,\widehat\beta)\) of \(G\). Moreover, every output bag is a subset of some input bag, every nonempty output adhesion is a subset of some input adhesion, and no non-root output bag is contained in its parent bag. If the input decomposition tree \(\tau\) has depth \(d\), then the output tree \(\widehat\tau\) has depth at most \(d+1\), where the additive \(1\) is only for the possible dummy root.
\end{lemma}

\begin{proof}
We first analyze the event forest before suppression. If \(x_C\) is an event node, let \(t_C\) be its birth node. By \cref{lem:activation_identity}, \(C\) is a connected component of \(G[\alpha_\tau(t_C)]\).

We use two elementary facts. First, every event node born at a node \(t\) contains a vertex \(v\) with \(\operatorname{top}(v)=t\). Indeed, an event node is created only when a component is new or changes while processing \(t\). If a vertex is activated, the claim is immediate for the new component containing it. If the change is caused by an activated edge, then that edge is assigned to \(t\), so one of its endpoints has top node \(t\), and this endpoint lies in the new component.

Second, if \(x_C\) is a child of \(x_D\), then \(t_C\) is a proper descendant of \(t_D\). Since \(x_C\) was already present before processing \(t_D\), and since \(C\subseteq D\subseteq\alpha_\tau(t_D)\), the previous paragraph gives a vertex \(v\in C\) with \(\operatorname{top}(v)=t_C\). \cref{lem:activation_identity} gives \(t_C\in V(\tau_{t_D})\). Moreover \(t_C\ne t_D\), because \(x_C\) already existed before \(t_D\) was processed.

We prove the following invariant for every event node \(x_C\). Let
\[
    U_C:=\bigcup_{x\text{ in the subtree of }x_C}B(x).
\]
Then
\[
    B(x_C)\subseteq \beta(t_C), \qquad U_C=C\cup N_G(C),
\]
and, if \(x_C\) has parent \(x_D\), then
\[
    B(x_D)\cap U_C=N_G(C), \qquad B(x_C)\cap B(x_D)=N_G(C).       \tag{1}
\]

First we show \(B(x_C)\subseteq\beta(t_C)\). Let \(y\in N_G(C)\), and choose \(xy\in E(G)\) with \(x\in C\). Since \(C\) is a component of \(G[\alpha_\tau(t_C)]\), we have \(y\notin\alpha_\tau(t_C)\). Also every bag containing \(x\) lies in the subtree \(\tau_{t_C}\): otherwise connectedness of the bags containing \(x\) would force \(x\in\sigma_\tau(t_C)\), contradicting \(x\in\alpha_\tau(t_C)\). Hence any edge-covering bag for \(xy\) lies in \(\tau_{t_C}\), so \(y\in\gamma_\tau(t_C)\). Since \(y\notin\alpha_\tau(t_C)\),
\[
    y\in\gamma_\tau(t_C)\setminus\alpha_\tau(t_C) = \sigma_\tau(t_C).
\]
Thus
\[
    N_G(C)\subseteq\sigma_\tau(t_C)\subseteq\beta(t_C).    \tag{2}
\]

Now let \(x_{C_i}\) be a child of \(x_C\). We claim that
\[
    N_G(C_i)\subseteq\beta(t_C).                           \tag{3}
\]
Let \(y\in N_G(C_i)\), and choose \(xy\in E(G)\) with \(x\in C_i\). If \(y\notin C\), then \(y\in N_G(C)\), so \(y\in\beta(t_C)\) by (2). Otherwise \(y\in C\setminus C_i\). Since \(x,y\in C\subseteq\alpha_\tau(t_C)\), the activation node \(a(xy)\) lies in the subtree of \(t_C\). If \(a(xy)\) were a proper descendant of \(t_C\), then \(xy\) would already have been active before processing \(t_C\). In that case both endpoints would already have been represented before processing \(t_C\), and the edge \(xy\) would put them in the same previous union-find component, contradicting the maximality of the child component \(C_i\). Hence \(a(xy)=t_C\). Since \(x\in C_i\) was already represented before processing \(t_C\), \(\operatorname{top}(x)\ne t_C\). As \(a(xy)\) is the shallower of \(\operatorname{top}(x)\) and \(\operatorname{top}(y)\), it follows that \(\operatorname{top}(y)=t_C\), and hence \(y\in\beta(t_C)\). This proves (3), and therefore \(B(x_C)\subseteq\beta(t_C)\). Thus every non-dummy event bag is contained in an input bag; the dummy root, if present, has empty bag.

Next we prove
\[
    U_C=C\cup N_G(C).
\]
The inclusion \(C\cup N_G(C)\subseteq U_C\) follows by induction over the event forest. The bag \(B(x_C)\) contains \(N_G(C)\) and \(C\cap\beta(t_C)\). Every vertex of \(C\) whose top node is not \(t_C\) was already active before processing \(t_C\), and therefore lies in one of the maximal child components \(C_i\). By induction, it lies in the subtree of \(x_{C_i}\). Conversely, each child subtree contributes only
\[
    C_i\cup N_G(C_i)\subseteq C\cup N_G(C),
\]
because \(C_i\subseteq C\), and any vertex outside \(C\) adjacent to \(C_i\) lies in \(N_G(C)\). Also \(B(x_C)\subseteq C\cup N_G(C)\). Thus \(U_C=C\cup N_G(C)\).

Now suppose \(x_D\) is the parent of \(x_C\). Since \(x_C\) is a child of \(x_D\), the bag \(B(x_D)\) contains \(N_G(C)\), and \(B(x_C)\) also contains \(N_G(C)\). Hence
\[
    N_G(C)\subseteq B(x_D)\cap U_C \quad\text{and}\quad N_G(C)\subseteq B(x_C)\cap B(x_D).
\]
We show that no vertex of \(C\) lies in \(B(x_D)\). The parent bag has the form
\[
    B(x_D) = N_G(D) \cup \bigl(D\cap\beta(t_D)\bigr) \cup \bigcup_{x_Y\text{ child of }x_D}N_G(Y).
\]
A vertex \(z\in C\) cannot lie in \(N_G(D)\), since \(C\subseteq D\). It cannot lie in \(D\cap\beta(t_D)\), because \(t_D\) is a proper ancestor of \(t_C\), and an occurrence of \(z\in C\subseteq\alpha_\tau(t_C)\) at \(t_D\) would force \(z\in\sigma_\tau(t_C)\), contradicting \(z\in\alpha_\tau(t_C)\). Finally, \(z\) cannot lie in \(N_G(Y)\) for a sibling child \(x_Y\ne x_C\). Otherwise there is an edge \(zy\) with \(y\in Y\). Both endpoints were already represented before processing \(t_D\), so neither has top node \(t_D\). Since both endpoints lie in \(D\subseteq\alpha_\tau(t_D)\), the activation node \(a(zy)\) is a proper descendant of \(t_D\). Thus \(zy\) was already active before processing \(t_D\), contradicting that \(C\) and \(Y\) were distinct maximal child components of \(x_D\).

Therefore \(B(x_D)\cap C=\emptyset\). Since \(U_C=C\cup N_G(C)\) and \(N_G(C)\subseteq B(x_D)\), we get
\[
    B(x_D)\cap U_C=N_G(C).
\]
Since \(B(x_C)\subseteq U_C\) and \(N_G(C)\subseteq B(x_C)\), this also gives
\[
    B(x_C)\cap B(x_D)=N_G(C).
\]
This proves the invariant.

The invariant gives compactness of the event decomposition. For every non-root event node \(x_C\),
\[
    \alpha(x_C) = U_C\setminus B(\operatorname{parent}(x_C)) = (C\cup N_G(C))\setminus N_G(C) = C,
\]
and
\[
    \sigma(x_C) = B(x_C)\cap B(\operatorname{parent}(x_C)) = N_G(C) = N_G(\alpha(x_C)).
\]
For a root event node, after all nodes of \(\tau\) have been processed, the active graph is all of \(G\). Hence each root event node represents a connected component \(C\) of \(G\), and so \(N_G(C)=\emptyset\). Thus the same compactness identity holds for roots, with empty adhesion. If a dummy root is added, its bag and adhesion are also empty.

We now verify the tree-decomposition axioms for the event forest. Vertex coverage follows from the identity \(U_C=C\cup N_G(C)\) applied to the roots of the event forest, whose represented sets are the connected components of \(G\).

For edge coverage, let \(uv\in E(G)\), and suppose without loss of generality that \(\operatorname{top}(u)\) is an ancestor of \(\operatorname{top}(v)\). If \(\operatorname{top}(u)=\operatorname{top}(v)=t\), then both endpoints are activated at \(t\), and the event node born at \(t\) for their component has a bag containing both endpoints in \(C\cap\beta(t)\). Otherwise, let \(t=\operatorname{top}(v)\), and let \(x_C\) be the event node born at \(t\) for the component containing \(v\). Then \(v\in C\cap\beta(t)\), while \(u\notin\alpha_\tau(t)\), because \(\operatorname{top}(u)\) is a proper ancestor of \(t\). Thus \(u\notin C\), and since \(uv\in E(G)\), we have \(u\in N_G(C)\). Hence \(u,v\in B(x_C)\).

For vertex connectedness, fix \(v\in V(G)\), and let \(x_v\) be the event node born at \(t_v:=\operatorname{top}(v)\) for the component containing \(v\) after processing \(t_v\). This event node exists because \(v\) is activated at \(t_v\). We show that every event bag containing \(v\) is connected to \(x_v\) through bags containing \(v\).

Let \(v\in B(x_C)\). If \(v\in C\cap\beta(t_C)\), then \(v\in C\subseteq\alpha_\tau(t_C)\) and \(v\in\beta(t_C)\), so \(\operatorname{top}(v)=t_C\), and hence \(x_C=x_v\).

If \(v\in N_G(C)\), choose an edge \(vw\) with \(w\in C\). Since \(v\notin \alpha_\tau(t_C)\), the node \(t_v=\operatorname{top}(v)\) is a proper ancestor of \(t_C\). Follow the parent chain upward from \(x_C\) until the first event node whose represented component contains \(v\). Such a node exists: when \(t_v\) is processed, the vertex \(v\) is activated and the edge \(vw\) is also active, since its activation node is \(t_v\). Before this first merge, the represented component contains \(w\), remains adjacent to \(v\) through the edge \(vw\), and does not contain \(v\). Hence \(v\) belongs to the boundary of every represented component on this chain and appears in every corresponding bag. From the first event node whose represented component contains \(v\), the path continues along the chain of event nodes whose represented components contain \(v\), which connects to \(x_v\). Thus \(x_C\) is connected to \(x_v\) through bags containing \(v\).

Finally, if \(v\in N_G(C_i)\) for a child \(x_{C_i}\) of \(x_C\), then \(v\in B(x_{C_i})\). Applying the previous cases to \(x_{C_i}\), the node \(x_{C_i}\) is connected to \(x_v\) by a path of bags containing \(v\). Since \(B(x_C)\) also contains \(v\), this connects \(x_C\) to the same subtree. Thus all event bags containing \(v\) form a connected subtree.

The depth bound for the event forest follows from the birth-node fact above: if \(x_C\) is a child of \(x_D\), then \(t_C\) is a proper descendant of \(t_D\). Hence birth nodes strictly descend along every root-to-leaf path in the event forest, so the event forest has depth at most the depth \(d\) of \(\tau\), except for the possible additive \(1\) from a dummy root.

It remains to justify suppression. A node \(x\) is suppressed only when \(B(x)\subseteq B(a)\), where \(a\) is its nearest kept ancestor. Deleting such a node preserves edge coverage because any edge covered using \(B(x)\) is also covered by \(B(a)\). It preserves vertex connectedness because, for each vertex, every deleted bag containing that vertex is replaced by an ancestor bag containing the same vertex, so the subtree of bags containing that vertex is only contracted.

Let \(u\) be a kept node, and let \(a\) be its new parent after suppression. Let
\[
    a=t_0,t_1,\ldots,t_r,u
\]
be the corresponding path in the original event forest, where \(t_1,\ldots,t_r\) were deleted. If \(r=0\), the adhesion of \(u\) is unchanged. If \(r>0\), then \(B(t_r)\subseteq B(a)\). By vertex connectedness in the event decomposition, every vertex in \(B(u)\cap B(a)\) appears in every bag on the path from \(a\) to \(u\), and therefore in \(B(t_r)\). Thus
\[
    B(u)\cap B(a)\subseteq B(u)\cap B(t_r).
\]
The reverse inclusion follows from \(B(t_r)\subseteq B(a)\), so
\[
    B(u)\cap B(a)=B(u)\cap B(t_r).
\]
Thus every output adhesion is exactly the corresponding event-forest adhesion, and is therefore either empty or equal to some \(N_G(C)\subseteq\sigma_\tau(t_C)\). Hence every nonempty output adhesion is a subset of an input adhesion.

Suppression also preserves compactness. Fix a kept node \(u\), and let
\[
    S_u:=\bigcup_{z\text{ in the original event-forest subtree of }u}B(z).
\]
The union of bags in the output subtree rooted at \(u\) is still \(S_u\): every deleted descendant has its bag contained in its nearest kept ancestor, which is itself in the output subtree of \(u\). Moreover,
\[
    S_u\cap B(a)=B(u)\cap B(a).
\]
Indeed, if a vertex belongs to both \(S_u\) and \(B(a)\), then it appears in some bag in the original subtree of \(u\) and in the ancestor bag \(B(a)\); by vertex connectedness in the event decomposition, it appears on the whole path between these bags, in particular in \(B(u)\). Therefore the new lower side of \(u\) is
\[
    S_u\setminus B(a) = S_u\setminus (B(u)\cap B(a)),
\]
which is the same as the old lower side by the adhesion identity above. Since compactness held in the event decomposition, it continues to hold after suppression.

Every non-dummy output bag is an event bag, hence is a subset of some input bag; the dummy root has empty bag. Suppression only contracts paths, so it cannot increase depth, and the dummy root contributes at most one additional level. Finally, by construction, a non-root node is kept only if its bag is not contained in the bag of its nearest kept ancestor, which is its parent in the output tree. Therefore no non-root output bag is contained in its parent bag.
\end{proof}

\begin{lemma}[Size and running time]
\label{lem:compactify_time}
\cref{alg:compactify_and_suppress} can be implemented in
\[
    \widetilde O(M+\eta n)
\]
time and space. The returned decomposition satisfies
\[
    |V(\widehat\tau)|\le n+1, \qquad \sum_{x\in V(\widehat\tau)}|\widehat\beta(x)| = O((\eta+1)n).
\]
In particular, if all input adhesions have size \(O(\Lambda)\) where $\Lambda \ge 1$, then the cleanup costs
\[
    \widetilde O(M+\Lambda n)
\]
time and outputs total bag volume \(O(\Lambda n)\).
\end{lemma}

\begin{proof}
Computing all top nodes and assigning edges to activation nodes takes \(O(M)\) time after computing depths in \(\tau\). The bottom-up activation process uses union-find. Each vertex is activated once, each edge is processed once, and each successful union decreases the number of components.

An event node born at \(t\) contains at least one vertex activated at \(t\): if a new component is created because of an activated edge, that edge has an endpoint whose top node is \(t\). Distinct event nodes born at the same \(t\) contain disjoint vertices activated at \(t\). Hence the number of non-dummy event nodes is at most \(n\).

For an event node \(x_C\) born at \(t\), we have
\[
    C\cap\beta(t)=\{v\in C:\operatorname{top}(v)=t\}.
\]
Indeed, \(C\subseteq\alpha_\tau(t)\), so no vertex of \(C\cap\beta(t)\) lies in \(\sigma_\tau(t)\); hence its topmost occurrence is \(t\). Thus, for all event nodes born at a fixed \(t\), the sets \(C\cap\beta(t)\) are formed by placing each vertex activated at \(t\) into its final union-find component after all activations at \(t\) have been processed.

We bound the total volume of all candidate event bags before suppression. For an event node \(x_C\) born at \(t_C\),
\[
    |N_G(C)|\le |\sigma_\tau(t_C)|\le \eta
\]
for \(t_C\ne r\), and the same bound holds for \(t_C=r\) because then \(N_G(C)=\emptyset\). Thus the total contribution of the \(N_G(C)\) terms is \(O(\eta n)\). Also,
\[
    C\cap\beta(t_C)\subseteq \beta(t_C)\setminus\sigma_\tau(t_C),
\]
and these newly introduced vertices are disjoint over all event nodes, so their total contribution is at most \(n\). Finally, each child boundary \(N_G(C_i)\) is charged once to the unique event edge from \(x_{C_i}\) to its parent, and has size at most \(\eta\). Since the event forest has \(O(n)\) edges, these child-boundary terms contribute \(O(\eta n)\). Therefore
\[
    \sum_x |B(x)|=O((\eta+1)n).
\]

It remains to explain how the sets \(N_G(C)\) and the child lists are maintained within the claimed time. For each union-find component, maintain a list of its current maximal represented child event nodes, a dictionary for its boundary vertices, and a membership list of active vertices in the component. When two components are united, their child lists and boundary dictionaries are melded small-to-large. Boundary entries whose vertex has become internal are deleted lazily when encountered; each such stale entry is generated by an original edge endpoint and is moved only \(O(\log n)\) times by small-to-large melding. Thus the total cost of maintaining and cleaning these dictionaries is \(\widetilde O(M)\). When an event node is created, its children are precisely the current maximal represented event nodes stored for the final union-find component, and after the event is created this child list is replaced by the single new event node.

Materializing and deduplicating all candidate bags costs
\[
    \widetilde O\!\left(\sum_x |B(x)|\right) = \widetilde O((\eta+1)n).
\]
For redundancy suppression, store each candidate bag as a membership dictionary. The top-down traversal tests \(B(x)\subseteq B(a)\) by scanning \(B(x)\) and querying the dictionary of the nearest kept ancestor \(a\). Since every candidate bag is scanned once, this costs
\[
    \widetilde O\!\left(\sum_x |B(x)|\right) = \widetilde O((\eta+1)n).
\]
Thus the total running time and space are
\[
    \widetilde O(M+\eta n).
\]

The final node bound \(|V(\widehat\tau)|\le n+1\) follows from the absence of redundant parent-child pairs. Every kept non-root node \(x\) has a vertex in
\[
    \widehat\beta(x)\setminus \widehat\beta(p(x)).
\]
Choose \(v_x\in\widehat\beta(x)\setminus\widehat\beta(p(x))\). Since the bags containing \(v_x\) form a connected subtree and \(p(x)\) does not contain \(v_x\), the node \(x\) is the root of the connected subtree of output bags containing \(v_x\). A fixed vertex has only one such root, so no vertex is charged to two kept non-root nodes. Therefore there are at most \(n\) kept non-root nodes, plus the possible root, giving \(|V(\widehat\tau)|\le n+1\).
\end{proof}

\begin{corollary}
\label{cor:compactify_preserves_parameters}
If every input bag is \((q,s)\)-edge-unbreakable, then every output bag is \((q,s)\)-edge-unbreakable. If every input adhesion has size at most \(\eta\), then every output adhesion has size at most \(\eta\).
\end{corollary}

\begin{proof}
Every non-dummy output bag is a subset of an input bag, and the dummy root has empty bag, so edge-unbreakability is preserved. Every output adhesion is either empty or a subset of an input adhesion, so adhesion size is preserved.
\end{proof}

\section{Rectangular Minimum Vertex-Weighted Triangle}
\label[appendix]{app:rectangular_vertex_weighted_triangles}

This appendix extends the algorithm of
Akmal and Fischer~\cite{AF26} to tripartite graphs with unequal part
sizes.  We use the following standard consequence of Sch\"onhage's
rectangular asymptotic sum inequality; see, for example,
\cite[Theorem~3.2]{ADWXXZ25}.

\begin{lemma}[Rectangular scaling inequality]
\label{lem:rectangular_scaling_inequality}
For \(a,b,c\ge0\) and \(0\le d\le\min\{a,b,c\}\),
\[
    2d+\omega(a-d,b-d,c-d)\le\omega(a,b,c).
\]
\end{lemma}

\begin{proof}
Write \(\langle x,y,z\rangle\) for the tensor of multiplying an
\(x\times y\) matrix by a \(y\times z\) matrix, and write
\(\widetilde R\) and \(\widetilde Q\) for asymptotic rank and
asymptotic subrank, respectively.  Up to immaterial rounding of the
dimensions,
\[
    \langle q^a,q^b,q^c\rangle
    =
    \langle q^d,q^d,q^d\rangle
    \otimes
    \langle q^{a-d},q^{b-d},q^{c-d}\rangle.
\]
The square matrix-multiplication tensor has asymptotic subrank
\[
    \widetilde Q
    \bigl(\langle q^d,q^d,q^d\rangle\bigr)
    =q^{2d};
\]
see, for example, \cite{CVZ21}.  Consequently, up to a factor
\(q^{o(N)}\), the \(N\)-th power of the tensor above restricts to a
direct sum of \(q^{2dN-o(N)}\) copies of
\[
    \langle
        q^{N(a-d)},q^{N(b-d)},q^{N(c-d)}
    \rangle.
\]
Indeed, tensoring a diagonal tensor of size \(t\) with a tensor \(S\)
gives the direct sum of \(t\) copies of \(S\).

The rectangular asymptotic sum inequality
\cite[Theorem~3.2]{ADWXXZ25} therefore lower-bounds the asymptotic rank
of this power by
\[
    q^{N(2d+\omega(a-d,b-d,c-d)-o(1))}.
\]
Taking \(N\)-th roots and logarithms base \(q\), and using
\[
    \omega(a,b,c)
    =
    \log_q
    \widetilde R
    \bigl(\langle q^a,q^b,q^c\rangle\bigr),
\]
proves the claim.
\end{proof}

We first consider the decision version.  A \emph{zero-sum triangle} is
a triangle whose three vertex weights sum to zero.

\begin{lemma}[Rectangular zero-sum triangle]
\label{lem:rectangular_zero_sum_triangle}
Let \(F\) be a vertex-weighted tripartite graph with \(O(\log n)\)-bit
integer vertex weights and parts \(A,B,C\) satisfying
\[
    |A|\le n^a,\qquad |B|\le n^b,\qquad |C|\le n^c.
\]
For every fixed
\(\widehat\omega(a,b,c)>\omega(a,b,c)\), whether \(F\) contains a
zero-sum triangle can be determined in
\(O(n^{\widehat\omega(a,b,c)})\) time.
\end{lemma}

\begin{proof}
By permuting the parts, assume \(a\le b\le c\), and put
\[
    N_A:=n^a,\qquad N_B:=n^b,\qquad N_C:=n^c.
\]
Write \(W\) for the vertex-weight function.  For
\(U\in\{A,B,C\}\), let
\[
    f_U(x):=
    |\{u\in U:W(u)=x\}|,
    \qquad
    \nu_U(x):=\frac{f_U(x)}{N_U}.
\]
We run the construction below twice, for
\[
    (U,V,Z)=(A,B,C)
    \qquad\text{and}\qquad
    (U,V,Z)=(B,A,C).
\]

Fix a weight \(x\) occurring in \(U\), put
\(p:=\nu_U(x)\), and skip this weight if \(p<n^{-a}\).  Let
\[
    \mathcal W_V(x)
    :=
    \{y:f_V(y)>0\text{ and }\nu_V(y)\le p\}.
\]
Scan \(\mathcal W_V(x)\) in an arbitrary fixed order and partition it
by next fit into maximal groups \(P\) satisfying
\[
    \sum_{y\in P}f_V(y)\le2pN_V.
\]
Every group other than possibly the last has total frequency greater
than \(pN_V\): otherwise, the next weight, whose frequency is at most
\(pN_V\), could have been added.  Thus there are \(O(1/p)\) groups.

For each group \(P\), define
\begin{align*}
    U_x&:=\{u\in U:W(u)=x\},\\
    V_P&:=\{v\in V:W(v)\in P\},\\
    Z_{x,P}&:=\{z\in Z:-x-W(z)\in P\}.
\end{align*}
The sets \(Z_{x,P}\) are disjoint as \(P\) varies.  They can be
constructed by storing the groups in a lookup table and scanning \(Z\)
once.  Split each \(Z_{x,P}\) into blocks of size
\(\lceil pN_Z\rceil\), except that the last block may be smaller.
Because \(p\ge n^{-a}\) and \(a\le c\), we have \(pN_Z\ge1\).
Consequently, every block has size at most \(2pN_Z\), and the total
number of blocks over all \(P\) is
\[
    O\left(
        \sum_P
        \left(
            \frac{|Z_{x,P}|}{pN_Z}+1
        \right)
    \right)
    =O(1/p).
\]

For every resulting block \(Z'\), delete from the graph induced by
\(U_x\cup V_P\cup Z'\) each edge
\(vz\in E(V_P,Z')\) for which
\[
    x+W(v)+W(z)\ne0,
\]
and run ordinary tripartite triangle detection.  Since every vertex
of \(U_x\) has weight \(x\), the remaining triangles are exactly the
zero-sum triangles in \(U_x\times V_P\times Z'\).

For correctness, let
\((u_A,u_B,u_C)\in A\times B\times C\) be a zero-sum triangle.  Choose
\(U\in\{A,B\}\) so that
\[
    \nu_U(W(u_U))
    =
    \max\{
        \nu_A(W(u_A)),
        \nu_B(W(u_B))
    \},
\]
and let \(V\) be the other of \(A,B\).  The selected frequency
\(p\) satisfies
\[
    p\ge\nu_A(W(u_A))\ge\frac1{N_A}=n^{-a},
\]
and \(\nu_V(W(u_V))\le p\).  Thus \(W(u_V)\) belongs to some group
\(P\).  Since
\[
    W(u_V)=-W(u_U)-W(u_C),
\]
the vertex \(u_C\) belongs to \(Z_{W(u_U),P}\), and hence to one of
its blocks.  The corresponding call therefore detects the triangle.
Conversely, every reported triangle has weight zero.

It remains to bound the running time.  Write
\[
    \widehat\omega:=\widehat\omega(a,b,c),
    \qquad
    \delta:=\widehat\omega-\omega(a,b,c)>0.
\]
For the current weight, write \(p=n^{-d}\).  Since \(p\ge n^{-a}\),
\[
    0\le d\le a\le\min\{a,b,c\}.
\]
Every call for \(x\) has part sizes
\[
    pN_U,\qquad O(pN_V),\qquad O(pN_Z),
\]
up to constant factors.

We make the exponent estimate uniform over the possibly
\(n\)-dependent value of \(d\).  Choose constants
\(\gamma,\delta_0>0\) satisfying
\[
    2\gamma+\delta_0<\delta,
\]
and let \(d'\le d\) be the largest multiple of \(\gamma\).  Pad the
three matrices to dimensions with exponents
\(a-d',b-d',c-d'\).  There are only finitely many possible values of
\(d'\), so we may fix algorithms for all of them whose exponents are
at most \(\delta_0\) above the corresponding rectangular exponents.
By \cref{lem:rectangular_scaling_inequality}, one call therefore takes at most
\begin{align*}
    n^{\omega(a-d',b-d',c-d')+\delta_0}
    \le
    n^{\omega(a,b,c)-2d'+\delta_0}
    \le
    n^{\widehat\omega-2d}
    =
    p^2n^{\widehat\omega}
\end{align*}
time.  Since there are \(O(1/p)\) calls for \(x\), their total cost is
\[
    O(pn^{\widehat\omega}).
\]
For either choice of \(U\),
\[
    \sum_{x:f_U(x)>0}\nu_U(x)
    =
    \frac{|U|}{N_U}
    \le1.
\]
Summing over all weights and both choices of \(U\) gives
\(O(n^{\widehat\omega})\) total matrix-multiplication time.

The remaining bookkeeping fits within the same bound.  Scanning the
weight classes in the other parts for every fixed weight costs
\[
    O(N_AN_B+N_AN_C+N_BN_C)
    =
    O(n^{\widehat\omega}),
\]
where we use the standard lower bound
\[
    \omega(a,b,c)
    \ge
    \max\{a+b,a+c,b+c\}.
\]
Moreover, for a fixed \(x\), the total sizes of all matrices
constructed over its \(O(1/p)\) calls are
\[
    O\bigl(
        p(N_UN_V+N_UN_Z+N_VN_Z)
    \bigr).
\]
Summing this expression over \(x\), using
\(\sum_xp\le1\), and then over the two orientations gives the same
pairwise bound.  This proves the lemma.
\end{proof}

\begin{proof}[Proof of
\cref{lem:rectangular_min_vertex_weighted_triangle}]
The reduction used in \cite[Corollary~2]{AF26} reduces
Minimum Vertex-Weighted Triangle to zero-sum vertex-weighted triangle
with \(O(\log W)\) overhead, where the absolute values of the input
weights are at most \(W\).

For an already tripartite instance, specialize this reduction to
triples in \(A\times B\times C\).  It leaves the three vertex sets fixed
and, in every generated instance, only changes vertex weights and
deletes vertices or edges.

Choose a constant
\[
    \omega(a,b,c)
    <
    \omega'
    <
    \widehat\omega(a,b,c).
\]
Apply \cref{lem:rectangular_zero_sum_triangle} with exponent
\(\omega'\) to every generated instance.  Since the weights are
polynomially bounded, the logarithmic overhead is absorbed by the
fixed gap between \(\omega'\) and
\(\widehat\omega(a,b,c)\).  A standard logarithmic-depth
self-reduction recovers a minimizing triangle within the same bound.
\end{proof}

\end{document}